%% file: FunctionEstimationOverGraphs.tex
\documentclass{sig-alternate}
\usepackage{algorithmic}
\usepackage[small,it]{caption}
\usepackage{lmodern}
\usepackage{rotating}
\usepackage{subfigure}
\usepackage{balance}
\usepackage{booktabs}
\usepackage{url}
\usepackage{ifpdf}
\usepackage[ruled]{algorithm2e}
\ifpdf 
  \usepackage[pdftex]{graphicx}
  \DeclareGraphicsExtensions{.pdf,.png,.jpg,.jpeg,.mps}
  \usepackage{pgf}
  \usepackage{tikz}
  \usetikzlibrary{mindmap,trees}
\else 
\usepackage{graphicx}
\usepackage{tikz}
\usepackage{pstricks}
\usepackage{eepic}
\usepackage{epic}
\usepackage{color}

\fi

\makeatletter
\let\STYLEFILEproof\proof
\let\STYLEFILEendproof\endproof
\let\proof\@undefined
\let\endproof\@undefined
\makeatother

\usepackage{amsmath,amsthm,amssymb}  

\let\proof\STYLEFILEproof
\let\endproof\STYLEFILEendproof

\newtheorem*{thm*}{Theorem}
\newtheorem{thm}{Theorem}[section]  
  
\newtheorem{lemma}[thm]{Lemma}  
\newtheorem*{lemma*}{Lemma}  
\newtheorem{corollary}[thm]{Corollary}  

\newcommand{\squishlist}{\begin{itemize}}

\newcommand{\squishlisttwo}{
 \begin{list}{$\bullet$}
  { \setlength{\itemsep}{0pt}
     \setlength{\parsep}{0pt}
    \setlength{\topsep}{0pt}
    \setlength{\partopsep}{0pt}
    \setlength{\leftmargin}{2em}
    \setlength{\labelwidth}{1.5em}
    \setlength{\labelsep}{0.5em} } }

\newcommand{\squishend}{\end{itemize}}

\newcommand{\given}{\,\vert\,}
\newcommand{\bP}{{\bf P}}
\newcommand{\bQ}{{\bf Q}}
\newcommand{\bA}{{\bf A}}
\newcommand{\bD}{{\bf D}}

\newcommand{\cL}{\mathcal{L}}
\newcommand{\MSE}{\mbox{MSE}}

\newcommand{\NMSE}{\mbox{\textnormal{NMSE}} }
\newcommand{\CNMSE}{\mbox{\textnormal{CNMSE}} }

\newcommand{\vol}{\mbox{\textnormal{vol}}}
\newcommand{\indeg}{\mbox{\textnormal{indeg}}}
\newcommand{\outdeg}{\mbox{\textnormal{outdeg}}}
\newcommand{\comment}[1]{}
\newcommand{\techreport}[2]{#1}

\newlength{\plotwidth}
\newlength{\plotheight}
\newlength{\aplotheight}
\newlength{\halfplotwidth}
\newlength{\halfplotheight}
\makeatletter
\let\@copyrightspace\relax
\makeatother

\begin{document}
\conferenceinfo{IMC'10,} {November 1--3, 2010, Melbourne, Australia.} 
\CopyrightYear{2010}
\crdata{} 

\title{\vspace{-10pt} Estimating and Sampling Graphs with \\ Multidimensional Random Walks}

\numberofauthors{2} 
%
\author{
\alignauthor
\vspace{-20pt}
 Bruno Ribeiro\\
       \affaddr{Computer Science Department}\\
       \affaddr{University of Massachusetts at Amherst}\\
       \affaddr{Amherst, MA, 01002}\\
       \email{ribeiro@cs.umass.edu}
\alignauthor
\vspace{-20pt}
Don Towsley\\
       \affaddr{Computer Science Department}\\
       \affaddr{University of Massachusetts at Amherst}\\
       \affaddr{Amherst, MA, 01002}\\
       \email{towsley@cs.umass.edu}
}

 \newcounter{copyrightbox}

\maketitle

\begin{abstract}
 Estimating characteristics of large graphs via sampling is a vital part of the study of complex networks. 
 Current sampling methods such as (independent) random vertex and random walks are useful but have drawbacks.
 Random vertex sampling may require too many resources (time, bandwidth, or money).
 Random walks, which normally require fewer resources {\em per sample}, can suffer from large estimation errors in the presence of disconnected or loosely connected graphs.
 In this work we propose a new $m$-dimensional random walk that uses $m$ {\em dependent} random walkers.
 We show that the proposed sampling method, which we call {\em Frontier sampling}, exhibits all of the nice sampling properties of a regular random walk.
 At the same time, our simulations over large real world graphs show that, in the presence of disconnected or loosely connected components, {\em Frontier sampling} exhibits lower estimation errors than regular random walks.
 We also show that {\em Frontier sampling} is more suitable than random vertex sampling to sample the tail of the degree distribution of the graph.
\end{abstract}

\techreport{
}{
}


\section{Introduction}\label{sec:intro}

 A number of recent studies~\cite{SurveyMeasuresGraphs,MobileFriends,MRWFacebook,Faloutsos,LeskovecCommunity, Mislove, MySpace,WillingerRDS,RDSprob} (to cite a few) are dedicated to the characterization of complex networks.
 A complex network is a network with non-trivial topological features (features that do not occur in simple networks such as lattices or random networks). Examples of such networks include the Internet, the World Wide Web, social, business, and biological networks~\cite{SurveyMeasuresGraphs,NewmanReview}.
 This work represents a complex network as a directed graph with labeled vertices and edges.
 A label can be, for instance, the degree of a vertex or, in a social network setting, someone's hometown.
 Examples of network characteristics include the degree distribution, the fraction of HIV positive individuals in a population~\cite{RDSINJECTION}, or the average number of copies of a file in a peer-to-peer (P2P) network~\cite{gkantsidis05randomwalksp2p}.
 
 Characterizing the labels of a graph requires querying vertices and/or edges; each query has an associated cost in resources (time, bandwidth, money).
 Characterizing a large graph by querying the whole graph is often too costly. 
 As a result, researchers have turned their attention to the estimation of graph characteristics based on incomplete (sampled) data.
 In this work we present a new tool, {\em Frontier Sampling}, to characterize complex networks. 
 In what follows {\em random vertex (edge) sampling} refers to sampling vertices (edges) independently and uniformly at random (with replacement).

 Distinct sampling strategies have different resource requirements depending on the network being sampled.
 For instance, in a network where each vertex is assigned a unique user-id (e.g., travelers and their passport numbers, Facebook, MySpace, Flickr, and Livejournal) it is a widespread practice to perform random vertex sampling by querying randomly generated user-ids.
 This approach can be resource-intensive if the user-id space is sparsely populated as the hit-to-miss ratio is low (e.g., less than $10\%$ of all MySpace user-ids between the highest and lowest valid user-ids are currently occupied~\cite{MySpace}).
 Another way to sample a network is by querying edges instead of vertices.
 Randomly sampling edges can be harder than randomly sampling vertices if edges are not be associated to unique IDs (or if edge IDs cannot be randomly queried).
 We summarize some drawbacks of random vertex and edge sampling:
 \squishlist
  \item Random edge sampling may be impractical when edges cannot be randomly queried (e.g., online social networks like Facebook~\cite{MRWFacebook}, MySpace~\cite{MySpace}, and Twitter or a P2P network like Bittorrent).
  \item Random vertex sampling may be undesirable when user-ids are sparsely populated (low hit-to-miss ratio) and queries are subject to resource constraints (e.g., queries are rate-limited in Flickr, Livejournal~\cite{Mislove}, and Bittorrent~\cite{Bittorrent}). In a P2P network like Bittorrent, a client can randomly sample peers (vertices) by querying a tracker (server); however, trackers may rate-limit client queries~\cite{Bittorrent}.
  \item Even when random vertex sampling is not severely resource-constrained, some characteristics may be better estimated with random edge sampling (e.g., the tail of the degree distribution of a graph).
\squishend
 An alternative, and often cheaper, way to sample a network is by means of a random walk (RW).
 A RW samples a graph by moving a particle (walker) from a vertex to a neighboring vertex (over an edge).
 By this process edges and vertices are sampled.
 The probability by which the random walker selects the next neighboring vertex determines the probability by which vertices and edges are sampled.
 In this work we are interested in random walks that sample {\em edges} uniformly.
 The edges sampled by RW can then be used to obtain unbiased estimates of a variety of graph characteristics (we present two examples in Section~\ref{sec:RW}).

 In this work we assume that a random walker has the ability to query a vertex to obtain all of its incoming and outgoing edges (Section~\ref{sec:RW} details the reason behind this assumption).
 This is possible for online networks such as Twitter, LiveJournal~\cite{Mislove}, You\-Tube~\cite{Mislove}, Facebook~\cite{MRWFacebook}, MySpace~\cite{MySpace}, P2P networks~\cite{WillingerRDS}, and the arXiv citations network.
 We revisit the theory behind random walks in Section~\ref{sec:RW}.

 Sampling graphs with random walks is not without drawbacks.
 The accuracy of the estimates depends not only on the graph structure but also on the characteristic being estimated.
 The graph structure can create distortions in the estimates by ``trapping'' the random walker inside a subgraph.
 An extreme case happens when the graph consists of two or more disconnected components (subgraphs).
 For instance, wireless mobile social networks exhibit connection graphs with multiple disconnected components~\cite{MobileFriends}.
 But even connected graphs can suffer from the same problem.
 A random walker can get ``temporarily trapped'' and spend most of its sampling budget exploring the local neighborhood near where it got ``trapped''.
 In the above scenarios estimates may be inaccurate if the characteristics of the local neighborhood differ from the overall characteristic of the graph.
 This problem is well documented (see~\cite{YuvalBook}) and our goal is to mitigate it.

 \subsection*{Contributions}
 This work proposes a new $m$-dimensional random walk sampling method ({\em Frontier sampling}) that, starting from a collection of $m$ randomly sampled vertices, preserves all of the important statistical properties of a regular random walk (e.g., vertices are visited with a probability proportional to their degree).
 While the vertices are visited with a probability proportional to their degree, we show that the joint steady state distribution of Frontier Sampling  (the joint distribution of all $m$ vertices) is closer to uniform (the starting distribution) than that of $m$ independent random walkers, for any $m > 0$.
 This property has the potential to dramatically reduce the transient of random walks.
 
 In our simulations using real world graphs we see that Frontier Sampling mitigates the large estimation errors caused by disconnected or loosely connected components that can ``trap'' a random walker and distort the estimated graph characteristic, i.e., Frontier sampling (FS) estimates have smaller Mean Squared Errors (MSEs) than estimates obtained from regular random walkers (single and multiple independent walkers, reviewed in Section~\ref{sec:MRW}) in a variety of scenarios. 

 We make two additional contributions: (1) we compare random walk-based estimates to those obtained from random vertex and random edge sampling. We show analytically that the tail of the degree distribution is better estimated using random edge sampling than random vertex sampling. We observe from simulations over real world networks (in Section~\ref{sec:FSRS}) that FS accuracy is comparable to the accuracy of random edge sampling. These results help explain recent empirical results~\cite{WillingerRDS};
 (2) we present asymptotically unbiased estimators using the edges sampled by a RW for the assortative mixing coefficient (defined in Section~\ref{sec:assortativity}) and the global clustering coefficient (defined in Section~\ref{sec:GCC}).

\subsection*{Outline}
 The outline of this work is as follows. 
 Section~\ref{sec:notation} presents the notation used in this paper.
 Section~\ref{sec:random} contrasts random vertex with random edge sampling.
 Section~\ref{sec:RW} revisits single and multiple independent random walk sampling and estimation.
 Section~\ref{sec:FS} introduces {\em Frontier Sampling} (FS), a sampling process that uses $m$ dependent random walkers in order to mitigate the high estimation errors caused by disconnected or loosely connected components. 
 Section~\ref{sec:FS} also shows that FS can be seen as an $m$-dimensional random walk over the $m$-th Cartesian power of the graph (formally defined in Section~\ref{sec:FS}).
 In Section~\ref{sec:results} we see that FS outperforms both single and multiple independent random walkers in a variety of scenarios. 
 We also compare (independent) random vertex and edge sampling with FS.
 Section~\ref{sec:literature} reviews the relevant literature.
 Finally, Section~\ref{sec:conclusions} presents our conclusions and future work.

\section{Definitions}\label{sec:notation}
 In what follows we present some definitions.
 Let $G_d = (V,E_d)$ be a labeled directed graph representing the (original) network graph, where $V$ is a set of vertices and $E_d$ is a set of ordered pairs of vertices $(u,v)$ representing a connection from $u$ to $v$ (a.k.a.\ edges).
 We assume that each vertex in $G_d$ has at least one incoming or outgoing edge.
 The in-degree of a vertex $u$ in $G_d$ is the number of distinct edges $(v_1,u),\dots,(v_i,u)$ into $u$, and its out-degree is the number of distinct edges $(u,v_1),\dots, (u,v_j)$ out of $u$. 
 Some complex networks can be modeled as undirected graphs.
 In this case, when the original graph is undirected, we model $G_d$ as a symmetric directed graph, i.e., $\forall (u,v) \in E_d, \, (v,u) \in E_d$. 

 Let $\cL_v$ and $\cL_e$ be a finite set of vertex and edge labels, respectively.
 Each edge $(u,v) \in E_d$ is associated with a set of labels $\cL_e(u,v) \subseteq \cL_e$.
 For instance, the label of edge $(u,v)$ can be the in-degree of $v$ in $G_d$.
 Similarly, we can associate a set of labels to each vertex, $\cL_v(v) \subseteq \cL_v,\, \forall v \in V$.
 Some edges and vertices may not have labels. If edge $(u,v)$  is unlabeled then $\cL_e(u,v) = \emptyset$. 
 Similarly, if vertex $v$  is unlabeled then $\cL_v(v) = \emptyset$.
 
 When performing a random walk, we assume that a random walker has the ability to retrieve incoming and outgoing edges from a queried vertex (and vertices are distinguishable). With this assumption we are able to build (on-the-fly) a symmetric directed graph while walking over $G_d$.
 Let $G=(V,E)$ be the symmetric counterpart of $G_d$, i.e., $$E = \bigcup_{\forall (u,v) \in E_d} \{(u,v),(v,u)\}.$$
 Note that $G$ may not be connected.
 As $G$ is symmetric, we denote by $\deg(v)$ to be the in-degree or the out-degree of $v \in V$ as they are equal.
 Let $\vol(S) = \sum_{\forall v \in S} \deg(v), \, \forall S \subseteq V$, denote the volume of the vertices in $S$.

 Let $\hat{\theta}_l$ be the estimated fraction of {\em vertices} with label $l$ obtained by some estimator.
 The two error metrics used in most of our examples are the normalized root mean square error of $\hat{\theta}_l$, which is a normalized measure of the dispersion of the estimates, defined as 
\begin{equation} \label{eq:NMSE}
 \NMSE(l) = \frac{\sqrt{E[( \hat{\theta_l} - \theta_l )^2]}}{\theta_l} \, .
\end{equation}
 and the normalized root mean square error of the Complementary Cumulative Distribution Function (CCDF) $\gamma = \{\gamma_l\}$, where $\gamma_l = \sum_{k=l+1}^\infty \theta_k$, defined as 
\begin{equation} \label{eq:CNMSE}
 \CNMSE(l) = \frac{\sqrt{E[( \hat{\gamma_l} - \gamma_l )^2]}}{\gamma_l} \, .
\end{equation}
 For the sake of simplicity, and unless stated otherwise, in the remainder of this paper we assume that all queries of edges and vertices have unitary cost and that we have a fixed sampling budget $B$.

\section{Vertex v.s. edge sampling}\label{sec:random}
 We consider a straightforward estimation problem to illustrate a tradeoff between random edge and random vertex sampling.
 Consider the problem of estimating the out-degree distribution of $G_d$.
 Let $\theta_i$ be the fraction of vertices with out-degree $i > 0$ and $d$ be the average out-degree.
 Let the label of vertex $u$, $\cL_v(u)$, be the out-degree of $u$.
 We assume that $d$ is known; also assume that from an edge $(u,v)$ we can query $\cL_v(u)$.
 In random edge sampling the probability of sampling a vertex with out-degree $i$ is proportional $i$: $\pi_i = i \, \theta_i / d$. 
 On the other hand, random vertex sampling samples a vertex with out-degree $i$ with probability $\theta_i$.
 A straightforward calculation shows that the \NMSE (equation~(\ref{eq:NMSE})) of $B$ randomly sampled {\em edges} with out-degree $i$ is 
\begin{equation}\label{eq:NMSEE}
    \NMSE(i) = \sqrt{{(1/\pi_i - 1)}/{B}} \, ,\quad i>0  .
\end{equation}
 Similarly, the $\NMSE(i)$ for random vertex sampling is 
\begin{equation}\label{eq:NMSEV}
    \NMSE(i) = \sqrt{{(1/\theta_i - 1)}/{B}} \, .
\end{equation}
 Now note that $\pi_i / \theta_i = i / d$, which means that $\pi_i > \theta_i$ if $i > d$ and $\pi_i < \theta_i$ if $i < d$.
 From equations~(\ref{eq:NMSEE}) and~(\ref{eq:NMSEV}) we see that random edge sampling more accurately estimates degrees larger than the average ($i > d$) while random vertex sampling more accurately estimates degrees smaller than the average ($i < d$).
 This means random edge sampling exhibits smaller \NMSE when estimating the tail of the out-degree distribution.

 Above we have seen that random edge sampling is more accurate than random vertex sampling in estimating the tail of the out-degree degree distribution.
 A similar result happens with the in-degree distribution and the degree distribution of undirected networks.
 The above analysis explains real world experiments~\cite{WillingerRDS}.
 Unfortunately, as discussed in Section~\ref{sec:intro}, random edge sampling is rarely practical.
 In what follows we see that, if $G$ is connected, random walks exhibit similar statistical properties to random edge sampling.

\section{Random walk sampling} \label{sec:RW}
 In this section we review random walk (RW) sampling and estimation over a non-bipartite, connected, directed, symmetric graph $G$.
 Sampling $G$ with a RW is straightforward. 
 The random walker has a sampling budget $B$ and starts at vertex $v_0 \in V$. 
 For the sake of simplicity, unless stated otherwise, we consider that all queries to vertices have unit cost and that we have a fixed sampling budget $B$.

 Let $\{(u_i,v_i)\}_{i=1}^B$ be the a sequence of edges sampled by a RW, where $u_i = v_{i-1}, \, i=2,\dots,B$.
 Note that edges may be sampled multiple times.
 We refer to $(u_i,v_i) $ as the $i$-th sampled edge.
 At the $i$-th step a walker at vertex $v_i$ chooses an outgoing edge $(v_i,u_i)$ uniformly at random from the set of outgoing edges of $v_i$ and adds $(v_i,u_i)$ to the sequence of sampled edges. 
 At step $i+1$ the random walker starts at vertex $u_i$ and the sampling continues until $i = B$.
 
 The RW described here is the most common type of RW found in the literature~\cite{LovaszSurvey}.
 Other types of random walks differ in the way in which outgoing edges are sampled.
 The Metropolis-Hastings RW~\cite{MRW2} is an example of a random walk that samples {\em vertices} (not edges) uniformly at random.
 However, experiments estimating a variety of metrics indicate that Metropolis-Hastings RW is less accurate than the random walk described in this work~\cite{MRWFacebook,WillingerRDS}.
 For more details about other types of RW please refer to~\cite[Chapter~7]{MCMC}.

 An important property of a RW is its ability to reach a unique stationary regime.
 A necessary condition for stationarity is that $G$ must be symmetric, connected, and non-bipartite (the non-bipartite assumption can be relaxed in a lazy random walk~\cite{LovaszSurvey}).
 In a stationary RW, the sequence of sampled edges is a stationary sequence.
 A sequence $X_1,X_2,\dots$ of random variables is said to be stationary if for any positive integers $n$ and $k$, the joint distribution of $(X_n,\dots , X_{n+k})$ is independent of $n$.
 Once the RW reaches steady state, it also shares two important properties with random edge (RE) sampling.
 First, both RW and RE sample edges uniformly at random~\cite{LovaszSurvey}, which means that the probability that a vertex $v$ is sampled is $\deg(v) / \vol(V).$
 Second, both RW and RE obey the strong law of large numbers, as we see next.

\subsection{Strong Law of Large Numbers}\label{sec:SLLN}
 The following variation of the strong law of large numbers is a powerful tool to build (asymptotically) unbiased estimators of graph characteristics.
 We provide a trivial extension of a well known result~\cite[Theorem~17.2.1]{MeynTweedie} to the case where we are interested in a subset of the graph edges.
 Let $E^\star \subseteq E$ be non-empty.
 Let $(u_i,v_i)$ be the $i$-th RW sampled edge such that $(u_i,v_i) \in E^\star$; and let $B^\star(B)$ be the number of such samples, where $B$ is the number of RW steps.
 $B^\star(B)$ is a random variable that represents the number of RW sampled edges that belong to $E^\star$.
 Note that $B^\star(B) \leq B$.

%

%
\begin{thm}[SLLN] \label{thm:SLLN}
 For any function $f$, where \\ $\sum_{(u,v) \in  E^\star} \vert f(u,v) \vert < \infty$,
 \[
     \lim_{B \to \infty} \frac{1}{B^\star(B)} \sum_{i=1}^{B^\star(B)} f(u_i,v_i) \to \frac{1}{\vert E^\star \vert} \sum_{\forall (u,v) \in E^\star} f(u,v)  \, ,
 \]
almost surely, i.e., the event occurs with probability one.
 \end{thm}
\begin{proof}
  Let
$$h(u,v)  = \begin{cases}
        1 & \text{if } (u,v) \in E^\star \, , \text { and }\\
        0 & \text{otherwise.}
       \end{cases}
 $$
 As the RW is stationary (and edges are sampled uniformly)
\begin{align*}
   \lim_{B \to \infty}  \frac{\sum_{i=1}^{B} f(u_i,v_i) h(u_i,v_i)}{\sum_{i=1}^{B} h(u_i,v_i)}  \to  \frac{\sum_{\forall (u,v) \in E} f(u,v) h(u,v)}{\sum_{\forall (u,v) \in E} h(u,v)} \, , 
\end{align*}
 almost surely~\cite[Theorem~17.2.1]{MeynTweedie}.
 The proof follows from noting that $B^\star(B) = \sum_{i = 1}^B h(u_i,v_i)$ and that $h(u,v) = 0, \forall (u,v) \in E \backslash E^\star$.
\end{proof}
 Theorem~\ref{thm:SLLN} allows us to construct estimators of graph characteristics that converge to their true values as the number of RW samples goes to infinity ($B \to \infty$). 
 If we are trying to estimate vertex labels we set $E^\star = E$ and $B^\star = B$.
 In what follows we apply Theorem~\ref{thm:SLLN} to estimate graph characteristics; we also present four examples of estimators.

\subsection{Estimators}\label{sec:estimators}

 An {\em estimator} is a function that takes a sequence of observations (sampled data) as input and outputs an estimate of a unknown population parameter (graph characteristic).
 In this section we see how we can estimate graph characteristics using the edges sampled by a RW.
 
 We present estimators of the following four graph characteristics: the edge label density (the fraction of edges with a given label in the graph), the assortative mixing coefficient~\cite{NewmanAssortativity}, the vertex label density, and the global clustering coefficient~\cite{GCC}. 
 Designing these estimators is straightforward:
\begin{enumerate}
 \item[(1)] First we find a function $f$ that computes the characteristic of $G$ using $E$; 
 \item[(2)] then we replace $E$ with the sequence of edges sampled by a stationary RW.
\end{enumerate}
 In what follows we illustrate how to build an estimator of the edge label density.

\subsubsection{Edge Label Density} \label{sec:edgelabel}
 We seek to estimate the fraction of {\em edges} with label $l \in \cL_e$ in $G_d$ among all edges $(u,v)$ that have labels, i.e., $\cL_e(u,v) \neq \emptyset$.
 Edge labels can be anything, from social networking labels to the amount of IP traffic over each link in a computer network.
 An edge label can be, for instance, a tuple $(\outdeg(u),\indeg(v))$ where $\outdeg(u)$ is the out-degree of $u$ and $\indeg(v)$ is the in-degree of $v$ in the original graph $G_d$.

 
 For now we assume that we know $E$.
 Let $E^\star $ be the non-empty subset of $E$ for which there are labels.
 Let $p_l$ denote the fraction of edges in $E^\star$ with label $l$; it is clear that
\begin{equation*}
 p_l = \sum_{\forall (u,v) \in E^\star} \frac{{\bf 1}(l \in \cL_e(u,v))}{\vert E^\star \vert} \, ,
\end{equation*} 
where 
$$
{\bf 1}(l \in \cL_e(u,v)) = \begin{cases}
                           1 & \text{if }l \in \cL_e(u,v) \, , \\
                           0 & \text{otherwise}.
                          \end{cases}
$$ 
 Let $B^\star(B)$ be the number of RW sampled edges that belong to $E^\star$ and $(u_i,v_i)$ be the $i$-th of such edges.
 Replacing $E^\star$ with the edges in $E^\star$ sampled by a stationary RW gives the following estimator 
\begin{equation} \label{eq:edgelabel}
 \hat{p}_l \equiv \sum_{i=1}^{{B^\star(B)}}  \frac{{\bf 1}(l \in \cL_e(u_i,v_i))}{{B^\star(B)}} \, .
\end{equation} 
 It follows directly from Theorem~\ref{thm:SLLN} (with $f(u,v) = {\bf 1}(l \in \cL_e(u,v))$) that $\lim_{{B} \to \infty } \hat{p}_l \stackrel{\mbox{a.s.}}{\to}  p_l$.
 Moreover, from the linearity of expectation, $E[\hat{p}_l] = p_l$ for all values of ${B^\star(B)} > 0$. 
 

\subsubsection{Assortative Mixing Coefficient}\label{sec:assortativity}
 The assortative mixing coefficient~\cite{NewmanAssortativity} is a measure of the correlation of labels between two neighboring vertices.
 By appropriately assigning edge labels derived from vertex labels, we can use the density estimator of equation~(\ref{eq:edgelabel}) to derive an estimator of the assortative mixing coefficient.
 In order to simplify our exposition, we restrict our analysis to the assortative mixing of vertex degrees in a directed graph (equation~(25) of~\cite{NewmanAssortativity}).
 It is trivial to extend our analysis to other types of assortative mixing coefficients, e.g.,  equations~(21) and~(23) of~\cite{NewmanAssortativity}.

 Let ($\outdeg(u),\indeg(v)$) denote the label of a directed edge $(u,v)$ in $G$ that also exists in $G_d$; 
 and let $E^\star$ be the set of all such edges ($E^\star = E_d$).
 Let $p_{ij}$ denote the fraction of labeled edges with label $(i,j)$.
 Let $W_\text{out}$ ($W_\text{in}$) denote the maximum {\em observed} out-degree (in-degree) of $G_d$ in the RW.
 The degree assortative mixing coefficient~\cite{NewmanAssortativity} of a directed graph can be estimated using
 \[
  \hat{r} \equiv \frac{1}{\hat{\sigma}_{\text{in}} \, \hat{\sigma}_{\text{out}}} \sum_{i = 0}^{W_\text{out}} \sum_{j=0}^{W_\text{in}} i j ( \hat{p}_{ij} -  \hat{q}^{\text{out}}_i \hat{q}^{\text{in}}_j) \, ,
 \]
where
\begin{equation*} 
 \hat{p}_{ij} \equiv \sum_{k=1}^{{B^\star(B)}}  \frac{{\bf 1}(\outdeg(u_k)=i \, ,\, \indeg(v_k)=j)}{{B^\star(B)}} \, ;
\end{equation*} 
 \begin{align*}
  & \hat{q}^{\text{out}}_i \equiv \sum_{k = 0}^{W_\text{in}} \hat{p}_{ik} \quad \text{;} \quad \hat{q}^{\text{in}}_j \equiv \sum_{k = 0 }^{W_\text{out}} \hat{p}_{kj} \quad ; \\
 & \hat{\sigma}_{\text{in}} = \sqrt{ \sum_{i=0}^{W_\text{out}} j^2 \hat{q}^{\text{in}}_j - \left(\sum_{i=0}^{W_\text{out}} j \hat{q}^{\text{in}}_j \right)^2 } \quad \text{ ; \, and }  \\ 
&	\hat{\sigma}_{\text{out}} = \sqrt{ \sum_{i=0}^{W_\text{in}} i^2 \hat{q}^{\text{in}}_i - \left(\sum_{i=0}^{W_\text{in}} i \hat{q}^{\text{out}}_i \right)^2 }
\end{align*}
 where $\hat{\sigma}_{\text{in}}$ and $\hat{\sigma}_{\text{out}}$ are the standard deviation of the distribution $\hat{q}^{\text{in}}_i$ $ (\hat{q}^{\text{out}}_j)$.
 As the estimate $\hat{p}_{ij}$ (equation~(\ref{eq:edgelabel})) asymptotically converges almost surely to its true value, it is trivial to show that $\hat{q}^{\text{in}}_i$, $\hat{q}^{\text{out}}_j$, $\hat{\sigma}_{\text{in}}$, and $\hat{\sigma}_{\text{out}}$ also asymptotically converge almost surely to their true values.  
 Thus, $\hat{r}$ asymptotically converges, almost surely, to the true assortative mixing coefficient of~\cite{NewmanAssortativity}, as long as $\sigma_{\text{in}} > 0$ and $\sigma_{\text{out}} > 0$.
 This implies that $\hat{r}$ is an asymptotically unbiased estimator of the assortative mixing coefficient of $G_d$.

\techreport{}{
\hfill\eject
}
\subsubsection{Vertex Label Density}\label{sec:labeldensity}
 Let $\cL_v(v)$ be the set of labels associated with vertex $v,\, \forall v \in V$.
%
  The fraction of {\em vertices} with label $l$ in $G$, $\theta_l$, is 
\begin{equation} \label{eq:vdg}
\theta_l = \frac{1}{\vert V \vert} \sum_{\forall (u,v) \in E}\frac{{\bf 1}(l \in \cL_v(v))}{\deg(v)} \, ,
\end{equation}
as $G=(V,E)$ is directed and symmetric.
By replacing $E$ with a sequence of edges sampled by a stationary RW (here we have $E^\star = E$ and $B^\star = B$) and renormalizing, we arrive at the following estimator for $\theta_l$ 
\begin{equation} \label{eq:hattheta}
\hat{\theta}_l \equiv \frac{1}{S \, B} \sum_{i=1}^B \frac{{\bf 1}(l \in \cL_v(v_i))}{\deg(v_i)} \: ,
\end{equation}
where 
$
   S = 1/B \sum_{i = 1}^B 1/\deg(v_i)  \, .
 $
 From Theorem~\ref{thm:SLLN} we have $\lim_{B \to \infty} S \to  \vert V \vert/\vert E \vert $, almost surely.
 Using again Theorem~\ref{thm:SLLN} we have 
\begin{align*} 
  \lim_{B \to \infty}  \frac{1}{B} \sum_{i=1}^B \frac{{\bf 1}(l \in \cL_v(v_i))}{\deg(v_i)} {\to}  
           \frac{1}{\vert E \vert} \sum_{\forall (u,v) \in E} \frac{{\bf 1}(l \in \cL_v(v))}{\deg(v)} \, ,
\end{align*}
 almost surely, which divided by $\vert V \vert/\vert E \vert $ yields equation~(\ref{eq:vdg}).
 As $S$ converges almost surely to $\vert V \vert/\vert E \vert $, we have $\lim_{B \to \infty}  \hat{\theta}_l  \to  \theta_l$, almost surely.
 This also implies that $\hat{\theta}_l$ is an asymptotically unbiased estimator of $\theta_l$.

\subsubsection{Global Clustering Coefficient}\label{sec:GCC}
In the literature the term {\em clustering coefficient} often refers to the local clustering coefficient~\cite{Strogatz}.
In our example we estimate a different metric: the {\em global} clustering coefficient.
In a social network the global clustering coefficient, $C$, is the probability that the friend of John's friend
is also John's friend~\cite{GCC}.
Let $V^\star$ be the set of vertices $v\in V$ with $\deg(v) > 1$.
The global clustering coefficient of an undirected graph is defined as~\cite{GCC}
\begin{equation}\label{eq:C}
   C \equiv \frac{1}{\vert V^\star \vert} \sum_{\forall v \in V} c(v) \, ,
\end{equation}
where
$$
 c(v) = \begin{cases}
            \Delta(v) / \binom{\deg(v)}{2} & \text{if }\deg(v)\geq 2 \\
             0 & \text{otherwise} \, ,
        \end{cases}
$$ 
 where $\Delta(v) = \vert \{(u,w) \in E \, : \, (v,u) \in E \text{ and }(v,w) \in E\}  \vert$ is the number of triangles that contain vertex $v$ and $\binom{\deg(v)}{2}$ is the maximum number of triangles that a vertex $v$ with degree $\deg(v)$ can belong to.

 Note that finding $\Delta(v)$ for a given vertex $v\in V$ requires knowing all vertices within two hops of $v$, which can be a resource intensive task.
 To avoid the cost of computing $\Delta$, we rewrite equation~(\ref{eq:C})
\begin{equation*}
   C = \frac{1}{\vert V^\star \vert} \sum_{\forall (v,u) \in E} \frac{f(v,u)}{\binom{\deg(v)}{2}} \, ,
\end{equation*}
where $f(v,u)$ gives the number of shared neighbors between $u$ and $v$.
 
 Let $(v_i,u_i)$ be the $i$-th sampled edge in a stationary RW
 and let 
 \[
  \hat{C} \equiv  \frac{1}{S \, B} \sum_{i=1}^B \frac{f(v_i,u_i)}{\binom{\deg(v_i)}{2}}\frac{1}{\deg(v_i)} \, ,
 \]
where 
$$
   S = \frac{1}{B} \sum_{i = 1}^B \frac{1}{\deg(v_i)} .
 $$
 \begin{corollary}
$
  \lim_{B \to \infty} \hat{C} {\to} C \, ,
$
 almost surely.
 \end{corollary}
\begin{proof}
  From Theorem~\ref{thm:SLLN} 
$$
\lim_{B \to \infty} S {\to}  {\vert V^\star \vert}/{\vert E \vert },
$$
almost surely. 
Also from Theorem~\ref{thm:SLLN} 
$$
 \lim_{B \to \infty} \frac{1}{B} \sum_{i=1}^B \frac{f(v_i,u_i)}{\binom{\deg(v_i)}{2}}\frac{1}{\deg(v_i)} {\to} \frac{1}{\vert E \vert} \sum_{\forall (v,u) \in E} \frac{f(v,u)}{\binom{\deg(v)}{2}},
$$
 almost surely, which together with the almost sure convergence of $S$ implies that $\lim_{B \to \infty} \hat{C} {\to} C$, almost surely.
\end{proof}
 Note that almost sure convergence implies that $\hat{C}$ is an asymptotically unbiased estimator of $C$.
 
%
%

\subsection{Estimator Accuracy \& Graph Structure}\label{sec:AccStru}
 Sampling a graph using a RW is not without drawbacks.
 A random walker can get (temporarily) ``trapped'' inside a subgraph whose characteristics differ from those of the whole graph.
 Even if the random walker starts in steady state (i.e., is stationary), this scenario may increase the mean squared error of the estimates.
 If the random walker does not start in steady state, this scenario may cause an increase in the estimation bias as well as the mean squared error.
 Ideally, the random walker needs to mitigate the effect of these traps on the estimates.

 The above two types of estimation errors are well documented in the literature and various solutions are available~\cite{GeyerPracticalMCMC}. 
 For instance, if the random walker does not start in a stationary regime (transient), it is common practice to discard the first $w$ samples~\cite{GeyerPracticalMCMC}.
 The value of $w$ is called the {\em burn-in period}.
 There are two problems with this solution: (1) it only reduces the error related to the non-stationarity of the samples; (2) it is difficult to determine a good value for $w$ if the sampling budget is small (compared to the size of the graph) and the size and structure of $G$ are unknown.
 
 A simple naive solution to the RW ``trapping'' problem (adopted in~\cite{MRWFacebook} to sample Facebook), is to sample the graph using multiple independent random walkers~\cite{GeyerPracticalMCMC}. 
 In what follows we see that this naive approach can lead to increased estimation errors.
 In Section~\ref{sec:frontier} we propose a method to mitigate the random walk ``trapping'' problem using $m$ {\em dependent} random walkers.

\subsection{Multiple Independent Random Walkers}\label{sec:MRW}
 The main problem of estimating graph characteristics using a single walker is that the walker may get trapped inside a local neighborhood.
 But there is the question of what happens if we could start $m$ independent random walkers ({\em MultipleRW}) at $m$ independently sampled vertices in the graph.
 Note that when $m=1$ we are back to sampling $G$ using a single random walker, which we denote as {\em SingleRW}. 
 Networks such as MySpace, Facebook, and Bittorrent admit random (uniform) vertex sampling at cost $c$ higher than the cost of sampling the neighbors of a known vertex (which is what a RW does).
 In such networks random vertex sampling may help us start $m$ random walkers at different parts of the graph.
 While the value of $c$ can be large, initializing $m$ random walkers with uniformly sampled vertices costs (only) $m c$ units of our sampling budget (where one unity of the budget is the cost of sampling a vertex in a RW).

	\begin{figure}[tbht]
	\begin{center}
         \begin{minipage}{\plotwidth}
  			   \vspace{-0pt}
				\resizebox{\plotwidth}{0.95\plotheight}{
				\hspace{-20pt}
      			   \input{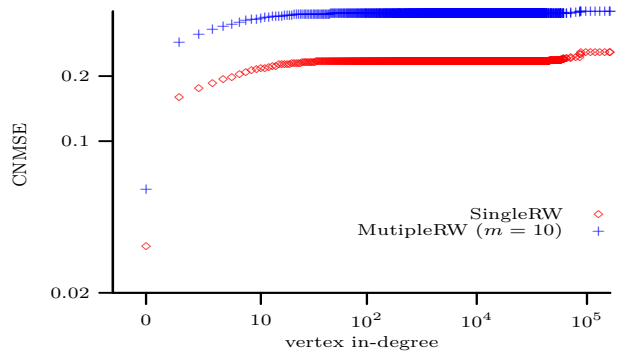}
			   }
			\caption{	(Flickr) The log-log plot of the \CNMSE of the in-degree distribution estimates with budget $B = \vert V \vert / 10$.}
			\label{fig:LCCFlickrMSRWMSE}
		\end{minipage}
	\end{center}
			\vspace{-15pt}
	\end{figure}

 Unfortunately, $m$ independent random walkers starting at $m$ randomly sampled vertices may decrease estimation accuracy.
 Consider the following experiment where each of the $m$ random walkers (independently) performs $\lfloor B/m - c \rfloor$ steps.
 We seek to estimate the CCDF (complementary cumulative distribution function) of the in-degree of the Flickr graph (the Flickr dataset is summarized in Table~\ref{tab:graphs}).
 According to Table~\ref{tab:graphs} the Flickr graph is disconnected.
 The goal of this simulation is to compare the estimation accuracy of SingleRW and MultipleRW when there are no disconnected components.
 For this we set $c = 1$.
 The sampling budget is $B = 171,525 =\vert V \vert / 10$, which amounts to a sampling budget equivalent to $10\%$ of the vertices in the graph. 
 Figure~\ref{fig:LCCFlickrMSRWMSE} shows a log-log plot of the CNMSE, equation~(\ref{eq:CNMSE}), of SingleRW and MultipleRW ($m=10$) averaged over $10,000$ runs.
 Note that the estimates obtained by SingleRW are, on average, more accurate than the estimates obtained by MultipleRW.
 Increasing the sampling budget $B$ does not reduce the gap.
 In Section~\ref{sec:results} we see, over other real-world graphs, that when starting random walkers from uniformly sampled vertices, MultipleRW has higher estimation errors than SingleRW.

\subsection{Disconnected Graph Example}\label{sec:example}
 The following example shows a situation in which both MultipleRW and SingleRW have large estimation errors.
 In this example we initialize MultipleRW with $m$ randomly (uniformly) sampled vertices.
 We simplify our exposition by assuming that each MultipleRW walker takes $B/m$ steps, where $B$ (the sampling budget) is a multiple of $m$. 
 Let $G=(V,E)$ be an undirected graph that has two large disconnected components $G_A = (V_A,E_A)$ and $G_B = (V_B,E_B)$. 
 Let $\vert  V_A \vert = \vert  V_B \vert$ and $\vol(V_A) > \vol(V_B)$.
 When initial vertices are uniformly sampled, the probability that each MultipleRW walker (independently) starts in $G_A$ ($G_B$) is $h_A = \vert V_A \vert / \vert V \vert $ ($h_B = \vert V_B \vert / \vert V \vert $).
 Recall that $G_A$ and $G_B$ are disconnected.
 For each random walker, after $B/m$ ($B \gg 1$) RW steps, an edge $(u_A,v_A) \in E_A$ is sampled with probability $p_A \approx h_A/\vol(V_A)$. Similarly an edge $(u_B,v_B) \in E_B$ is sampled with probability $p_B \approx h_B/\vol(V_B)$.
 Thus $p_A < p_B$, i.e., the edges in $G_B$ are sampled with higher probability than the edges in $G_A$.
 As our estimators assume that all edges are sampled with the same probability, this imbalance between $p_A$ and $p_B$ has the potential to introduce large MSEs (and biases).
 Note that increasing $m$ does not change $p_A$ and $p_B$. 
 Increasing $B$ only mitigates this problem if $G$ is connected and, in a loosely connected graph, only large values of $B$ positively impact the MSE.
 Ideally we want a RW algorithm that does not rely on large sampling budgets $B$ to achieve low estimation errors.

 Now consider the same thought experiment where each random walker starts in $G_A$ and $G_B$ (independently) with probabilities $h_A = \vol(V_A) /\vol(V)$ and $h_B = \vol(V_B) / \vol(V)$, respectively.
 In this new scenario it is easy to see that $p_A = p_B = 1/\vol(V) = 1/\vol(V)$. 
 Thus, we would like to start a RW at vertex $v$ with probability $\deg(v)/\vol(V)$, $\forall v \in V$.
 Section~\ref{sec:results} shows that in practice this approach can successfully mitigate estimation errors caused by disconnected components.
 Unfortunately, it is difficult to sample $m$ mutually {\em independent} vertices with probabilities proportional to their degrees.
 In the case where $G$ is connected, this is equivalent to jointly start $m$ independent random walkers in steady state.
 In networks such as MySpace, Facebook, and Bittorrent it is unclear how one can (efficiently) sample vertices with probabilities proportional to their degrees.

{ \em
 We want an $m$-dimensional random walk that, in steady state, samples {\bf edges uniformly} at random {\bf but}, unlike MultipleRW, can benefit from starting its walkers at {\bf uniformly sampled  vertices}.
}

\section{Frontier Sampling (FS)}\label{sec:FS}
In this section we present a new and promising approach to an $m$-dimensional random walk that benefits from starting its walkers at uniformly sampled  vertices.
{\em Frontier Sampling} (FS) performs $m$ {\em dependent} random walks in the graph.
 We refer to $m$ as the dimension of the FS random walk.
 Let $c$ be the cost of randomly sampling a vertex.
The FS algorithm, given in Algorithm~\ref{alg:FS} is a centrally coordinated sampling algorithm that maintains a list of $m$ vertices representing $m$ random walkers. 
\begin{algorithm}
\begin{algorithmic}[1]
\STATE $n \gets 0$ \COMMENT{$n$ is the number of steps}
\STATE Initialize $L=(v_1,\dots,v_m)$ with $m$ randomly chosen vertices (uniformly)
\REPEAT
\STATE Select $u \in L$ with probability $\deg(u)/\sum_{\forall v \in L} \deg(v)$  \label{FSloop}
\STATE Select an outgoing edge of $u$, $(u,v)$, uniformly at random\label{FSselect}
\STATE Replace $u$ by $v$ in $L$ and add $(u,v)$ to sequence of sampled edges
\STATE $n \leftarrow n + 1$
\UNTIL $n \geq B - m c$ 
\end{algorithmic}
\caption{Frontier Sampling (FS).\label{alg:FS}}
\end{algorithm}
 This way FS is less likely to get stuck in loosely connected components than a single random walker.
 However, in Section~\ref{sec:FSvardist} we see that the joint steady state distribution of FS is much closer to the uniform distribution than is the steady state distribution of $m$ independent random walkers.
 Section~\ref{sec:distributedFS} describes how the FS algorithm can be made fully distributed.
 In Section~\ref{sec:results} we see that, if the initial set of random walk vertices is chosen uniformly at random, FS estimates are more accurate than both single and $m$ independent random walkers.
 
\subsubsection*{Frontier Sampling: An m-dimensional Random Walk} \label{sec:frontier}
 FS shares many of the same statistical properties of a single random walker.
 The key insight behind Theorem~\ref{thm:FST} below is that the FS stochastic process is equivalent to the stochastic process of a single random walker over the $m$-th Cartesian power of $G$, $G^m = (V^m,E_m)$, where 
\[
V^m= \{(v_1,\dots,v_m) \given v_1 \in V \wedge \dots \wedge v_m \in V\}
\]
is the $m$-th Cartesian power of $V$ and $\forall {\bf v},{\bf u}\in V^m$, $({\bf v},{\bf u}) \in E_m\,$ if exists an index $i$ such that $(v_i,u_i) \in E$ and $u_j = v_j$ for $j \neq i$.

\begin{figure}[htb]
\begin{center}
\def\JPicScale{0.8}
\input{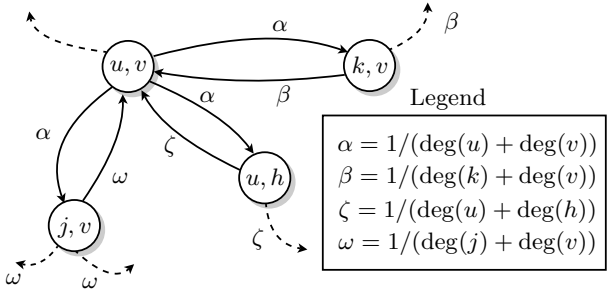}
\caption{Illustration of the Markov chain associated to the Frontier sampler with dimension $m=2$.\label{fig:exFS}}
\vspace{-18pt}
\end{center}
\end{figure}

\begin{lemma}\label{lem:FSGM}
 The Frontier sampling process is equivalent to the sampling process of a single random walker over $G^m$.
\end{lemma}
\begin{proof}
 Consider the $(n-1)$-st step of FS. 
 The reader may find Figure~\ref{fig:exFS} helpful in following the proof.
 Let $L_n=(v_1,\dots,v_m)$ be the state of FS before the $n$-th step.
 Clearly $L_n \in V^m$.
 Let $e(L_n)$ denote the collection of all edges associated to the vertices in $L_n$. 
 We refer to $e(L_n)$ as the edge frontier at the $n$-th step.
 We describe the transition from state $L_n$ to state $L_{n+1}$ as follows
(lines~(\ref{FSloop}) and~(\ref{FSselect}) of the FS algorithm):
Select a vertex $v \in L_n$ with probability proportional to $\deg(v)$ and 
then replace vertex $v$ in $L_n$ with one of its neighbors (selected uniformly at random).
This is equivalent to randomly sampling an edge from $e(L_n)$ with probability 
\[
 p = \frac{1}{\vert e(L_n) \vert} = \frac{1}{\sum_{\forall v \in L_n} \deg(v)}.
\]
Therefore, $L_n$ transits to state $L_{n+1}$ iff $(L_n,L_{n+1}) \in E_m$ and the transition probability from $L_n$ to $L_{n+1}$ is $1/\vert e(L_n) \vert$.
Thus, the Markov chain that describes FS is equivalent to the Markov chain of a single random walker over $G^m$.
\end{proof}

\begin{thm} \label{thm:FST}
 Recall that $G$ is a directed symmetric graph. 
 If $G$ is connected and non-bipartite, then in steady state FS has the following properties: 
\renewcommand{\labelenumi}{(\Roman{enumi})}
\begin{enumerate}
 \item edges are sampled uniformly at random and form a stationary sequence, 
 \item has distribution, $L_\infty=(v_1,\dots,v_m)$, equal to $$\frac{\sum_{i = 1}^m \deg(v_i)}{ m \vert V \vert^{m-1} \vol(V) } \, ,\text{ ,}$$ which is unique, and
 \item the sequence of sampled edges satisfies the Strong Law of Large Numbers (Theorem~\ref{thm:SLLN}).
\end{enumerate}
\end{thm}
\begin{proof}
 \techreport{The proof is found in Appendix~\ref{appx:proofFST}.}{The proof is found in our technical report~\cite{TechReport}.}
\end{proof}
In Section~\ref{sec:example} we observed that, when starting multiple RWs, the MSE is reduced when the number of walkers inside each subgraph matches the number obtained when the graph is connected and all walkers are in steady state.
In what follows we see that, in steady state, the average number of MultipleRW walkers in $V_A$ is far from the average number obtained with $m$ uniformly sampled vertices.
In contrast, Section~\ref{sec:FSvardist} shows that as $m \to \infty$, by uniformly sampling the starting vertices, FS starts in steady state with respect to the number of random walkers in any subset of vertices $V_A \subseteq V$.

	\begin{table*}
		\begin{center}
       { 
		\begin{tabular}{@{}lrrrr@{}} \toprule
			Graph & Flickr & LiveJournal & YouTube& Internet RLT \\ \midrule
			Description & Social Net. & Social Net.  & Social Net. & Internet tracert. \\
			Type of graph  & Directed & Directed  & Directed & Directed \\
			\# of Vertices & $1,715,255$ & $5,204,176$  & $1,138,499$  & $192,244$  \\
			Size of LCC & $1,624,992$ & $5,189,809$ & $1,134,890$   & $609,066$  \\
			\# of Edges & $22,613,981$ & $77,402,652$ & $9,890,764$ & $609,066$  \\
			Average Degree & $12.2$ & $14.6$ & $8.7$  & $3.2$  \\
            $w_\text{max}$ & $2232$ & $1029$ & $3305$ & $335$ \\
			$\%$ of Original Graph & $26.9\%$ & $95.4\%$ &  NA  & NA \\ \bottomrule
		\end{tabular}
        }
		\caption{Summary of the graph datasets used in our simulations. ``Size of LCC'' refers to the size of the largest connected component and $w_\text{max}$ is the value of the largest vertex degree divided by the average degree.}\label{tab:graphs}
		\end{center}
\vspace{-20pt}
	\end{table*}

\subsection{MultipleRW Steady State v.s.\ Uniform Distribution} \label{sec:MultipleRWvardist}
 Consider a MultipleRW process with $m$ walkers and let $K_{mw}(m)$ be a random variable that denotes the steady state number of MultipleRW random walkers in $V_A$.
 Let $\alpha_A = E[K_{mw}(m)]/E[K_{un}(m)]$ be the ratio between the steady state number of MultipleRW in $V_A$ and the number of random walkers that start in $V_A$ from uniformly sampled vertices.
 As all random walkers are independent we have 
 \[
     E[K_{mw}(m)] = \frac{m \, \vert V_A \vert d_A }{\vert V \vert d} \, .
 \]
 It is also easy to see that 
 \[
    E[K_{un}(m)] = \frac{m \, \vert V_A \vert}{\vert V \vert}.
 \]
 From the above we have 
\[
  \alpha_A = E[K_{mw}(m)]/E[K_{un}(m)] = {d_A}/{d} .
\]
 Note that the value of $\alpha_A$ may be quite large or close to zero depending on both (1) the choice of $V_A$ and (2) the average degree of $G$.

\subsection{FS Steady State v.s.\ Uniform Distribution} \label{sec:FSvardist}
 Let $G=(V,E)$ be a connected graph and $V_A \subset V$ be a proper subset of $V$; define $V_B = V \backslash V_A$.
 Let $d_A = \vol(V_A) /\vert V_A \vert , $  $d_B = \vol(V_B) /\vert V_B \vert $, and $d = \vol(V)/\vert V \vert$ be the average degrees of the vertices in $V_A$, $V_B$, and $V$, respectively.
 Consider a FS process with $m$ walkers and let $K_{f\!s}(m)$ be a random variable that denotes the number of random walkers in $V_A$ in steady state.
 Let $K_{un}(m)$ be a random variable that denotes the number of sampled vertices, out of $m$ uniformly (randomly) sampled vertices from $V$, that belong to $V_A$.
 $K_{un}(m)$ has distribution
$
P[K_{un}(m) = k] = \binom{m}{k} p^k (1-p)^{m-k}, \, \forall k \geq 0 \, ,
$ where $p = \vert V_A \vert/\vert V \vert$.
 In this section we show that $K_{f\!s}(m)$ and $K_{un}(m)$ converge to the same limiting distribution, i.e., 
\begin{equation}\label{eq:fsuni}
 \lim_{m \to \infty} P[K_{f\!s}(m) = k] = \lim_{m \to \infty} P[K_{un}(m) = k], \: \forall k \geq 0 .
\end{equation}
 Recall that the FS algorithm starts $m$ random walkers at $m$ uniformly sampled vertices (sampled independently).
 Let $V_A \subseteq V$.
 As $m$ increases, eq.~(\ref{eq:fsuni}), the number of FS walkers that are initially selected to be in $V_A$ approaches the steady state distribution (assuming $G$ is connected).

 Let $L \in V^m$ be the state of FS;
 from (Theorem~\ref{thm:FST}) we have $$ P[L = (v_1, \dots, v_m)] = \frac{\sum_{i = 1}^m \deg(v_i)}{ m \vert V \vert^{m-1} \vol(V) }.$$
 In the following lemma we find the probability that $K_{f\!s}(m) = k$, $0 \leq k \leq m$.
\begin{lemma}\label{lem:Kfs}
\[
    P[K_{f\!s}(m) = k] = \frac{1 }{m d} \binom{m}{k} p^k (1-p)^{m-k} (k \, d_A +  (m-k)  d_B) \, ,
 \]
where $p = \vert V_A \vert /\vert V \vert$ and $0 \leq k \leq m$. 
\end{lemma}
\begin{proof}
 $P[K_{f\!s}(m) = k]$ is the sum of the probabilities $P[L = (v_1, \dots, v_m)]$ over all states $L$ in which exactly $k$ vertices belong to $V_A$.
 Consider $v_i$, the $i$-th element of $L$.
 When $v_i \in V_A$, the contribution of $v_i$ towards $P[K_{f\!s}(m) = k]$ is 
$$
  \vert V_A \vert^{k-1} \vert V_B \vert^{m-k} {\deg(v_i)}/({m \vert V \vert^{m-1} \vol(V)}) \, ;
$$
 when $v_i \in V_B$, the contribution of $v_i$ towards $P[K_{f\!s}(m) = k]$ is 
$$
  \vert V_A \vert^{k} \vert V_B \vert^{m-k-1}{\deg(v_i)}/({m \vert V \vert^{m-1} \vol(V)}).
$$
 Summing over all elements in $L$ and over all vertices yields
\begin{equation}
\begin{split}
   P[ & K_{f\!s}(m) = k]  = \binom{m}{k} \frac{\vert V_A \vert^{k} \vert V_B \vert^{m-k}}{ m \vert V \vert^{m-1} \vol(V) } \\ & \times \left( \sum_{i = 1}^k \sum_{\forall u_i \in V_A} \frac{\deg(u_i)}{\vert V_A \vert} \right. 
    \left. +  \sum_{j = 1}^{m-k} \sum_{\forall t_j \in V_B} \frac{\deg(t_j)}{\vert V_B \vert}  \right) \\
   & = \frac{1}{ m \, d}\binom{m}{k}\, p^k\, (1-p)^{m-k}  ( k\, d_A  + (m-k) d_B)
\end{split}
\end{equation}

%
\end{proof}
 The previous lemma gives the probability that a subset of vertices $V_A$ has $K_{f\!s}(m) \in \{0,\dots,m\}$ FS random walkers.
 The following theorem shows that $K_{f\!s}(m)$ and $K_{un}(m)$ converge to the same limiting distribution.
 \begin{thm}~\\
  $\lim_{m \to \infty} P[K_{f\!s}(m) = k]  = \lim_{m \to \infty} P[K_{un}(m) = k],$ $ \forall k \geq 0 .$
 \end{thm}
 \begin{proof}
 From Lemma~\ref{lem:Kfs} 
  \begin{align}\label{eq:Kfs2}
 P[K_{f\!s}(m) = k] = \frac{k\, d_A  + (m-k) d_B}{ m \, d} P[K_{un}(m) = k].
 \end{align}
 Note that if $k = mp$ we have
\begin{equation}\label{eq:mpeq1}
 \frac{ mp d_A  + (m-mp) d_B }{ m \, d} = \frac{ m \frac{\vol(V_A) }{\vert V \vert}  +  m\frac{\vol(V_B)}{\vert V \vert} }{ m \, d} = 1 \, .
\end{equation}
 As $m \to \infty$, the probability mass of $P[K_{un}(m) = k]$ gets highly concentrated around the interval $k \in [mp - c \sqrt{m},mp + c \sqrt{m]}$, for large enough values of $c$.
 Let $k^-(m) = m p - z(m) \sqrt{m p (1-p)}$ and $k^+(m) = m p + z(m) \sqrt{m p (1-p)}$, where $z(m) = o(\sqrt[6]{m})$\footnote{$f(m) = o(h(m))$ implies $\lim_{m \to \infty} f(m)/h(m) = 0$.} is a slow increasing function of $m$.
 Note that
  \begin{align*}
 \frac{z(m) \sqrt{m} \sqrt{p (1-p)} (d_A-d_B)}{m \, d} = o(m^{4/6})/m
 \end{align*}
 and, thus, eq.~(\ref{eq:mpeq1}) yields
\begin{align}\label{eq:km}
\lim_{m \to \infty} \frac{k^-(m) \, d_A  + (m-k^-(m)) d_B}{ m \, d} = 1
\end{align}
and
\begin{align}\label{eq:kp}
\lim_{m \to \infty} \frac{ k^+(m) \, d_A  + (m-k^+(m)) d_B}{ m \, d} = 1.
\end{align}
 All that is left to show is that $\lim_{m \to \infty}P[K_{un}(m) < k^-(m)] = 0$ and $\lim_{m \to \infty} P[K_{un}(m) > k^+] = 0$.
 Using an	 extension of the {\em de Moivre-Laplace} limit theorem shown in~\cite[pg.\ 193]{Feller1} yields
\begin{equation}\label{eq:ml}
\begin{split}
&\lim_{m \to \infty} P[K_{un}(m) < k^-(m)] = 0 \:, \mbox{ and}\\
& \lim_{m \to \infty} P[K_{un}(m) > k^+(m)] =  0 .
\end{split}
\end{equation}
Putting together eqs.~(\ref{eq:Kfs2}),~(\ref{eq:km}),~(\ref{eq:kp}), and~(\ref{eq:ml}), with 
$$
 \lim_{m \to \infty}  \frac{k(m) \, d_A  + (m-k(m)) d_B}{ m \, d} < \infty \, , \, k(m) = o(m)\, ,
$$
 yields
$$\lim_{m \to \infty} P[K_{f\!s}(m) = k] = \lim_{m \to \infty} P[K_{un}(m) = k], \: \forall k \geq 0 ,$$
which concludes our proof.
 \end{proof}
 We have seen as $m$ gets larger, FS gets closer to starting in steady state with respect to the number of FS random walkers inside $V_A$, $\forall V_A \subset V$.

 We have seen that if we initialize $m$ random walkers with uniformly sampled vertices, FS starts closer to steady state than MultipleRW.
 In what follows we show that FS is well suited to be used in large scale (parallel, asynchronous) experiments without incurring in any coordination or communication costs between the random walkers.
\subsection{Distributed FS} \label{sec:distributedFS}
 FS is well suited to be used in large scale (parallel, asynchronous) experiments.
 Let $B$ be the budget of FS. 
 In the distributed version of FS the budget is not directly related to the number of sampled vertices obtained by the algorithm.
 This is because distributed FS  is achieved using multiple independent random walkers where the cost of sampling a vertex $v$ is an exponentially distributed random variable with parameter $\deg(v)$.
 In what follows we show, using the Uniformization principle of Markov chains~\cite[Chapter 7.5]{Cassandras} and the Poisson decomposition property, that FS can be made fully distributed.

 \begin{thm}
  A MultipleRW sampling process where the cost of sampling a vertex $v$ is an exponentially distributed random variable with parameter $\deg(v)$ is equivalent to a FS process.
 \end{thm}
\begin{proof}
 Consider the following Distributed FS (DFS) process.
 Let $\chi = \{L(\tau) \in V^m : \tau \in \mathbb{R}^\star \}$ be the Markov chain associated with a random walker over $G^m = (V^m,E_m)$, the $m$-th Cartesian power of $G$, with transition rate matrix 
\[
 \bQ = \bA - \bD \, ,
\]
 where $\bA$ is the adjacency matrix of $G^m$, $\bA_{i,j} \in \{0,1\}, \, \forall i,j$, and $\bD$ is a diagonal matrix with $\bD_{i,i} = \sum_{\forall j} \bA_{i,j}$.
 We observe this FS process over the interval $[0,B]$. 
 
 In the DFS process, the probability that the $k$-th random walker transitions out of vertex $v_k$ at step $\tau + \Delta$ depends only on $\deg(v_k)$ and not on the state of $L(\tau)$.
 Thus, we can decompose the Poisson process describing a departure from the state $L_n^\prime=(v_1,\dots,v_m)$ into $m$ independent stochastic processes, where the $i$-th process is a Poisson process with parameter $\lambda_i = \deg(v_i)\, , i=1,\dots,m$.
 The above is equivalent to the stochastic process of a MultipleRW process with $m$ random walkers and budget $B$, where the cost of sampling a vertex $v$ is an exponentially distributed random variable with rate $\deg(v)$.

 The DFS is equivalent to a FS process via the {\em Uniformization} property of Markov chains~\cite[Chapter 7.5]{Cassandras}.
 The transition probability matrix of the Uniformized Markov chain (with unitary uniformization parameter) at the embedded transition points is 
\[
  \bP = I - \bD^{-1}\bQ = \bD^{-1}\bA \, ,
\]
 which is also the transition probability matrix of a FS process.
\end{proof}


\section{Results}\label{sec:results}

 In this section we compare FS with SingleRW and MultipleRW.
 We also contrast FS with random vertex and edge sampling.
 The experiments consist of executing these sampling methods on a variety of real world graphs.
 The datasets used in the simulations are summarized in Table~\ref{tab:graphs}:
 ``Flickr'', ``Livejournal'', and ``YouTube'' are popular photosharing, blog (weblog), and video sharing websites, respectively.
 Users are represented as vertices of a graph.
 In these websites a user can subscribe to other user updates; an edge $(u,v)$ exists between users $u$ and $v$ if user $u$ subscribes to user $v$. 
 At ``Livejournal'' and ``YouTube'' it is possible to query the incoming and outgoing edges of a given user.
 Further details of these three datasets can be found in~\cite{Mislove}.
 ``Internet RLT'' is a router-level Internet graph collected from traceroute measurements of 23 monitors distributed over the world~\cite{CAIDA}.
 Note that some of these graphs contain disconnected components (subgraphs).

 In the following simulations the starting vertex of each random walker is chosen uniformly at random from the set of all vertices.
 Our results show that FS estimates are consistently more accurate than their SingleRW and MultipleRW counterparts.
 Moreover, when restricted to the largest connected component, FS reaches steady state faster than SingleRW and MultipleRW in the simulations presented in \techreport{Appendix~\ref{appx:convergence}}{our technical report~\cite{TechReport}}.

\subsection{Assortative Mixing Coefficient}\label{sec:assortativityresults}
 In our first experiment we treat the graphs in Table~\ref{tab:graphs} as undirected graphs.
 In-degrees and out-degrees are represented as vertex labels 
 and the assortative mixing coefficient is obtained using the estimator described in Section~\ref{sec:assortativity}.
 
 In our experiment we average the estimates and calculate their mean squared error (MSE) over $100$ runs.
 The sampling budget is $\vert V \vert / 100$ for all graphs.
 Let $\hat{r}$ denote the estimated value of $r$.
 Table~\ref{tab:AMC} shows a summary of the relative bias of $\hat{r}$ ($1-E[\hat{r}]/r$) and $\hat{r}$'s \NMSE with respect to the true value of $r$.
 We observe that FS is consistently more accurate than both MultipleRW and SingleRW.
 If we focus on Flickr, the FS bias is $7$ fold smaller than the bias of both MultipleRW and SingleRW.
 In addition FS's \NMSE is one order of magnitude smaller than the {\NMSE}\!\!s of MultipleRW and SingleRW.
 The Internet graph (``Internet RLT'') is the only graph we studied that shows little difference between FS and MultipleRW.

 We also perform an extreme experiment that focuses on the impact of loosely connected components on the assortative mixing estimates. 
 Consider a graph that consists of two instances of a random undirected Barab\'asi-Albert~\cite{BA} graph,  $G_A$ and $G_B$, with $5\times 10^5$ vertices each and average degrees 2 and 10, respectively,
 joined by a single edge connecting the two smallest degree vertices in $G_A$ and $G_B$ (ties are resolved arbitrarily). 
 Henceforth, we use $G_{AB}$ to denote the above graph.
 It is worth noting that over the $G_{AB}$ graph, SingleRW consistently finds $\hat{r} = 0$ over all $100$ runs.
 This is because SingleRW only estimated the assortative mixing of either subgraph $A$ or subgraph $B$, which are both zero.
 Over $G_{AB}$ MultipleRW performs almost as bad as SingleRW while FS is able to accurately estimate $r$.

\begin{table*}[htb]
\begin{center}
{ 
\begin{tabular}{@{}llrrrrrr@{}} \toprule
Graph&$r$&\multicolumn{2}{c}{FS}&\multicolumn{2}{r}{MultipleRW}&\multicolumn{2}{r}{SingleRW}\\ 
&&Bias & $\vert \NMSE \vert$ &Bias & $\vert \NMSE \vert$ &Bias & $\vert \NMSE \vert$ \\\midrule 
Flickr&$0.007 $&$8\%$& $1.08$ &$752\%$& $7.65$  & $-619\%$ & $27.32$ \\ 
LiveJournal&$0.07$&$-0.5\%$&$0.11$ &$-12\%$& $0.16$ &$1\%$ & $0.17$ \\ 
Internet RLT&$0.17$&$3\%$&$0.33$ &$2\%$& $0.32$ &$17\%$ & $0.44$\\ 
Youtube&$-0.03$&$0.001\%$& $0.02$ &$2\%$& $0.03$ &$-1\%$ & $0.1$ \\
$G_{AB}$&$0.08$&$0.01\%$& $0.12$ &$70\%$& $0.72$ &$100\%$ & $1.00$ \\ 
  \bottomrule
\end{tabular}
}
\caption{Assortative mixing coefficient estimate bias and the module of the estimator \NMSE; $r$ is the true value of the global clustering coefficient and these values are estimated value over 100 runs. The sampling budget is $\vert V \vert / 100$ for all graphs.}\label{tab:AMC}
\end{center}
\vspace{-20pt}
\end{table*}

\subsection{In- and Out-degree Distribution Estimates} \label{sec:degdist}
 We now focus on estimating the in-degree distribution.
 Let $\theta = \{\theta_i\}_{\forall i \in \cL}$ denote the in-degree distribution, where $\theta_i$ is the fraction of vertices with in-degree $i$.
 In our simulations we estimate $\gamma_i = \sum_{k=i+1}^\infty \theta_k$, the CCDF of $\theta$, using equation~(\ref{eq:hattheta}).
 We choose to estimate the CCDF instead of the density because the CCDF is the plot of choice when it comes to displaying degree distributions.
 Each simulation consists of $10,000$ runs (sample paths) used to compute the empirical \CNMSE (equation~(\ref{eq:CNMSE})).
 The \CNMSE is used to compare the accuracy of the estimates obtained from FS (dimension $m \in \{10,1000\}$), SingleRW, and MultipleRW ($m \in \{10,1000\}$ walkers).
 For the sake of conciseness, we restrict our presentation to a handful of representative results.

	\begin{figure}[tbht]
	\begin{center}
         \begin{minipage}{\plotwidth}
  			   \vspace{-7pt}
				\resizebox{\plotwidth}{0.95\plotheight}{
				\hspace{-20pt}
      			   \input{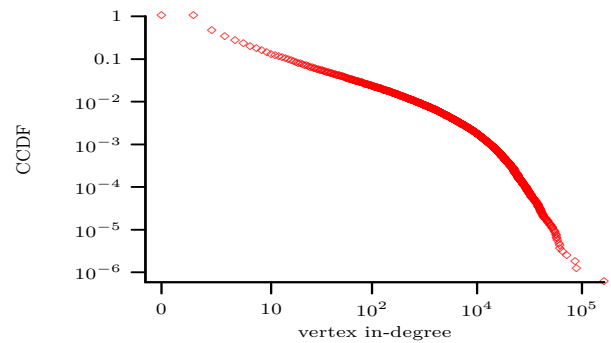}
			   }
			\caption{(Flickr) Log-log plot of the in-degree CCDF.}
			\label{fig:FlikrCCDF}
		\end{minipage}
	\end{center}
			\vspace{-10pt}
	\end{figure}

	\begin{figure}[tbht]
	\begin{center}
         \begin{minipage}{\plotwidth}
  			   \vspace{-0pt}
				\resizebox{\plotwidth}{0.95\plotheight}{
				\hspace{-20pt}
      			   \input{figs/samplingCCDF_indegree_FSunif_RWunif__flickr-links_STEPS18612_K1000_SCC1_INDEPCOST1_RUNS10000.amazonia.cs.umass.edu}
			   }
			\caption{	(LCC of Flickr) The log-log plot of the \CNMSE of the in-degree distribution estimates with budget $B = \vert V \vert / 100$.}
			\label{fig:LCCFlickrMSE}
		\end{minipage}
	\end{center}
			\vspace{-15pt}
	\end{figure}

	\begin{figure}[tbht]
	\begin{center}
         \begin{minipage}{\plotwidth}
  			   \vspace{-0pt}
				\resizebox{\plotwidth}{0.95\plotheight}{
				\hspace{-15pt}
      			   \input{figs/samplingCCDF_indegree_FSunif_RWunif__flickr-links_STEPS18612_K1000_SCC0_INDEPCOST1_RUNS10000.amazonia.cs.umass.edu}
			   }
			\caption{	(Flickr) The log-log plot of the \CNMSE of the in-degree distribution estimates with budget $B = \vert V \vert / 100$.}
			\label{fig:FlickrMSE}
		\end{minipage}
	\end{center}
			\vspace{-15pt}
	\end{figure}

	\begin{figure}[tbht]
	\begin{center}
         \begin{minipage}{\plotwidth}
  			   \vspace{0pt}
				\resizebox{\plotwidth}{0.95\plotheight}{
				\hspace{-15pt}
      			   \input{figs/graph_flickr_K1000_4runs_sample_path_biasedInDeg_flickr-links_STEPS2000000_K1000_SCC0_RUNS4.compute-0-17.local.log.gzdeg1}
			   }
			\caption{	(LCC of Flickr) Four sample paths of $\hat{\theta}_{1}$ ($\theta_{1} = 0.53$) as a function of the number of steps $n$ (horizontal axis in log scale).}
			\label{fig:4SPLCCFlickr}
		\end{minipage}
	\end{center}
			\vspace{-23pt}
	\end{figure}
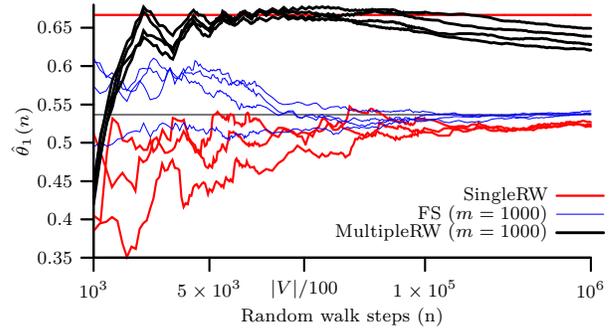

 Consider first two representative results from the Flickr graph, whose in-degree CCDF (complementary cumulative distribution function) log-log plot is shown in Figure~\ref{fig:FlikrCCDF}.
 The sampling budget is $B = 17,152 =\vert V \vert / 100$, which amounts to sampling $1\%$ of the vertices. 
 In the first simulation, we are restricted to the {\em Largest Connected Component} (LCC) (which contains $94\%$ of the vertices).
 The objective is to test if FS can outperform SingleRW and MultipleRW even when there are no disconnected components.
 Figure~\ref{fig:LCCFlickrMSE} shows a log-log plot of the \CNMSE of FS ($m=1000$), SingleRW, and MultipleRW ($m=1000$).
 In this experiment FS outperforms both SingleRW and MultipleRW.
 It is interesting to note that the estimates obtained by SingleRW are more accurate than the estimates obtained by MultipleRW.
 Now consider the complete Flickr graph.
 Figure~\ref{fig:FlickrMSE} shows a log-log plot of the \CNMSE of the in-degree distribution estimates.
 Contrasting the plots shown in Figures~\ref{fig:LCCFlickrMSE} and~\ref{fig:FlickrMSE} we see that the gap between FS and both SingleRW and MultipleRW has significantly increased, favoring FS.

 To better understand the differences between these sampling methods, Figure~\ref{fig:4SPLCCFlickr} focuses on four runs (sample paths) of the simulation over the complete Flickr graph.  
 Figure~\ref{fig:4SPLCCFlickr} plots the evolution of $\hat{\theta}_1$ (the estimate of $\theta_1$) as a function of $n$ (the number of steps in the random walk).
 At each run of the simulator both FS and MultipleRW start at the same initial set of vertices (chosen using random vertex sampling).
 Figure~\ref{fig:4SPLCCFlickr} shows that all four FS sample paths (runs) quickly converge to the value of $\theta_1$.
 For SingleRW, three of the four runs start inside the LCC. 
 These runs do not converge to the value of $\theta_1$ as some vertices with in-degree one lie outside the LCC.
 In one of the runs, SingleRW starts in a small disconnected component and, thus, grossly overestimates the value of $\theta_1$.
 For a similar reason, i.e., walkers starting at small disconnected components, MultipleRW grossly overestimates the value of $\theta_1$. 
 The MultipleRW rapid increase of $\hat{\theta}_1$ at around $n=10^3$ steps needs further investigation.
 It may be due to the transient of the random walk (discussed in Section~\ref{sec:MRW}).
 Even when $n \gg 1$ (not shown in Figure~\ref{fig:4SPLCCFlickr}) the MultipleRW estimate is unable to converge to $\theta_1$.
 Modifying both SingleRW and MultipleRW methods to cope with disconnected components is an interesting open problem.

	\begin{figure}[tbht]
	\begin{center}
         \begin{minipage}{\plotwidth}
  			   \vspace{10pt}
				\resizebox{\plotwidth}{0.95\plotheight}{
				\hspace{-20pt}
      			   \input{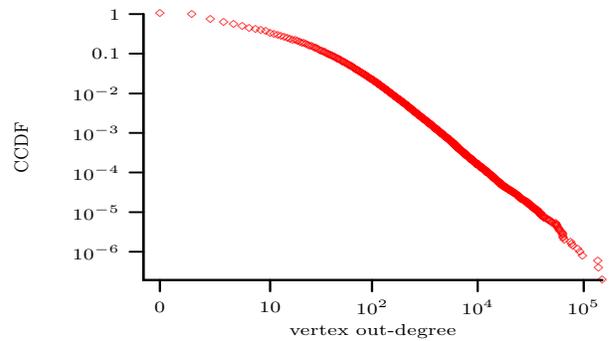}
			   }
			\caption{	(Livejournal) Log-log plot of the out-degree CCDF.}
			\label{fig:LiveJournalOutDeg}
		\end{minipage}
	\end{center}
			\vspace{-10pt}
	\end{figure}

 For the sake of conciseness, we omit the results of the simulations over the remaining graphs (Table~\ref{tab:graphs}) as they are similar to the results observed over the Flickr graph.
 However, consider the {\em out-degree} distribution estimates of Livejournal.
 Figure~\ref{fig:LiveJournalOutDeg} shows a log-log plot of the CCDF of the out-degrees.
 The log-log plot of the \CNMSE is shown in Figure~\ref{fig:LiveJournalMSE} for FS ($m=100$), SingleRW, and MultipleRW ($m =100$) with sampling budget $B = \vert V \vert / 10$. 
 From Figure~\ref{fig:LiveJournalMSE} we see that estimates of vertices with small out-degrees in FS are up to one order of magnitude more accurate than those obtained from both SingleRW and MultipleRW.

	\begin{figure}[tbht]
	\begin{center}
         \begin{minipage}{\plotwidth}
  			   \vspace{0pt}
				\resizebox{\plotwidth}{0.95\plotheight}{
				\hspace{-20pt}
      			   \input{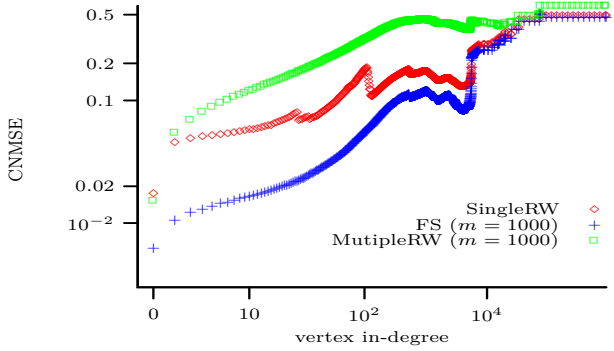}
			   }
			\caption{	(Livejournal) The log-log plot of the \CNMSE of the out-degree distribution estimation with sampling budget $B = \vert V \vert / 100$ (\CNMSE over $10,000$ runs). }
			\label{fig:LiveJournalMSE}
		\end{minipage}
	\end{center}
			\vspace{-10pt}
	\end{figure}

	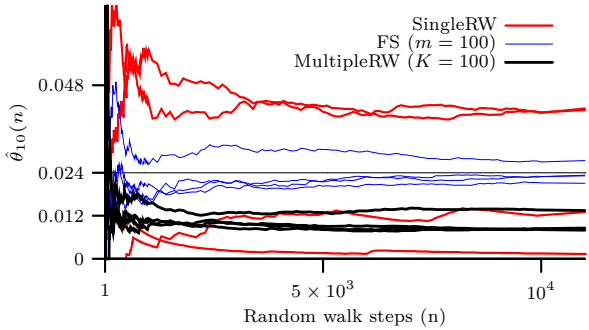
\begin{figure}[tbht]
	\begin{center}
         \begin{minipage}{\plotwidth}
  			   \vspace{0pt}
				\resizebox{\plotwidth}{0.95\plotheight}{
				\hspace{-15pt}
      			   \input{figs/graph_BA_4runs_deg10}
			   }
			\caption{($G_{AB}$ graph) Four paths of $\hat{\theta}_{10}$ as a function of the number of steps $n$ ($\theta_{10} = 0.024$).}
			\label{fig:4SP}
		\end{minipage}
	\end{center}
			\vspace{-20pt}
	\end{figure}

 The next experiment focuses on studying the impact of loosely connected components on the degree distribution estimates. 
 For this we use the two Barab\'asi-Albert joined graphs $G_{AB}$ presented in Section~\ref{sec:assortativity}.
 The experiment consists of estimating the degree distribution of $G_{AB}$ using FS ($m=100$), SingleRW, and MultipleRW ($m = 100$). 
 Again, both FS and MultipleRW start at the initial set of vertices in each simulation (chosen uniformly at random).
 In this experiment the hypothesis is that, for small sampling budgets, each random walker will see the degree distribution of either $G_A$ or $G_B$ but not the degree distribution of $G_{AB}$.
 Moreover, as the starting vertex of each random walker is chosen uniformly at random, $G_A$, which has the same number of vertices as $G_B$ but $1/5$ of the edges, receives more random walkers than its {\em per edge} ``share''. 
 Consequently, MultipleRW oversamples $G_A$.
 
Figure~\ref{fig:4SP} shows the results of four simulation runs and plots the evolution of the estimates of $\theta_{10}$ ($\hat{\theta}_{10}$) as a function of the number of steps. 
 In this simulation note that: (1) FS quickly converges to a value that is close to the correct value; (2) two out of the four SingleRW runs overestimate $\theta_{10}$ and the remaining two underestimate it; (3) three out of the four MultipleRW runs converge to the same, incorrect, fraction (underestimating the true value of $\theta_{10}$).
 FS is designed to be robust to disconnected or loosely connected components.
 All of the FS runs quickly converge to a good estimates of $\theta_{10}$.
 Figure~\ref{fig:BANMSE} also shows that the \CNMSE for FS is consistently lower than the \CNMSE for SingleRW and MultipleRW.

	\begin{figure}[tbht]
	\begin{center}
         \begin{minipage}{\plotwidth}
  			   \vspace{30pt}
				\resizebox{\plotwidth}{0.95\plotheight}{
				\hspace{-20pt}
      			   \input{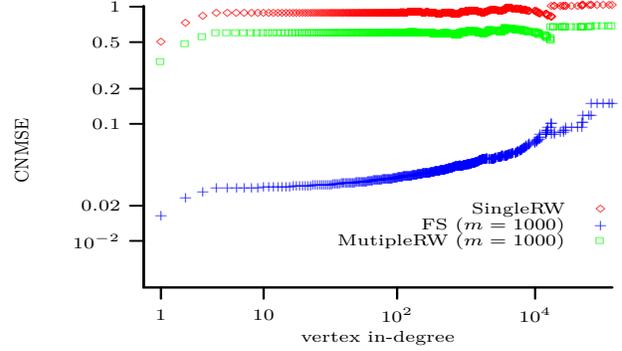}
			   }
			\caption{	($G_{AB}$ graph) The log-log plot of the \CNMSE of the degree distribution estimation with sampling budget $B = \vert V \vert / 100$ (\CNMSE over $10,000$ runs). }
			\label{fig:BANMSE}
		\end{minipage}
	\end{center}
			\vspace{-20pt}
	\end{figure}

\subsection{FS v.s. Stationary MultipleRW \& SingleRW}
 We now compare FS with SingleRW and MultipleRW, when the latter two start in steady state.
 Figure~\ref{fig:SSCCDFNMSE} shows the results (over the Flickr graph) of the same simulation scenario used to obtain the results in Figure~\ref{fig:FlickrMSE}, except that now MultipleRW and SingleRW both start in steady state.
 While SingleRW has improved slightly (most notably at the tail errors), the benefit of starting in steady state is most felt by the MultipleRW method.
 In this simulation we see that the large estimation errors of MultipleRW in the previous simulations were due to the starting vertices being sampled uniformly at random.
 It is interesting to observe that MultipleRW starting in steady state and FS have similar estimation errors.

	\begin{figure}[tbht]
	\begin{center}
         \begin{minipage}{\plotwidth}
  			   \vspace{30pt}
				\resizebox{\plotwidth}{0.95\plotheight}{
				\hspace{-20pt}
      			   \input{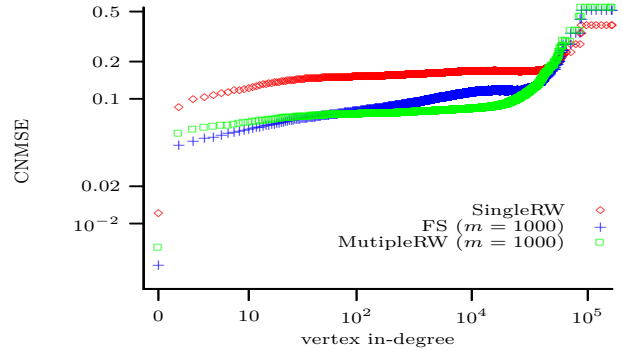}
			   }
			\caption{	(Flickr) The log-log plot of the \CNMSE of the in-degree distribution estimation of MultipleRW and SingleRW starting in steady state; sampling budget $B = \vert V \vert / 100$ (\NMSE over $10,000$ runs). }
			\label{fig:SSCCDFNMSE}
		\end{minipage}
	\end{center}
			\vspace{-20pt}
	\end{figure}  

\subsection{FS v.s.\ Random Independent Sampling}\label{sec:FSRS}
 In Section~\ref{sec:random} we showed that, if the degrees of two neighboring vertices are independent, random edge sampling is more accurate than random vertex sampling when it comes to estimating the tail of the degree distribution.
 In this section we observe this to be true over large real world graphs; we also observe that the accuracy of FS closely matches the accuracy of random edge sampling.
 In the following simulations we estimate the in-degree distribution. 
 Random edge sampling uses the estimator $\hat{\theta}_i$, equation~(\ref{eq:hattheta}) (the estimator used for sampled vertices is trivial). 
 
 In our first simulation we set the sampling cost of random vertex sampling to one and random edge sampling has cost two (as each edge samples two vertices).
 The sampling budget is $B = \vert V \vert / 100$.
 We label this simulation ``$100\%$ hit ratio'' to indicate the unitary cost of randomly sampling vertices.
 Figure~\ref{fig:Flickr100lIID} shows a log-log plot of the \NMSE (not the \CNMSE) of our simulation over the (complete) Flickr graph.
 Here we use the \NMSE (instead of the \CNMSE) in order to be able to compare our results with the ones presented in equations~(\ref{eq:NMSEE}) and~(\ref{eq:NMSEV}).
 The vertical line indicates the average in-degree.
 Note that random edge sampling is more (less) accurate than random vertex sampling at estimating in-degrees larger (smaller) than the average in-degree, as predicted by equations~(\ref{eq:NMSEE}) and~(\ref{eq:NMSEV}) of our model in Section~\ref{sec:random}.
 We also observe that the accuracy of FS ($m=1000$) closely matches the accuracy of random edge sampling.

	\begin{figure}[tbht]
	\begin{center}
         \begin{minipage}{\plotwidth}
  			   \vspace{20pt}
				\resizebox{\plotwidth}{0.95\plotheight}{
				\hspace{-15pt}
      			   \input{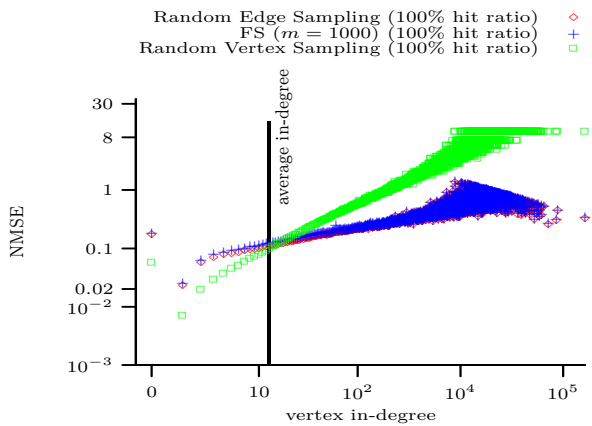}
			   }
			\caption{	(Flickr) The log-log plot shows the \NMSE of the in-degree distribution estimation with budget $B = \vert V \vert / 100 = 18612$ (\CNMSE over $10,000$ runs).\vspace{-10pt}}
			\label{fig:Flickr100lIID}
		\end{minipage}
	\end{center}
			\vspace{-5pt}
	\end{figure}

	\begin{figure}[tbht]
	\begin{center}
         \begin{minipage}{\plotwidth}
  			   \vspace{20pt}
				\resizebox{\plotwidth}{0.95\plotheight}{
				\hspace{-20pt}
      			   \input{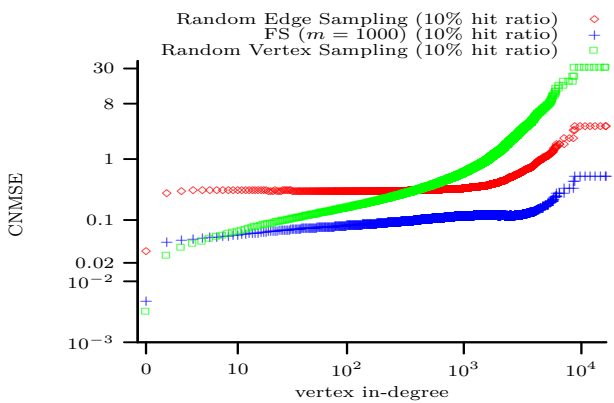}
			   }
			\caption{	(Livejournal) The log-log plot shows the \CNMSE of the in-degree distribution estimation with budget $B = \vert V \vert / 100 = 52844$ (\CNMSE over $10,000$ runs).\vspace{-10pt}}
			\label{fig:LiveJournalIID}
		\end{minipage}
	\end{center}
			\vspace{-10pt}
	\end{figure}

 Some complex networks exhibit a sparse user-id space. 
 In this scenario a fraction of the sampling budget $B$ can be spent querying invalid users-ids. 
 Motivated by recent experiments over the MySpace network~\cite{MySpace}, the following experiment assumes that only $10\%$ of the user-ids are valid, i.e., on average only one in every ten randomly sampled vertices are valid. 
 We denote this value ($10\%$) to be the {\em hit ratio}.  
 For random edge sampling we assume a {\em hit ratio} of $1\%$ (the choice of $1\%$ is arbitrary).
 Figure~\ref{fig:LiveJournalIID} shows a log-log plot of the \CNMSE of our simulation over the (complete) Livejournal graph with sampling budget $B=\vert V \vert / 100 = 52844$.
 We observe that FS ($m=1000$), which samples $m=1000$ random vertices and (on average) crawls $B-10\,m$ vertices, outperforms random edge sampling.
 Also note that FS estimates are more accurate than the estimates obtained from random vertex sampling for all but the three smallest in-degrees. 
 This indicates that FS is more robust to low hit ratios than random vertex and edges sampling.

	\begin{figure}[tbht]
	\begin{center}
         \begin{minipage}{\plotwidth}
  			   \vspace{5pt}
				\resizebox{\plotwidth}{0.95\plotheight}{
				\hspace{-20pt}
      			   \input{figs/biasedGroupFreqDist_flickr-links_flickr-groupmemberships_STEPS18612_K100_SCC0_RUNS10000.compute-0-15.local}
			   }
			\caption{	(Flickr) The \NMSE of the density estimates of the most popular groups in the Flickr graph.\vspace{-10pt}}
			\label{fig:FlickrGroups}
		\end{minipage}
	\end{center}
			\vspace{-10pt}
	\end{figure}
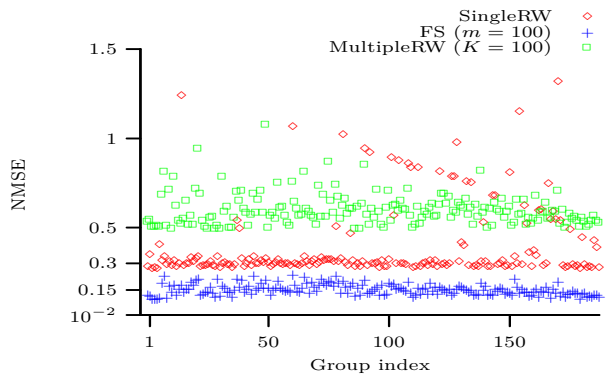

\subsection{Density of Special Interest Groups} \label{sec:groupsdist}
 In a variety of complex networks, e.g.\ on-line social networks, each vertex (user) is associated with multiple labels that represent group affiliations, e.g.\ user interests, user geolocation, among others.
 For example, in the Flickr graph $21\%$ of the users belong to one or more special interest groups~\cite{Mislove}.
 Let $\cL$ denote the set of groups in the Flickr graph and $\theta_l$ denote the fraction of vertices that belong to group $l \in \cL$.
 In the simulations we estimate $\theta_l$ using FS ($m=100$), SingleRW, and MultipleRW $(m=100)$ with budget $B= \vert V \vert / 100$.
 Figure~\ref{fig:FlickrGroups} shows the \NMSE (from $10,000$ runs) of the most popular $200$ groups ordered in decreasing popularity.
 FS is clearly superior to both SingleRW and MultipleRW.
 Even when restricting the random walks to the largest connected component, FS still noticeably outperforms MultipleRW ($m=100$) and SingleRW.

\begin{table}[htb]
\begin{center}
\textsf{ \small
\begin{tabular}{@{}llllll@{}} \toprule
 &&&\multicolumn{3}{c}{$E[\hat{C}]\:(\NMSE)$} \\ \cmidrule(r){4-6} 
  Graph & $B$  & $C$  & FS & SingleRW & MultipleRW\\ \midrule 
Flickr & $1\%$ & $0.14$ & $0.13 \: (0.04)$ & $0.12 \: (0.33)$ & $0.16 \: (0.18)$ \\ 
LiveJournal & $1\%$ & $0.16$  & $0.16 \: (0.02)$  &$ 0.16 \: (0.02)$ & $0.17 \: (0.06)$ \\ 
  \bottomrule
\end{tabular}
}
\caption{Global clustering coefficient estimates. $C$ is the true value of the global clustering coefficient and $\hat{C}$ is its estimated value.}\label{tab:GCC}
\end{center}
\vspace{-20pt}
\end{table}

\subsection{Global Clustering Coefficient Estimates}
 In our last set of experiments we evaluate the accuracy of estimating the global clustering coefficient using FS, SingleRW, and MultipleRW. 
 Our simulations show a small difference between FS ($m = 1000$), SingleRW, MultipleRW ($m =1000$). 
 Let $C$ be the clustering coefficient and $\hat{C}$ denote its estimated value.
 Table~\ref{tab:GCC} presents the empirical value of $E[\hat{C}]$ and the empirical \NMSE of the clustering coefficient, given by
 \[
 \NMSE = \frac{\sqrt{E[( \hat{C} - C )^2]}}{C} \, ,
 \]  
 over $10,000$ runs of FS, SingleRW, and MultipleRW over the Flickr and Livejournal graphs. From the results of Table~\ref{tab:GCC} we see that FS accurately estimates the global clustering coefficient and has smaller error than both SingleRW and MultipleRW.

\section{Related work}\label{sec:literature}

 This section is devoted to review the related literature. 
 FS can be classified as a Markov Chain Monte Carlo (MCMC) method.
 Other MCMC-based methods have been applied to characterize complex networks.
 Applications include, but are not limited to estimating characteristics of a population~\cite{RDSprob} (e.g.\ estimation of HIV seroprevalence among drug users~\cite{RDSINJECTION}),  content density in peer-to-peer networks~\cite{gkantsidis05randomwalksp2p, LaurentP2P, WillingerRDS, WillingerMRW}, uniformly sampling Web pages from the Internet~\cite{UniformURLRW,WWWUniformSamplingRW}, and 
uniformly sampling Web pages from a search engine's index~\cite{YossefSamplingWeb}.
 The above literature is mostly concerned with random walks that seek to sample vertices uniformly (also known as Metropolized Random Walks or Metropolis-RW)~\cite{gkantsidis05randomwalksp2p,UniformURLRW,WWWUniformSamplingRW,YossefSamplingWeb,WillingerMRW}.
 The accuracy of RW and Metropolis-RW (MRW) is compared in~\cite{MRWFacebook,WillingerRDS}, and in a variety of experiments RW estimates are shown to be consistently more accurate than or equal to MRW estimates.

 The above literature does not consider the use of multiple random walks to address the problem of estimating characteristics of disconnected or loosely connected graphs.
 While multiple independent random walkers have been used as a convergence test in the literature, our simulations in Section~\ref{sec:results} show that independent walkers are not suited to sample loosely connected graphs when the starting vertices are selected uniformly at random.
 
 A number of real complex networks are known to have disconnected or loosely connected components.
 A large body of MCMC literature is dedicated to overcome the locality problem described in Section~\ref{sec:AccStru}.
 However, the literature either assumes that the graph is very structured, e.g., a $2$ dimensional lattice, or that the graph is completely known.  These assumptions make the solutions inapplicable to our problem.
 A comprehensive list of MCMC methods and their characteristics can be found in~\cite{MCMC}. 

 Projecting a RW onto a higher dimensional space has been used in~\cite{Lifting} to make the Markov chain associated to the random walker nonreversible, which can speed up the mixing of the original RW. 
 Unfortunately, it is unclear if this method can be successfully used to estimate characteristics of complex networks. 

 In networks that cannot be crawled (e.g., the Internet topology), samples must be obtained along shortest paths, and vertex degrees cannot be queried,~\cite{TracerouteBiasTheory} shows that observed vertex degrees are biased.
 Our work, however, assumes a graph can be crawled and vertex degrees queried.
 Our scenario admits a RW with an unbiased estimator.
 Multiple random walks also find other applications besides the one presented in this work.
 They are used to collect Web data~\cite{ParallelCrawlers}, search P2P networks~\cite{MultipleRWSearch,MRW2}, and decrease the time to discover ``new wireless nodes''~\cite{WirelessMultipleRW}.
 Dependent multiple random walks are also used in percolation theory~\cite{2Dpercolation}.

\section{Discussion and Future work}\label{sec:conclusions}
 In this work we presented a new and promising random walk-based method ({\em Frontier sampling}) that mitigates the estimation errors caused by subgraphs that ``trap'' a random walker. 
 Frontier sampling (FS) uses multiple ($m$) mutually {\em dependent} random walkers starting from vertices sampled uniformly at random.
 The FS samples are shown to be the projection (onto the original graph) of a special type of $m$-dimensional (single) random walker. 
 Simulations over real world graphs in Section~\ref{sec:results} show that Frontier sampling (FS) is more robust than single and multiple independent random walkers (starting out of steady state) to estimate in-degree distributions and the fraction of users that belong to a social group.
 We also present evidence, using an analytical argument (also substantiated by simulations), that random walks (in particular, FS) are better suited to estimate the tail (all degrees greater than the average) of degree distributions than random vertex sampling.
 FS can also be made fully distributed without incurring in any coordination or communication costs.
 
 The ideas behind FS can have far reaching implications, from estimating characteristics of dynamic networks to the design of new MCMC-based approximation algorithms.

\section{Acknowledgments}
{

We would like to thank Weibo Gong for many helpful discussions and Alan Mislove for kindly making available some of the data used in this paper. This research was sponsored by the ARO under MURI W911NF-08-1-0233, and the U.S. Army Research Laboratory under Cooperative Agreement Number W911NF-09-2-0053. The views and conclusions contained in this document are those of the author(s) and should not be interpreted as representing the official policies, either expressed or implied, of the U.S. Army Research Laboratory or the U.S. Government. The U.S. Government is authorized to reproduce and distribute reprints for Government purposes notwithstanding any copyright notation hereon.
 }
\bibliographystyle{plain}
\vspace{-7pt}
\balance
\bibliography{FunctionEstimationOverGraphs}
\techreport{
\input{appendix}

}{}

\end{document}

%% file: figs/graph_flickr_K1000_4runs_sample_path_biasedInDeg_flickr-links_STEPS2000000_K1000_SCC0_RUNS4.compute-0-17.local.log.gzdeg1.tex
\setlength{\unitlength}{0.120450pt}
\begin{picture}(2100,1080)(0,0)
\fontsize{7}{8.98332}\selectfont
\thicklines \path(374,226)(333,226)
\put(304,226){\makebox(0,0)[r]{ 0.35}}
\thicklines \path(374,347)(333,347)
\put(304,347){\makebox(0,0)[r]{ 0.4}}
\thicklines \path(374,467)(333,467)
\put(304,467){\makebox(0,0)[r]{ 0.45}}
\thicklines \path(374,588)(333,588)
\put(304,588){\makebox(0,0)[r]{ 0.5}}
\thicklines \path(374,708)(333,708)
\put(304,708){\makebox(0,0)[r]{ 0.55}}
\thicklines \path(374,829)(333,829)
\put(304,829){\makebox(0,0)[r]{ 0.6}}
\thicklines \path(374,950)(333,950)
\put(304,950){\makebox(0,0)[r]{ 0.65}}
\thicklines \path(374,226)(374,185)
\put(374,127){\makebox(0,0){$10^3$}}
\thicklines \path(759,226)(759,185)
\put(759,127){\makebox(0,0){$5 \times 10^3$}}
\thicklines \path(1074,226)(1074,185)
\put(1074,127){\makebox(0,0){$\vert V \vert/100$}}
\thicklines \path(1476,226)(1476,185)
\put(1476,127){\makebox(0,0){$1 \times 10^5$}}
\thicklines \path(2027,226)(2027,185)
\put(2027,127){\makebox(0,0){$10^6$}}
\thicklines \path(374,1022)(374,226)(2027,226)
\put(101,624){\makebox(0,0)[l]{\rotatebox[origin=c]{90}{$\hat{\theta}_{1}(n)$}}}
\put(1200,40){\makebox(0,0){Random walk steps (n)}}
\thinlines \path(374,676)(374,676)(397,676)(418,676)(437,676)(455,676)(471,676)(486,676)(501,676)(515,676)(528,676)(540,676)(552,676)(563,676)(573,676)(583,676)(593,676)(603,676)(612,676)(620,676)(629,676)(637,676)(645,676)(652,676)(660,676)(667,676)(674,676)(681,676)(687,676)(693,676)(700,676)(706,676)(712,676)(717,676)(723,676)(729,676)(734,676)(739,676)(744,676)(749,676)(754,676)(759,676)(764,676)(769,676)(773,676)(778,676)(782,676)(786,676)(790,676)(795,676)(799,676)
\thinlines \path(799,676)(803,676)(807,676)(811,676)(814,676)(818,676)(822,676)(826,676)(829,676)(833,676)(836,676)(840,676)(843,676)(846,676)(850,676)(853,676)(856,676)(859,676)(862,676)(866,676)(869,676)(872,676)(875,676)(878,676)(880,676)(883,676)(886,676)(889,676)(892,676)(894,676)(897,676)(900,676)(902,676)(905,676)(908,676)(910,676)(913,676)(915,676)(918,676)(920,676)(923,676)(925,676)(948,676)(969,676)(988,676)(1006,676)(1022,676)(1037,676)(1052,676)(1066,676)(1079,676)
\thinlines \path(1079,676)(1091,676)(1103,676)(1114,676)(1124,676)(1134,676)(1144,676)(1154,676)(1163,676)(1171,676)(1180,676)(1188,676)(1196,676)(1203,676)(1211,676)(1218,676)(1225,676)(1232,676)(1238,676)(1244,676)(1251,676)(1257,676)(1263,676)(1268,676)(1274,676)(1280,676)(1285,676)(1290,676)(1295,676)(1300,676)(1305,676)(1310,676)(1315,676)(1320,676)(1324,676)(1329,676)(1333,676)(1337,676)(1341,676)(1346,676)(1350,676)(1354,676)(1358,676)(1362,676)(1365,676)(1369,676)(1373,676)(1377,676)(1380,676)(1384,676)(1387,676)
\thinlines \path(1387,676)(1391,676)(1394,676)(1397,676)(1401,676)(1404,676)(1407,676)(1410,676)(1413,676)(1417,676)(1420,676)(1423,676)(1426,676)(1429,676)(1431,676)(1434,676)(1437,676)(1440,676)(1443,676)(1445,676)(1448,676)(1451,676)(1453,676)(1456,676)(1459,676)(1461,676)(1464,676)(1466,676)(1469,676)(1471,676)(1474,676)(1476,676)(1499,676)(1520,676)(1539,676)(1557,676)(1573,676)(1588,676)(1603,676)(1617,676)(1630,676)(1642,676)(1654,676)(1665,676)(1675,676)(1685,676)(1695,676)(1705,676)(1714,676)(1722,676)(1731,676)
\thinlines \path(1731,676)(1739,676)(1747,676)(1754,676)(1762,676)(1769,676)(1776,676)(1783,676)(1789,676)(1795,676)(1802,676)(1808,676)(1814,676)(1819,676)(1825,676)(1831,676)(1836,676)(1841,676)(1846,676)(1851,676)(1856,676)(1861,676)(1866,676)(1871,676)(1875,676)(1880,676)(1884,676)(1888,676)(1892,676)(1897,676)(1901,676)(1905,676)(1909,676)(1913,676)(1916,676)(1920,676)(1924,676)(1928,676)(1931,676)(1935,676)(1938,676)(1942,676)(1945,676)(1948,676)(1952,676)(1955,676)(1958,676)(1961,676)(1964,676)(1968,676)(1971,676)
\thinlines \path(1971,676)(1974,676)(1977,676)(1980,676)(1982,676)(1985,676)(1988,676)(1991,676)(1994,676)(1996,676)(1999,676)(2002,676)(2004,676)(2007,676)(2010,676)(2012,676)(2015,676)(2017,676)(2020,676)(2022,676)(2025,676)(2027,676)(2027,676)
\put(1884,421){\makebox(0,0)[r]{SingleRW}}
\color{red}
\thicklines \path(1913,421)(2070,421)
\thicklines \path(374,359)(374,359)(397,342)(418,558)(437,664)(455,648)(471,595)(486,591)(501,603)(515,568)(528,579)(540,565)(552,566)(563,596)(573,560)(583,571)(593,596)(603,629)(612,614)(620,600)(629,595)(637,666)(645,663)(652,659)(660,657)(667,655)(674,660)(681,634)(687,615)(693,593)(700,583)(706,575)(712,558)(717,561)(723,555)(729,548)(734,594)(739,598)(744,577)(749,581)(754,600)(759,597)(764,610)(769,667)(773,675)(778,678)(782,679)(786,685)(790,680)(795,676)(799,672)
\thicklines \path(799,672)(803,667)(807,669)(811,661)(814,660)(818,670)(822,660)(826,676)(829,678)(833,675)(836,676)(840,672)(843,648)(846,652)(850,651)(853,650)(856,650)(859,648)(862,645)(866,640)(869,647)(872,644)(875,646)(878,649)(880,676)(883,680)(886,682)(889,679)(892,663)(894,654)(897,655)(900,656)(902,657)(905,663)(908,661)(910,664)(913,665)(915,663)(918,660)(920,659)(923,661)(925,659)(948,628)(969,614)(988,606)(1006,631)(1022,623)(1037,624)(1052,608)(1066,608)(1079,582)
\thicklines \path(1079,582)(1091,616)(1103,622)(1114,629)(1124,623)(1134,626)(1144,633)(1154,649)(1163,646)(1171,638)(1180,641)(1188,654)(1196,652)(1203,643)(1211,652)(1218,669)(1225,702)(1232,695)(1238,689)(1244,691)(1251,687)(1257,686)(1263,684)(1268,682)(1274,682)(1280,681)(1285,681)(1290,690)(1295,695)(1300,694)(1305,690)(1310,688)(1315,684)(1320,686)(1324,683)(1329,678)(1333,677)(1337,675)(1341,672)(1346,668)(1350,665)(1354,660)(1358,666)(1362,661)(1365,660)(1369,657)(1373,657)(1377,658)(1380,655)(1384,656)(1387,654)
\thicklines \path(1387,654)(1391,652)(1394,654)(1397,658)(1401,658)(1404,655)(1407,651)(1410,647)(1413,647)(1417,651)(1420,648)(1423,646)(1426,644)(1429,648)(1431,647)(1434,644)(1437,645)(1440,648)(1443,653)(1445,665)(1448,661)(1451,662)(1453,662)(1456,662)(1459,660)(1461,660)(1464,662)(1466,661)(1469,660)(1471,660)(1474,664)(1476,660)(1499,660)(1520,653)(1539,645)(1557,639)(1573,637)(1588,640)(1603,640)(1617,637)(1630,645)(1642,646)(1654,640)(1665,636)(1675,634)(1685,629)(1695,626)(1705,625)(1714,625)(1722,628)(1731,625)
\thicklines \path(1731,625)(1739,625)(1747,625)(1754,630)(1762,629)(1769,634)(1776,633)(1783,631)(1789,628)(1795,631)(1802,629)(1808,628)(1814,629)(1819,629)(1825,629)(1831,629)(1836,627)(1841,630)(1846,632)(1851,635)(1856,634)(1861,634)(1866,635)(1871,638)(1875,640)(1880,639)(1884,638)(1888,639)(1892,638)(1897,638)(1901,638)(1905,639)(1909,639)(1913,640)(1916,640)(1920,641)(1924,639)(1928,650)(1931,650)(1935,649)(1938,647)(1942,646)(1945,645)(1948,646)(1952,647)(1955,647)(1958,646)(1961,645)(1964,645)(1968,644)(1971,644)
\thicklines \path(1971,644)(1974,646)(1977,646)(1980,646)(1982,646)(1985,646)(1988,646)(1991,645)(1994,644)(1996,644)(1999,644)(2002,643)(2004,641)(2007,640)(2010,638)(2012,639)(2015,639)(2017,639)(2020,640)(2022,640)(2025,639)(2027,639)(2027,639)
\thicklines \path(374,617)(395,564)(416,505)(435,448)(453,474)(470,493)(485,472)(500,461)(513,436)(526,541)(539,527)(551,534)(562,558)(572,596)(583,604)(592,586)(602,572)(611,556)(620,557)(628,561)(636,565)(644,574)(652,598)(659,593)(666,608)(673,624)(680,614)(686,597)(693,600)(699,591)(705,585)(711,580)(717,563)(723,565)(728,555)(733,542)(739,538)(744,518)(749,511)(754,509)(759,519)(763,514)(768,510)(773,538)(777,527)(782,532)(786,559)(790,561)(794,562)(798,569)(802,583)
\thicklines \path(802,583)(806,577)(810,568)(814,562)(818,581)(822,596)(825,623)(829,620)(832,620)(836,620)(839,622)(843,618)(846,616)(849,614)(853,610)(856,608)(859,611)(862,604)(865,602)(868,596)(871,592)(874,591)(877,587)(880,586)(883,582)(886,583)(889,579)(891,585)(894,580)(897,581)(900,584)(902,581)(905,579)(907,586)(910,590)(912,587)(915,583)(917,584)(920,585)(922,584)(925,577)(948,585)(968,613)(988,635)(1005,634)(1022,629)(1037,618)(1052,613)(1066,621)(1078,619)(1091,616)
\thicklines \path(1091,616)(1102,607)(1114,631)(1124,639)(1134,663)(1144,663)(1154,651)(1163,639)(1171,629)(1180,628)(1188,631)(1196,641)(1203,634)(1211,632)(1218,630)(1225,628)(1231,625)(1238,622)(1244,619)(1251,618)(1257,616)(1263,612)(1268,614)(1274,606)(1279,605)(1285,608)(1290,611)(1295,611)(1300,608)(1305,613)(1310,611)(1315,625)(1319,623)(1324,618)(1329,618)(1333,617)(1337,615)(1341,610)(1346,612)(1350,614)(1354,609)(1358,608)(1362,604)(1365,606)(1369,604)(1373,600)(1377,597)(1380,606)(1384,601)(1387,600)(1391,598)
\thicklines \path(1391,598)(1394,604)(1397,605)(1401,605)(1404,607)(1407,609)(1410,607)(1413,606)(1417,608)(1420,608)(1423,609)(1426,608)(1428,611)(1431,608)(1434,616)(1437,613)(1440,612)(1443,617)(1445,614)(1448,613)(1451,611)(1453,624)(1456,621)(1459,619)(1461,620)(1464,621)(1466,626)(1469,628)(1471,632)(1474,634)(1476,633)(1499,627)(1520,622)(1539,633)(1557,634)(1573,629)(1588,629)(1603,623)(1617,624)(1630,629)(1642,632)(1654,651)(1665,643)(1675,642)(1685,638)(1695,636)(1705,637)(1714,635)(1722,631)(1731,628)(1739,626)
\thicklines \path(1739,626)(1747,623)(1754,628)(1762,626)(1769,625)(1776,629)(1783,631)(1789,630)(1795,628)(1802,630)(1808,631)(1814,629)(1819,633)(1825,634)(1831,636)(1836,635)(1841,638)(1846,640)(1851,639)(1856,637)(1861,639)(1866,638)(1871,637)(1875,635)(1880,634)(1884,637)(1888,637)(1892,637)(1897,637)(1901,639)(1905,641)(1909,640)(1913,638)(1916,638)(1920,638)(1924,641)(1928,640)(1931,640)(1935,642)(1938,642)(1942,643)(1945,643)(1948,644)(1952,645)(1955,648)(1958,647)(1961,647)(1964,646)(1968,646)(1971,647)(1974,647)
\thicklines \path(1974,647)(1977,648)(1980,648)(1982,647)(1985,647)(1988,646)(1991,646)(1994,645)(1996,645)(1999,644)(2002,644)(2004,644)(2007,643)(2010,642)(2012,644)(2015,645)(2017,644)(2020,645)(2022,645)(2025,645)(2027,645)(2027,645)
\thicklines \path(374,990)(395,990)(416,990)(435,990)(453,990)(470,990)(485,990)(500,990)(514,990)(527,990)(539,990)(551,990)(562,990)(572,990)(583,990)(592,990)(602,990)(611,990)(620,990)(628,990)(636,990)(644,990)(652,990)(659,990)(666,990)(673,990)(680,990)(687,990)(693,990)(699,990)(705,990)(711,990)(717,990)(723,990)(728,990)(733,990)(739,990)(744,990)(749,990)(754,990)(759,990)(763,990)(768,990)(773,990)(777,990)(782,990)(786,990)(790,990)(794,990)(798,990)(802,990)
\thicklines \path(802,990)(806,990)(810,990)(814,990)(818,990)(822,990)(825,990)(829,990)(832,990)(836,990)(839,990)(843,990)(846,990)(849,990)(853,990)(856,990)(859,990)(862,990)(865,990)(868,990)(871,990)(874,990)(877,990)(880,990)(883,990)(886,990)(889,990)(891,990)(894,990)(897,990)(900,990)(902,990)(905,990)(907,990)(910,990)(913,990)(915,990)(918,990)(920,990)(922,990)(925,990)(948,990)(968,990)(988,990)(1005,990)(1022,990)(1037,990)(1052,990)(1066,990)(1078,990)(1091,990)
\thicklines \path(1091,990)(1102,990)(1114,990)(1124,990)(1134,990)(1144,990)(1154,990)(1163,990)(1171,990)(1180,990)(1188,990)(1196,990)(1203,990)(1211,990)(1218,990)(1225,990)(1231,990)(1238,990)(1244,990)(1251,990)(1257,990)(1263,990)(1268,990)(1274,990)(1279,990)(1285,990)(1290,990)(1295,990)(1300,990)(1305,990)(1310,990)(1315,990)(1319,990)(1324,990)(1329,990)(1333,990)(1337,990)(1341,990)(1346,990)(1350,990)(1354,990)(1358,990)(1362,990)(1365,990)(1369,990)(1373,990)(1377,990)(1380,990)(1384,990)(1387,990)(1391,990)
\thicklines \path(1391,990)(1394,990)(1397,990)(1401,990)(1404,990)(1407,990)(1410,990)(1413,990)(1417,990)(1420,990)(1423,990)(1426,990)(1428,990)(1431,990)(1434,990)(1437,990)(1440,990)(1443,990)(1445,990)(1448,990)(1451,990)(1453,990)(1456,990)(1459,990)(1461,990)(1464,990)(1466,990)(1469,990)(1471,990)(1474,990)(1476,990)(1499,990)(1520,990)(1539,990)(1557,990)(1573,990)(1588,990)(1603,990)(1617,990)(1630,990)(1642,990)(1654,990)(1665,990)(1675,990)(1685,990)(1695,990)(1705,990)(1714,990)(1722,990)(1731,990)(1739,990)
\thicklines \path(1739,990)(1747,990)(1754,990)(1762,990)(1769,990)(1776,990)(1783,990)(1789,990)(1795,990)(1802,990)(1808,990)(1814,990)(1819,990)(1825,990)(1831,990)(1836,990)(1841,990)(1846,990)(1851,990)(1856,990)(1861,990)(1866,990)(1871,990)(1875,990)(1880,990)(1884,990)(1888,990)(1892,990)(1897,990)(1901,990)(1905,990)(1909,990)(1913,990)(1916,990)(1920,990)(1924,990)(1928,990)(1931,990)(1935,990)(1938,990)(1942,990)(1945,990)(1948,990)(1952,990)(1955,990)(1958,990)(1961,990)(1964,990)(1968,990)(1971,990)(1974,990)
\thicklines \path(1974,990)(1977,990)(1980,990)(1982,990)(1985,990)(1988,990)(1991,990)(1994,990)(1996,990)(1999,990)(2002,990)(2004,990)(2007,990)(2010,990)(2012,990)(2015,990)(2017,990)(2020,990)(2022,990)(2025,990)(2027,990)(2027,990)
\thicklines \path(374,310)(395,334)(416,339)(435,354)(453,342)(470,271)(485,228)(500,261)(514,254)(527,316)(539,334)(551,344)(562,381)(573,413)(583,470)(593,480)(602,522)(611,500)(620,503)(628,523)(636,534)(644,509)(652,507)(659,502)(666,416)(673,386)(680,398)(687,394)(693,423)(699,421)(705,447)(711,455)(717,457)(723,446)(728,451)(734,459)(739,463)(744,460)(749,445)(754,426)(759,442)(764,439)(768,434)(773,445)(777,458)(782,466)(786,464)(790,459)(794,457)(798,450)(802,455)
\thicklines \path(802,455)(806,448)(810,465)(814,478)(818,473)(822,469)(825,467)(829,473)(832,468)(836,465)(839,486)(843,520)(846,534)(849,538)(853,529)(856,530)(859,517)(862,521)(865,519)(868,512)(871,516)(874,515)(877,513)(880,516)(883,529)(886,520)(889,517)(891,516)(894,510)(897,510)(900,511)(902,511)(905,510)(907,512)(910,505)(913,500)(915,499)(918,498)(920,494)(922,491)(925,490)(948,505)(968,500)(988,524)(1005,534)(1022,556)(1037,561)(1052,579)(1066,569)(1079,560)(1091,566)
\thicklines \path(1091,566)(1102,555)(1114,548)(1124,541)(1134,553)(1144,545)(1154,547)(1163,558)(1171,554)(1180,553)(1188,552)(1196,557)(1203,554)(1211,542)(1218,547)(1225,551)(1231,572)(1238,583)(1244,600)(1251,599)(1257,606)(1263,601)(1268,600)(1274,606)(1280,615)(1285,619)(1290,614)(1295,613)(1300,618)(1305,620)(1310,626)(1315,627)(1319,631)(1324,652)(1329,647)(1333,647)(1337,646)(1341,641)(1346,638)(1350,641)(1354,639)(1358,638)(1362,637)(1365,637)(1369,638)(1373,634)(1377,634)(1380,633)(1384,632)(1387,640)(1391,639)
\thicklines \path(1391,639)(1394,640)(1397,636)(1401,636)(1404,634)(1407,632)(1410,635)(1413,633)(1417,632)(1420,632)(1423,630)(1426,626)(1428,625)(1431,624)(1434,622)(1437,626)(1440,624)(1443,624)(1445,628)(1448,627)(1451,628)(1453,626)(1456,624)(1459,623)(1461,623)(1464,620)(1466,625)(1469,623)(1471,621)(1474,622)(1476,621)(1499,620)(1520,610)(1539,610)(1557,611)(1573,634)(1588,634)(1603,634)(1617,630)(1630,629)(1642,631)(1654,635)(1665,633)(1675,640)(1685,636)(1695,635)(1705,633)(1714,639)(1722,639)(1731,638)(1739,633)
\thicklines \path(1739,633)(1747,634)(1754,633)(1762,629)(1769,628)(1776,626)(1783,630)(1789,630)(1795,631)(1802,634)(1808,633)(1814,631)(1819,634)(1825,638)(1831,641)(1836,640)(1841,640)(1846,640)(1851,645)(1856,646)(1861,644)(1866,643)(1871,647)(1875,647)(1880,645)(1884,646)(1888,645)(1892,645)(1897,649)(1901,647)(1905,648)(1909,647)(1913,647)(1916,648)(1920,646)(1924,647)(1928,647)(1931,647)(1935,648)(1938,648)(1942,649)(1945,653)(1948,653)(1952,653)(1955,653)(1958,654)(1961,653)(1964,652)(1968,652)(1971,651)(1974,652)
\thicklines \path(1974,652)(1977,650)(1980,651)(1982,653)(1985,654)(1988,654)(1991,654)(1994,655)(1996,655)(1999,657)(2002,657)(2004,657)(2007,656)(2010,655)(2012,654)(2015,653)(2017,653)(2020,653)(2022,654)(2025,653)(2027,652)(2027,652)
\color{black}
\put(1884,363){\makebox(0,0)[r]{FS ($m = 1000$)}}
\color{blue}
\thinlines \path(1913,363)(2070,363)
\thinlines \path(374,575)(374,575)(397,589)(418,577)(437,584)(455,600)(471,602)(486,617)(501,622)(515,625)(528,622)(540,650)(552,626)(563,606)(573,625)(583,627)(593,644)(603,650)(612,647)(620,642)(629,636)(637,630)(645,649)(652,643)(660,626)(667,642)(674,651)(681,639)(687,630)(693,644)(700,633)(706,639)(712,628)(717,627)(723,625)(729,629)(734,629)(739,631)(744,621)(749,615)(754,619)(759,629)(764,631)(769,634)(773,629)(778,625)(782,636)(786,628)(790,625)(795,622)(799,619)
\thinlines \path(799,619)(803,607)(807,609)(811,602)(814,604)(818,599)(822,605)(826,597)(829,608)(833,610)(836,603)(840,608)(843,608)(846,619)(850,616)(853,612)(856,615)(859,614)(862,616)(866,615)(869,613)(872,610)(875,602)(878,610)(880,605)(883,607)(886,611)(889,611)(892,621)(894,622)(897,621)(900,620)(902,621)(905,625)(908,622)(910,621)(913,621)(915,617)(918,619)(920,616)(923,616)(925,622)(948,630)(969,618)(988,621)(1006,621)(1022,621)(1037,631)(1052,631)(1066,625)(1079,630)
\thinlines \path(1079,630)(1091,627)(1103,633)(1114,635)(1124,633)(1134,631)(1144,633)(1154,627)(1163,624)(1171,614)(1180,618)(1188,616)(1196,624)(1203,631)(1211,634)(1218,637)(1225,634)(1232,634)(1238,637)(1244,636)(1251,636)(1257,633)(1263,637)(1268,636)(1274,636)(1280,634)(1285,634)(1290,634)(1295,636)(1300,639)(1305,637)(1310,637)(1315,637)(1320,633)(1324,631)(1329,635)(1333,639)(1337,641)(1341,644)(1346,644)(1350,645)(1354,647)(1358,650)(1362,647)(1365,647)(1369,651)(1373,652)(1377,655)(1380,654)(1384,652)(1387,651)
\thinlines \path(1387,651)(1391,648)(1394,651)(1397,651)(1401,653)(1404,652)(1407,653)(1410,653)(1413,654)(1417,656)(1420,653)(1423,659)(1426,661)(1429,661)(1431,662)(1434,661)(1437,661)(1440,663)(1443,663)(1445,665)(1448,665)(1451,665)(1453,664)(1456,664)(1459,665)(1461,664)(1464,665)(1466,665)(1469,665)(1471,664)(1474,662)(1476,662)(1499,658)(1520,664)(1539,667)(1557,667)(1573,670)(1588,670)(1603,669)(1617,669)(1630,668)(1642,670)(1654,670)(1665,672)(1675,674)(1685,673)(1695,673)(1705,674)(1714,673)(1722,675)(1731,674)
\thinlines \path(1731,674)(1739,674)(1747,674)(1754,674)(1762,676)(1769,677)(1776,679)(1783,677)(1789,678)(1795,679)(1802,681)(1808,680)(1814,679)(1819,678)(1825,679)(1831,680)(1836,679)(1841,679)(1846,680)(1851,680)(1856,680)(1861,681)(1866,681)(1871,681)(1875,681)(1880,681)(1884,682)(1888,682)(1892,682)(1897,682)(1901,681)(1905,681)(1909,680)(1913,680)(1916,679)(1920,679)(1924,679)(1928,680)(1931,680)(1935,679)(1938,679)(1942,679)(1945,680)(1948,680)(1952,679)(1955,680)(1958,682)(1961,683)(1964,683)(1968,683)(1971,684)
\thinlines \path(1971,684)(1974,683)(1977,684)(1980,684)(1982,684)(1985,684)(1988,684)(1991,684)(1994,685)(1996,687)(1999,687)(2002,687)(2004,686)(2007,687)(2010,687)(2012,686)(2015,687)(2017,687)(2020,688)(2022,688)(2025,688)(2027,688)(2027,688)
\thinlines \path(374,855)(395,817)(416,799)(435,789)(453,780)(470,782)(485,751)(500,771)(513,782)(526,819)(539,822)(551,820)(562,823)(572,809)(583,813)(592,796)(602,781)(611,786)(620,814)(628,818)(636,800)(644,792)(652,789)(659,791)(666,809)(673,802)(680,797)(686,785)(693,775)(699,781)(705,777)(711,781)(717,774)(723,777)(728,769)(733,767)(739,754)(744,745)(749,744)(754,741)(759,745)(763,740)(768,734)(773,730)(777,736)(782,745)(786,751)(790,747)(794,753)(798,749)(802,748)
\thinlines \path(802,748)(806,744)(810,742)(814,743)(818,735)(822,737)(825,733)(829,731)(832,721)(836,723)(839,722)(843,719)(846,719)(849,713)(853,704)(856,698)(859,693)(862,702)(865,711)(868,708)(871,699)(874,699)(877,700)(880,701)(883,704)(886,702)(889,704)(891,700)(894,702)(897,699)(900,700)(902,697)(905,701)(907,703)(910,699)(912,699)(915,702)(917,707)(920,707)(922,705)(925,707)(948,699)(968,686)(988,673)(1005,679)(1022,688)(1037,688)(1052,670)(1066,666)(1078,672)(1091,675)
\thinlines \path(1091,675)(1102,676)(1114,668)(1124,672)(1134,677)(1144,667)(1154,661)(1163,659)(1171,662)(1180,661)(1188,656)(1196,653)(1203,656)(1211,654)(1218,653)(1225,657)(1231,661)(1238,666)(1244,664)(1251,661)(1257,659)(1263,657)(1268,655)(1274,658)(1279,657)(1285,660)(1290,657)(1295,663)(1300,663)(1305,663)(1310,665)(1315,661)(1319,661)(1324,659)(1329,659)(1333,662)(1337,660)(1341,659)(1346,658)(1350,657)(1354,657)(1358,658)(1362,657)(1365,659)(1369,661)(1373,659)(1377,661)(1380,659)(1384,662)(1387,661)(1391,662)
\thinlines \path(1391,662)(1394,659)(1397,659)(1401,659)(1404,661)(1407,659)(1410,661)(1413,659)(1417,661)(1420,661)(1423,662)(1426,662)(1428,663)(1431,662)(1434,662)(1437,662)(1440,663)(1443,664)(1445,664)(1448,665)(1451,668)(1453,668)(1456,669)(1459,670)(1461,673)(1464,672)(1466,672)(1469,673)(1471,673)(1474,673)(1476,676)(1499,675)(1520,679)(1539,677)(1557,676)(1573,676)(1588,676)(1603,672)(1617,674)(1630,675)(1642,674)(1654,676)(1665,675)(1675,674)(1685,674)(1695,676)(1705,674)(1714,673)(1722,672)(1731,670)(1739,671)
\thinlines \path(1739,671)(1747,671)(1754,672)(1762,674)(1769,673)(1776,675)(1783,674)(1789,674)(1795,674)(1802,676)(1808,675)(1814,676)(1819,677)(1825,677)(1831,678)(1836,679)(1841,680)(1846,680)(1851,678)(1856,678)(1861,678)(1866,679)(1871,678)(1875,679)(1880,680)(1884,679)(1888,680)(1892,678)(1897,678)(1901,679)(1905,680)(1909,680)(1913,680)(1916,681)(1920,682)(1924,680)(1928,680)(1931,681)(1935,680)(1938,680)(1942,681)(1945,680)(1948,680)(1952,680)(1955,680)(1958,679)(1961,679)(1964,680)(1968,680)(1971,680)(1974,680)
\thinlines \path(1974,680)(1977,680)(1980,679)(1982,679)(1985,680)(1988,679)(1991,679)(1994,679)(1996,679)(1999,678)(2002,679)(2004,680)(2007,681)(2010,681)(2012,681)(2015,681)(2017,682)(2020,682)(2022,682)(2025,681)(2027,681)(2027,681)
\thinlines \path(374,766)(395,743)(416,751)(435,782)(453,758)(470,733)(485,718)(500,752)(514,768)(527,783)(539,770)(551,749)(562,723)(572,726)(583,729)(592,737)(602,757)(611,783)(620,806)(628,822)(636,801)(644,773)(652,790)(659,790)(666,784)(673,807)(680,815)(687,826)(693,827)(699,816)(705,817)(711,846)(717,842)(723,829)(728,835)(733,827)(739,820)(744,821)(749,816)(754,813)(759,811)(763,818)(768,814)(773,818)(777,814)(782,815)(786,819)(790,817)(794,819)(798,807)(802,803)
\thinlines \path(802,803)(806,803)(810,804)(814,804)(818,807)(822,807)(825,796)(829,789)(832,780)(836,776)(839,776)(843,780)(846,779)(849,784)(853,783)(856,783)(859,776)(862,779)(865,772)(868,775)(871,775)(874,769)(877,765)(880,763)(883,765)(886,762)(889,756)(891,758)(894,754)(897,752)(900,749)(902,749)(905,756)(907,748)(910,742)(913,741)(915,735)(918,741)(920,742)(922,743)(925,742)(948,717)(968,703)(988,686)(1005,680)(1022,684)(1037,683)(1052,668)(1066,665)(1078,672)(1091,669)
\thinlines \path(1091,669)(1102,665)(1114,659)(1124,669)(1134,668)(1144,670)(1154,672)(1163,670)(1171,667)(1180,665)(1188,666)(1196,660)(1203,655)(1211,658)(1218,654)(1225,649)(1231,650)(1238,647)(1244,645)(1251,643)(1257,640)(1263,643)(1268,643)(1274,640)(1279,640)(1285,645)(1290,644)(1295,649)(1300,650)(1305,650)(1310,647)(1315,649)(1319,651)(1324,647)(1329,651)(1333,652)(1337,655)(1341,652)(1346,653)(1350,657)(1354,657)(1358,661)(1362,661)(1365,664)(1369,661)(1373,663)(1377,667)(1380,668)(1384,667)(1387,668)(1391,667)
\thinlines \path(1391,667)(1394,665)(1397,662)(1401,662)(1404,662)(1407,665)(1410,667)(1413,667)(1417,666)(1420,666)(1423,666)(1426,662)(1428,662)(1431,661)(1434,660)(1437,661)(1440,659)(1443,660)(1445,657)(1448,656)(1451,654)(1453,656)(1456,655)(1459,655)(1461,654)(1464,655)(1466,655)(1469,655)(1471,655)(1474,656)(1476,656)(1499,662)(1520,661)(1539,653)(1557,651)(1573,655)(1588,658)(1603,654)(1617,651)(1630,654)(1642,653)(1654,655)(1665,653)(1675,655)(1685,657)(1695,661)(1705,663)(1714,664)(1722,664)(1731,665)(1739,663)
\thinlines \path(1739,663)(1747,663)(1754,664)(1762,665)(1769,668)(1776,666)(1783,665)(1789,667)(1795,667)(1802,667)(1808,669)(1814,668)(1819,669)(1825,671)(1831,672)(1836,674)(1841,675)(1846,678)(1851,678)(1856,679)(1861,677)(1866,677)(1871,677)(1875,676)(1880,676)(1884,675)(1888,674)(1892,674)(1897,674)(1901,674)(1905,673)(1909,672)(1913,673)(1916,673)(1920,674)(1924,674)(1928,673)(1931,673)(1935,673)(1938,674)(1942,674)(1945,675)(1948,675)(1952,674)(1955,675)(1958,675)(1961,675)(1964,675)(1968,675)(1971,675)(1974,676)
\thinlines \path(1974,676)(1977,676)(1980,676)(1982,676)(1985,677)(1988,677)(1991,677)(1994,677)(1996,676)(1999,677)(2002,677)(2004,677)(2007,677)(2010,678)(2012,678)(2015,678)(2017,678)(2020,677)(2022,677)(2025,677)(2027,677)(2027,677)
\thinlines \path(374,842)(395,839)(416,793)(435,808)(453,790)(470,757)(485,755)(500,749)(514,753)(527,779)(539,798)(551,808)(562,842)(573,854)(583,848)(593,832)(602,823)(611,823)(620,811)(628,824)(636,809)(644,801)(652,793)(659,797)(666,807)(673,821)(680,807)(687,808)(693,807)(699,808)(705,802)(711,807)(717,810)(723,811)(728,815)(734,803)(739,798)(744,796)(749,797)(754,795)(759,792)(764,789)(768,801)(773,802)(777,794)(782,791)(786,793)(790,781)(794,779)(798,776)(802,789)
\thinlines \path(802,789)(806,789)(810,791)(814,784)(818,782)(822,771)(825,769)(829,775)(832,770)(836,777)(839,778)(843,772)(846,771)(849,768)(853,770)(856,770)(859,775)(862,780)(865,777)(868,785)(871,783)(874,783)(877,777)(880,769)(883,770)(886,769)(889,773)(891,777)(894,772)(897,772)(900,773)(902,775)(905,776)(907,781)(910,781)(913,774)(915,766)(918,768)(920,768)(922,766)(925,766)(948,744)(968,731)(988,714)(1005,708)(1022,709)(1037,711)(1052,716)(1066,713)(1079,707)(1091,707)
\thinlines \path(1091,707)(1102,711)(1114,713)(1124,706)(1134,696)(1144,693)(1154,697)(1163,701)(1171,701)(1180,702)(1188,702)(1196,696)(1203,700)(1211,701)(1218,693)(1225,690)(1231,690)(1238,694)(1244,692)(1251,688)(1257,684)(1263,686)(1268,690)(1274,688)(1280,690)(1285,687)(1290,691)(1295,686)(1300,685)(1305,687)(1310,687)(1315,692)(1319,688)(1324,687)(1329,686)(1333,686)(1337,688)(1341,686)(1346,686)(1350,686)(1354,686)(1358,685)(1362,684)(1365,686)(1369,686)(1373,688)(1377,685)(1380,687)(1384,688)(1387,687)(1391,688)
\thinlines \path(1391,688)(1394,689)(1397,689)(1401,688)(1404,687)(1407,688)(1410,689)(1413,691)(1417,691)(1420,690)(1423,693)(1426,691)(1428,690)(1431,687)(1434,686)(1437,683)(1440,685)(1443,686)(1445,685)(1448,685)(1451,685)(1453,684)(1456,684)(1459,684)(1461,684)(1464,686)(1466,687)(1469,687)(1471,687)(1474,685)(1476,686)(1499,683)(1520,676)(1539,679)(1557,679)(1573,682)(1588,677)(1603,674)(1617,673)(1630,672)(1642,670)(1654,668)(1665,671)(1675,671)(1685,673)(1695,674)(1705,674)(1714,672)(1722,671)(1731,674)(1739,673)
\thinlines \path(1739,673)(1747,671)(1754,670)(1762,670)(1769,671)(1776,670)(1783,668)(1789,667)(1795,667)(1802,667)(1808,666)(1814,664)(1819,663)(1825,663)(1831,665)(1836,666)(1841,666)(1846,667)(1851,666)(1856,666)(1861,667)(1866,667)(1871,666)(1875,667)(1880,668)(1884,668)(1888,668)(1892,669)(1897,669)(1901,669)(1905,670)(1909,671)(1913,671)(1916,671)(1920,671)(1924,673)(1928,673)(1931,673)(1935,674)(1938,674)(1942,674)(1945,673)(1948,674)(1952,674)(1955,674)(1958,674)(1961,675)(1964,676)(1968,676)(1971,676)(1974,676)
\thinlines \path(1974,676)(1977,676)(1980,676)(1982,677)(1985,677)(1988,676)(1991,677)(1994,677)(1996,677)(1999,678)(2002,678)(2004,678)(2007,678)(2010,679)(2012,679)(2015,678)(2017,678)(2020,679)(2022,679)(2025,679)(2027,679)(2027,679)
\color{black}
\put(1884,305){\makebox(0,0)[r]{MultipleRW ($m=1000$)}}
\Thicklines \path(1913,305)(2070,305)
\Thicklines \path(374,394)(374,394)(397,547)(418,656)(437,726)(455,753)(471,784)(486,825)(501,850)(515,892)(528,895)(540,920)(552,912)(563,903)(573,893)(583,894)(593,876)(603,865)(612,872)(620,867)(629,855)(637,851)(645,869)(652,890)(660,910)(667,914)(674,919)(681,935)(687,941)(693,952)(700,959)(706,967)(712,959)(717,951)(723,945)(729,940)(734,930)(739,925)(744,929)(749,925)(754,922)(759,923)(764,931)(769,940)(773,947)(778,946)(782,951)(786,955)(790,960)(795,963)(799,966)
\Thicklines \path(799,966)(803,976)(807,973)(811,970)(814,968)(818,966)(822,964)(826,962)(829,964)(833,964)(836,961)(840,960)(843,968)(846,973)(850,976)(853,977)(856,981)(859,982)(862,987)(866,992)(869,994)(872,999)(875,997)(878,993)(880,989)(883,987)(886,984)(889,983)(892,982)(894,981)(897,977)(900,976)(902,980)(905,987)(908,989)(910,988)(913,989)(915,987)(918,995)(920,1000)(923,1001)(925,1005)(948,988)(969,1007)(988,999)(1006,1014)(1022,1002)(1037,1012)(1052,1004)(1066,1015)(1079,1008)
\Thicklines \path(1079,1008)(1091,1015)(1103,1010)(1114,1015)(1124,1012)(1134,1016)(1144,1012)(1154,1015)(1163,1011)(1171,1013)(1180,1011)(1188,1013)(1196,1009)(1203,1008)(1211,1005)(1218,1007)(1225,1003)(1232,1005)(1238,1004)(1244,1006)(1251,1003)(1257,1005)(1263,1003)(1268,1006)(1274,1003)(1280,1006)(1285,1004)(1290,1006)(1295,1005)(1300,1008)(1305,1005)(1310,1007)(1315,1004)(1320,1007)(1324,1006)(1329,1007)(1333,1007)(1337,1007)(1341,1007)(1346,1007)(1350,1006)(1354,1006)(1358,1005)(1362,1006)(1365,1004)(1369,1007)(1373,1005)(1377,1005)(1380,1003)(1384,1002)(1387,999)
\Thicklines \path(1387,999)(1391,1000)(1394,997)(1397,998)(1401,997)(1404,999)(1407,996)(1410,998)(1413,997)(1417,998)(1420,997)(1423,997)(1426,997)(1429,998)(1431,997)(1434,998)(1437,997)(1440,997)(1443,997)(1445,997)(1448,996)(1451,997)(1453,995)(1456,997)(1459,996)(1461,997)(1464,996)(1466,998)(1469,996)(1471,998)(1474,997)(1476,999)(1499,994)(1520,995)(1539,993)(1557,991)(1573,992)(1588,993)(1603,993)(1617,991)(1630,989)(1642,988)(1654,987)(1665,986)(1675,984)(1685,983)(1695,981)(1705,978)(1714,976)(1722,974)(1731,972)
\Thicklines \path(1731,972)(1739,971)(1747,971)(1754,971)(1762,970)(1769,970)(1776,970)(1783,969)(1789,970)(1795,968)(1802,967)(1808,967)(1814,965)(1819,966)(1825,965)(1831,964)(1836,964)(1841,963)(1846,962)(1851,962)(1856,962)(1861,961)(1866,961)(1871,961)(1875,961)(1880,961)(1884,960)(1888,959)(1892,958)(1897,958)(1901,957)(1905,956)(1909,956)(1913,955)(1916,955)(1920,955)(1924,955)(1928,954)(1931,954)(1935,954)(1938,954)(1942,953)(1945,953)(1948,953)(1952,953)(1955,953)(1958,952)(1961,952)(1964,951)(1968,951)(1971,951)
\Thicklines \path(1971,951)(1974,951)(1977,951)(1980,950)(1982,950)(1985,949)(1988,949)(1991,949)(1994,949)(1996,949)(1999,948)(2002,948)(2004,948)(2007,948)(2010,948)(2012,948)(2015,948)(2017,948)(2020,948)(2022,948)(2025,948)(2027,948)(2027,948)
\Thicklines \path(374,446)(374,446)(397,613)(418,691)(437,739)(455,813)(471,833)(486,849)(501,865)(515,889)(528,904)(540,941)(552,925)(563,926)(573,923)(583,916)(593,913)(603,907)(612,897)(620,898)(629,887)(637,873)(645,896)(652,906)(660,918)(667,940)(674,939)(681,950)(687,953)(693,959)(700,964)(706,973)(712,963)(717,960)(723,958)(729,954)(734,956)(739,954)(744,949)(749,946)(754,943)(759,939)(764,946)(769,950)(773,957)(778,965)(782,967)(786,971)(790,971)(795,974)(799,979)
\Thicklines \path(799,979)(803,983)(807,978)(811,975)(814,973)(818,973)(822,973)(826,971)(829,967)(833,963)(836,959)(840,958)(843,957)(846,959)(850,961)(853,963)(856,966)(859,971)(862,972)(866,973)(869,979)(872,982)(875,978)(878,976)(880,974)(883,973)(886,972)(889,969)(892,967)(894,964)(897,963)(900,965)(902,964)(905,967)(908,968)(910,971)(913,973)(915,976)(918,977)(920,976)(923,980)(925,981)(948,971)(969,984)(988,978)(1006,988)(1022,976)(1037,981)(1052,977)(1066,982)(1079,974)
\Thicklines \path(1079,974)(1091,980)(1103,973)(1114,977)(1124,974)(1134,979)(1144,974)(1154,978)(1163,975)(1171,975)(1180,973)(1188,976)(1196,974)(1203,977)(1211,974)(1218,974)(1225,971)(1232,973)(1238,972)(1244,972)(1251,968)(1257,970)(1263,968)(1268,970)(1274,970)(1280,970)(1285,970)(1290,973)(1295,972)(1300,972)(1305,971)(1310,971)(1315,972)(1320,972)(1324,970)(1329,971)(1333,969)(1337,968)(1341,967)(1346,968)(1350,967)(1354,968)(1358,968)(1362,969)(1365,968)(1369,969)(1373,967)(1377,967)(1380,964)(1384,965)(1387,964)
\Thicklines \path(1387,964)(1391,965)(1394,963)(1397,962)(1401,960)(1404,961)(1407,959)(1410,959)(1413,956)(1417,956)(1420,953)(1423,955)(1426,953)(1429,954)(1431,950)(1434,952)(1437,951)(1440,952)(1443,950)(1445,950)(1448,949)(1451,950)(1453,949)(1456,950)(1459,948)(1461,948)(1464,947)(1466,948)(1469,947)(1471,947)(1474,945)(1476,945)(1499,937)(1520,932)(1539,932)(1557,929)(1573,929)(1588,931)(1603,931)(1617,931)(1630,931)(1642,930)(1654,929)(1665,930)(1675,929)(1685,926)(1695,927)(1705,925)(1714,922)(1722,920)(1731,919)
\Thicklines \path(1731,919)(1739,917)(1747,916)(1754,915)(1762,914)(1769,915)(1776,915)(1783,915)(1789,914)(1795,915)(1802,914)(1808,916)(1814,916)(1819,916)(1825,916)(1831,916)(1836,915)(1841,915)(1846,915)(1851,915)(1856,915)(1861,916)(1866,916)(1871,915)(1875,915)(1880,915)(1884,914)(1888,913)(1892,912)(1897,911)(1901,910)(1905,910)(1909,910)(1913,909)(1916,908)(1920,907)(1924,907)(1928,907)(1931,906)(1935,906)(1938,905)(1942,904)(1945,904)(1948,903)(1952,903)(1955,903)(1958,903)(1961,903)(1964,903)(1968,902)(1971,902)
\Thicklines \path(1971,902)(1974,903)(1977,902)(1980,902)(1982,902)(1985,902)(1988,902)(1991,901)(1994,901)(1996,900)(1999,900)(2002,900)(2004,899)(2007,899)(2010,898)(2012,898)(2015,898)(2017,898)(2020,898)(2022,897)(2025,897)(2027,897)(2027,897)
\Thicklines \path(374,415)(374,415)(397,553)(418,651)(437,701)(455,775)(471,843)(486,881)(501,901)(515,943)(528,969)(540,1001)(552,989)(563,969)(573,958)(583,951)(593,939)(603,919)(612,907)(620,897)(629,893)(637,884)(645,900)(652,910)(660,915)(667,928)(674,943)(681,958)(687,964)(693,973)(700,977)(706,988)(712,981)(717,973)(723,970)(729,964)(734,963)(739,959)(744,954)(749,950)(754,947)(759,948)(764,953)(769,964)(773,964)(778,969)(782,978)(786,979)(790,981)(795,986)(799,986)
\Thicklines \path(799,986)(803,992)(807,990)(811,986)(814,983)(818,983)(822,980)(826,980)(829,975)(833,972)(836,971)(840,969)(843,974)(846,977)(850,979)(853,984)(856,992)(859,993)(862,998)(866,998)(869,997)(872,1001)(875,1001)(878,998)(880,995)(883,994)(886,992)(889,991)(892,986)(894,983)(897,985)(900,983)(902,985)(905,991)(908,992)(910,994)(913,998)(915,998)(918,1002)(920,1003)(923,1004)(925,1007)(948,998)(969,1010)(988,1000)(1006,1005)(1022,997)(1037,1001)(1052,995)(1066,999)(1079,996)
\Thicklines \path(1079,996)(1091,1000)(1103,994)(1114,996)(1124,992)(1134,992)(1144,985)(1154,986)(1163,982)(1171,985)(1180,980)(1188,981)(1196,976)(1203,977)(1211,973)(1218,976)(1225,972)(1232,971)(1238,966)(1244,966)(1251,962)(1257,963)(1263,961)(1268,961)(1274,960)(1280,960)(1285,959)(1290,957)(1295,955)(1300,955)(1305,955)(1310,955)(1315,953)(1320,953)(1324,951)(1329,952)(1333,950)(1337,950)(1341,950)(1346,948)(1350,949)(1354,948)(1358,947)(1362,949)(1365,948)(1369,948)(1373,947)(1377,947)(1380,947)(1384,947)(1387,946)
\Thicklines \path(1387,946)(1391,945)(1394,944)(1397,946)(1401,945)(1404,946)(1407,944)(1410,945)(1413,943)(1417,945)(1420,944)(1423,944)(1426,942)(1429,941)(1431,940)(1434,940)(1437,938)(1440,939)(1443,937)(1445,938)(1448,937)(1451,937)(1453,936)(1456,936)(1459,936)(1461,936)(1464,936)(1466,936)(1469,936)(1471,936)(1474,935)(1476,935)(1499,931)(1520,926)(1539,922)(1557,920)(1573,919)(1588,917)(1603,914)(1617,913)(1630,911)(1642,911)(1654,909)(1665,907)(1675,906)(1685,904)(1695,902)(1705,903)(1714,901)(1722,900)(1731,899)
\Thicklines \path(1731,899)(1739,900)(1747,901)(1754,900)(1762,899)(1769,899)(1776,897)(1783,896)(1789,895)(1795,894)(1802,894)(1808,892)(1814,893)(1819,892)(1825,893)(1831,892)(1836,892)(1841,892)(1846,892)(1851,891)(1856,891)(1861,889)(1866,889)(1871,889)(1875,888)(1880,887)(1884,887)(1888,888)(1892,888)(1897,887)(1901,887)(1905,887)(1909,887)(1913,887)(1916,887)(1920,886)(1924,886)(1928,885)(1931,885)(1935,885)(1938,884)(1942,883)(1945,883)(1948,883)(1952,883)(1955,883)(1958,884)(1961,884)(1964,884)(1968,884)(1971,885)
\Thicklines \path(1971,885)(1974,884)(1977,884)(1980,883)(1982,883)(1985,883)(1988,883)(1991,883)(1994,883)(1996,883)(1999,883)(2002,883)(2004,883)(2007,883)(2010,882)(2012,881)(2015,881)(2017,880)(2020,880)(2022,879)(2025,879)(2027,879)(2027,879)
\Thicklines \path(374,463)(374,463)(397,600)(418,704)(437,787)(455,823)(471,867)(486,906)(501,922)(515,967)(528,982)(540,1015)(552,1002)(563,1001)(573,986)(583,969)(593,956)(603,957)(612,951)(620,940)(629,932)(637,925)(645,939)(652,946)(660,952)(667,958)(674,966)(681,973)(687,979)(693,992)(700,999)(706,1005)(712,997)(717,998)(723,989)(729,990)(734,988)(739,988)(744,983)(749,973)(754,969)(759,963)(764,963)(769,964)(773,968)(778,970)(782,973)(786,975)(790,978)(795,981)(799,989)
\Thicklines \path(799,989)(803,1002)(807,999)(811,997)(814,994)(818,990)(822,986)(826,986)(829,985)(833,980)(836,979)(840,975)(843,981)(846,978)(850,981)(853,984)(856,985)(859,987)(862,989)(866,991)(869,995)(872,1001)(875,1001)(878,1002)(880,999)(883,996)(886,996)(889,995)(892,995)(894,993)(897,990)(900,985)(902,987)(905,988)(908,989)(910,992)(913,991)(915,994)(918,994)(920,998)(923,999)(925,1003)(948,990)(969,1002)(988,992)(1006,997)(1022,984)(1037,987)(1052,979)(1066,987)(1079,978)
\Thicklines \path(1079,978)(1091,985)(1103,980)(1114,990)(1124,984)(1134,992)(1144,988)(1154,994)(1163,989)(1171,993)(1180,986)(1188,988)(1196,986)(1203,991)(1211,988)(1218,993)(1225,988)(1232,992)(1238,989)(1244,995)(1251,993)(1257,996)(1263,992)(1268,995)(1274,995)(1280,997)(1285,996)(1290,999)(1295,998)(1300,1002)(1305,1001)(1310,1002)(1315,1000)(1320,1002)(1324,998)(1329,999)(1333,996)(1337,997)(1341,994)(1346,997)(1350,995)(1354,997)(1358,994)(1362,997)(1365,994)(1369,995)(1373,993)(1377,995)(1380,993)(1384,993)(1387,991)
\Thicklines \path(1387,991)(1391,993)(1394,991)(1397,991)(1401,990)(1404,991)(1407,988)(1410,989)(1413,987)(1417,987)(1420,985)(1423,985)(1426,984)(1429,986)(1431,984)(1434,984)(1437,982)(1440,982)(1443,980)(1445,980)(1448,980)(1451,980)(1453,980)(1456,980)(1459,979)(1461,979)(1464,978)(1466,979)(1469,977)(1471,977)(1474,976)(1476,976)(1499,974)(1520,973)(1539,966)(1557,963)(1573,961)(1588,959)(1603,956)(1617,956)(1630,955)(1642,953)(1654,950)(1665,949)(1675,948)(1685,947)(1695,947)(1705,945)(1714,945)(1722,946)(1731,945)
\Thicklines \path(1731,945)(1739,944)(1747,944)(1754,944)(1762,943)(1769,943)(1776,943)(1783,942)(1789,942)(1795,942)(1802,941)(1808,941)(1814,940)(1819,938)(1825,938)(1831,937)(1836,936)(1841,936)(1846,935)(1851,936)(1856,935)(1861,936)(1866,935)(1871,935)(1875,935)(1880,934)(1884,934)(1888,934)(1892,933)(1897,933)(1901,932)(1905,932)(1909,931)(1913,931)(1916,931)(1920,931)(1924,931)(1928,931)(1931,931)(1935,931)(1938,930)(1942,930)(1945,929)(1948,929)(1952,929)(1955,929)(1958,929)(1961,929)(1964,928)(1968,928)(1971,928)
\Thicklines \path(1971,928)(1974,927)(1977,927)(1980,927)(1982,926)(1985,926)(1988,925)(1991,925)(1994,924)(1996,924)(1999,924)(2002,924)(2004,923)(2007,923)(2010,924)(2012,923)(2015,923)(2017,923)(2020,923)(2022,923)(2025,923)(2027,922)(2027,922)
\thicklines \path(374,1022)(374,226)(2027,226)
\end{picture}

%% file: figs/graph_BA_4runs_deg10.tex
\setlength{\unitlength}{0.120450pt}
\begin{picture}(2100,1080)(0,0)
\fontsize{7}{8.98332}\selectfont
\thicklines \path(432,226)(391,226)
\put(362,226){\makebox(0,0)[r]{$0$}}
\thicklines \path(432,362)(391,362)
\put(362,362){\makebox(0,0)[r]{$0.012$}}
\thicklines \path(432,497)(391,497)
\put(362,497){\makebox(0,0)[r]{$0.024$}}
\thicklines \path(432,772)(391,772)
\put(362,772){\makebox(0,0)[r]{$0.048$}}
\thicklines \path(432,226)(432,185)
\put(432,127){\makebox(0,0){$1$}}
\thicklines \path(1157,226)(1157,185)
\put(1157,127){\makebox(0,0){$5\times 10^3$}}
\thicklines \path(1882,226)(1882,185)
\put(1882,127){\makebox(0,0){$10^4$}}
\thicklines \path(432,1022)(432,226)(2027,226)
\put(101,624){\makebox(0,0)[l]{\rotatebox[origin=c]{90}{$\hat{\theta}_{10}(n)$}}}
\put(1229,40){\makebox(0,0){Random walk steps (n)}}
\thinlines \path(432,497)(432,497)(432,497)(432,497)(432,497)(433,497)(433,497)(433,497)(433,497)(433,497)(433,497)(435,497)(436,497)(438,497)(439,497)(441,497)(442,497)(443,497)(445,497)(446,497)(448,497)(449,497)(451,497)(452,497)(454,497)(455,497)(457,497)(458,497)(459,497)(461,497)(462,497)(464,497)(465,497)(467,497)(468,497)(470,497)(471,497)(472,497)(474,497)(475,497)(477,497)(478,497)(480,497)(481,497)(483,497)(484,497)(486,497)(487,497)(488,497)(490,497)(491,497)
\thinlines \path(491,497)(493,497)(494,497)(496,497)(497,497)(499,497)(500,497)(501,497)(503,497)(504,497)(506,497)(507,497)(509,497)(510,497)(512,497)(513,497)(515,497)(516,497)(517,497)(519,497)(520,497)(522,497)(523,497)(525,497)(526,497)(528,497)(529,497)(530,497)(532,497)(533,497)(535,497)(536,497)(538,497)(539,497)(541,497)(542,497)(544,497)(545,497)(546,497)(548,497)(549,497)(551,497)(552,497)(554,497)(555,497)(557,497)(558,497)(559,497)(561,497)(562,497)(564,497)
\thinlines \path(564,497)(565,497)(567,497)(568,497)(570,497)(571,497)(573,497)(574,497)(575,497)(577,497)(591,497)(606,497)(620,497)(635,497)(649,497)(664,497)(678,497)(693,497)(707,497)(722,497)(736,497)(751,497)(765,497)(780,497)(794,497)(809,497)(823,497)(838,497)(852,497)(867,497)(881,497)(896,497)(910,497)(925,497)(939,497)(954,497)(968,497)(983,497)(997,497)(1012,497)(1026,497)(1041,497)(1055,497)(1070,497)(1084,497)(1099,497)(1113,497)(1128,497)(1142,497)(1157,497)(1171,497)
\thinlines \path(1171,497)(1186,497)(1200,497)(1215,497)(1229,497)(1244,497)(1258,497)(1273,497)(1287,497)(1302,497)(1316,497)(1331,497)(1345,497)(1360,497)(1374,497)(1389,497)(1403,497)(1418,497)(1432,497)(1447,497)(1461,497)(1476,497)(1490,497)(1505,497)(1519,497)(1534,497)(1548,497)(1563,497)(1577,497)(1592,497)(1606,497)(1621,497)(1635,497)(1650,497)(1664,497)(1679,497)(1693,497)(1708,497)(1722,497)(1737,497)(1751,497)(1766,497)(1780,497)(1795,497)(1809,497)(1824,497)(1838,497)(1853,497)(1867,497)(1882,497)(2027,497)
\thinlines \path(2027,497)(2027,497)
\put(1734,960){\makebox(0,0)[r]{SingleRW}}
\color{red}
\thicklines \path(1763,960)(1920,960)
\thicklines \path(432,226)(432,226)(432,226)(432,226)(432,226)(433,766)(433,663)(433,613)(433,557)(433,528)(433,503)(435,358)(436,304)(438,285)(439,316)(441,462)(442,429)(443,403)(445,385)(446,363)(448,351)(449,361)(451,370)(452,448)(454,429)(455,411)(457,396)(458,385)(459,374)(461,365)(462,379)(464,393)(465,385)(467,389)(468,382)(470,376)(471,370)(472,365)(474,360)(475,355)(477,351)(478,347)(480,344)(481,340)(483,338)(484,335)(486,332)(487,330)(488,327)(490,325)(491,324)
\thicklines \path(491,324)(493,321)(494,325)(496,343)(497,346)(499,343)(500,346)(501,344)(503,341)(504,338)(506,336)(507,343)(509,358)(510,365)(512,361)(513,358)(515,357)(516,354)(517,352)(519,350)(520,348)(522,346)(523,344)(525,342)(526,340)(528,338)(529,336)(530,335)(532,334)(533,333)(535,331)(536,330)(538,328)(539,326)(541,325)(542,323)(544,322)(545,320)(546,319)(548,318)(549,317)(551,316)(552,315)(554,314)(555,313)(557,312)(558,311)(559,310)(561,309)(562,308)(564,308)
\thicklines \path(564,308)(565,307)(567,306)(568,305)(570,304)(571,303)(573,303)(574,302)(575,301)(577,300)(591,294)(606,288)(620,283)(635,279)(649,276)(664,273)(678,270)(693,268)(707,265)(722,263)(736,261)(751,260)(765,258)(780,257)(794,256)(809,254)(823,253)(838,252)(852,251)(867,250)(881,250)(896,249)(910,248)(925,248)(939,247)(954,246)(968,248)(983,249)(997,249)(1012,248)(1026,248)(1041,247)(1055,247)(1070,246)(1084,246)(1099,245)(1113,245)(1128,245)(1142,244)(1157,244)(1171,244)
\thicklines \path(1171,244)(1186,243)(1200,243)(1215,243)(1229,242)(1244,243)(1258,243)(1273,243)(1287,242)(1302,242)(1316,249)(1331,252)(1345,252)(1360,252)(1374,252)(1389,251)(1403,251)(1418,251)(1432,250)(1447,250)(1461,250)(1476,249)(1490,249)(1505,249)(1519,248)(1534,248)(1548,248)(1563,247)(1577,247)(1592,247)(1606,247)(1621,246)(1635,246)(1650,246)(1664,246)(1679,245)(1693,245)(1708,245)(1722,245)(1737,245)(1751,244)(1766,244)(1780,244)(1795,244)(1809,244)(1824,243)(1838,243)(1853,243)(1867,243)(1882,243)(2027,241)
\thicklines \path(2027,241)(2027,241)
\thicklines \path(441,1022)(442,929)(444,822)(445,748)(447,825)(448,897)(449,947)(451,881)(452,901)(454,1008)(454,1022)
\thicklines \path(456,1022)(457,980)(458,1006)(460,959)(461,993)(462,957)(464,981)(465,1008)(467,979)(468,1022)
\thicklines \path(471,1022)(471,1010)(473,986)(474,959)(476,980)(477,958)(478,975)(480,958)(481,977)(483,959)(484,937)(486,912)(487,897)(489,877)(490,856)(491,841)(493,830)(494,812)(496,797)(497,782)(499,765)(500,758)(502,746)(503,732)(505,790)(506,803)(507,791)(509,825)(510,835)(512,845)(513,855)(515,865)(516,852)(518,841)(519,852)(520,844)(522,833)(523,824)(525,831)(526,820)(528,831)(529,822)(531,830)(532,840)(534,827)(535,836)(536,828)(538,821)(539,828)(541,821)(542,831)
\thicklines \path(542,831)(544,823)(545,815)(547,807)(548,844)(549,851)(551,841)(552,861)(554,855)(555,861)(557,867)(558,858)(560,890)(561,881)(563,875)(564,868)(565,861)(567,877)(568,883)(570,875)(571,870)(573,863)(574,881)(576,885)(590,849)(605,841)(619,802)(634,811)(648,777)(663,777)(677,771)(692,794)(706,788)(721,783)(735,779)(750,785)(764,779)(779,781)(793,773)(808,762)(822,774)(837,764)(851,748)(866,746)(880,735)(895,733)(909,726)(924,716)(938,719)(953,720)(967,723)
\thicklines \path(967,723)(982,720)(996,719)(1011,719)(1025,718)(1040,709)(1054,703)(1069,701)(1083,694)(1098,704)(1112,704)(1127,704)(1141,699)(1156,689)(1170,685)(1185,678)(1199,674)(1214,674)(1228,678)(1243,683)(1257,688)(1272,684)(1286,678)(1301,680)(1315,678)(1330,691)(1344,691)(1359,687)(1373,687)(1388,695)(1402,705)(1417,709)(1431,706)(1446,704)(1460,705)(1475,705)(1489,708)(1504,709)(1518,710)(1533,709)(1547,711)(1562,707)(1576,708)(1591,707)(1605,703)(1620,701)(1634,698)(1649,698)(1663,695)(1678,700)(1692,700)
\thicklines \path(1692,700)(1707,699)(1721,699)(1736,696)(1750,694)(1765,699)(1779,698)(1794,699)(1808,695)(1823,694)(1837,692)(1852,690)(1866,689)(1881,687)(2026,694)(2027,694)
\thicklines \path(434,226)(434,226)(435,226)(436,565)(438,482)(439,442)(441,583)(442,532)(444,495)(445,583)(447,556)(448,523)(450,502)(451,564)(452,545)(454,601)(455,582)(457,560)(458,606)(460,590)(461,632)(463,615)(464,595)(465,631)(467,613)(468,641)(470,623)(471,607)(473,594)(474,619)(476,683)(477,666)(479,690)(480,711)(481,699)(483,685)(484,675)(486,723)(487,710)(489,701)(490,693)(492,682)(493,671)(495,745)(496,757)(497,771)(499,837)(500,845)(502,834)(503,893)(505,900)
\thicklines \path(505,900)(506,887)(508,874)(509,862)(510,850)(512,839)(513,852)(515,838)(516,826)(518,835)(519,843)(521,832)(522,843)(524,850)(525,856)(526,846)(528,837)(529,831)(531,823)(532,815)(534,806)(535,797)(537,790)(538,782)(539,772)(541,768)(542,764)(544,757)(545,750)(547,758)(548,750)(550,758)(551,752)(553,746)(554,739)(555,732)(557,742)(558,735)(560,743)(561,737)(563,744)(564,738)(566,734)(567,728)(568,723)(570,731)(571,725)(573,719)(574,727)(576,721)(590,668)
\thicklines \path(590,668)(605,670)(619,698)(634,674)(648,688)(663,665)(677,674)(692,686)(706,691)(721,685)(735,688)(750,687)(764,687)(779,667)(793,671)(808,671)(822,685)(837,689)(851,685)(866,699)(880,716)(895,721)(909,712)(924,698)(938,720)(953,720)(967,714)(982,711)(996,716)(1011,711)(1025,723)(1040,717)(1054,717)(1069,713)(1083,714)(1098,710)(1112,709)(1127,711)(1141,708)(1156,706)(1170,708)(1185,701)(1199,700)(1214,699)(1228,698)(1243,698)(1257,697)(1272,691)(1286,693)(1301,687)(1315,682)
\thicklines \path(1315,682)(1330,682)(1344,675)(1359,673)(1373,677)(1388,676)(1402,674)(1417,672)(1431,671)(1446,673)(1460,670)(1475,671)(1489,668)(1504,671)(1518,669)(1533,669)(1547,667)(1562,665)(1576,665)(1591,669)(1605,676)(1620,679)(1634,679)(1649,680)(1663,682)(1678,682)(1692,681)(1707,679)(1721,679)(1736,680)(1750,680)(1765,688)(1779,688)(1794,689)(1808,691)(1823,689)(1837,690)(1852,691)(1866,689)(1881,685)(2026,699)(2027,699)
\thicklines \path(434,226)(434,226)(435,226)(437,226)(438,226)(440,226)(441,226)(442,226)(444,226)(445,226)(447,226)(448,226)(450,226)(451,226)(453,226)(454,226)(455,226)(457,226)(458,226)(460,226)(461,226)(463,226)(464,226)(466,226)(467,226)(469,226)(470,226)(471,226)(473,226)(474,226)(476,226)(477,226)(479,226)(480,226)(482,226)(483,226)(484,226)(486,226)(487,226)(489,226)(490,226)(492,226)(493,226)(495,226)(496,226)(498,226)(499,226)(500,226)(502,226)(503,240)(505,240)
\thicklines \path(505,240)(506,239)(508,239)(509,239)(511,238)(512,238)(513,238)(515,246)(516,253)(518,253)(519,260)(521,266)(522,277)(524,276)(525,293)(527,292)(528,291)(529,290)(531,289)(532,289)(534,288)(535,287)(537,286)(538,285)(540,285)(541,284)(543,283)(544,282)(545,282)(547,281)(548,280)(550,280)(551,279)(553,278)(554,278)(556,277)(557,277)(558,276)(560,275)(561,275)(563,274)(564,274)(566,273)(567,273)(569,272)(570,272)(572,271)(573,270)(574,270)(576,270)(590,265)
\thicklines \path(590,265)(605,262)(619,297)(634,292)(648,288)(663,300)(677,308)(692,304)(706,300)(721,297)(735,294)(750,305)(764,307)(779,334)(793,341)(808,349)(822,352)(837,358)(851,354)(866,359)(880,359)(895,362)(909,357)(924,354)(938,353)(953,352)(967,360)(982,359)(996,359)(1011,370)(1025,367)(1040,366)(1054,368)(1069,365)(1083,366)(1098,366)(1112,363)(1127,365)(1141,364)(1156,368)(1170,374)(1185,373)(1199,378)(1214,375)(1228,373)(1243,370)(1257,367)(1272,365)(1286,363)(1301,361)(1315,359)
\thicklines \path(1315,359)(1330,357)(1344,355)(1359,353)(1373,351)(1388,349)(1402,348)(1417,346)(1431,346)(1446,345)(1460,347)(1475,348)(1489,346)(1504,345)(1518,348)(1533,355)(1547,360)(1562,365)(1576,371)(1591,374)(1605,379)(1620,385)(1634,383)(1649,383)(1663,383)(1678,381)(1692,379)(1707,379)(1721,377)(1736,375)(1750,374)(1765,372)(1779,371)(1794,369)(1808,368)(1823,366)(1837,365)(1852,364)(1866,362)(1881,363)(2026,373)(2027,373)
\color{black}
\put(1734,902){\makebox(0,0)[r]{FS ($m=100$)}}
\color{blue}
\thinlines \path(1763,902)(1920,902)
\thinlines \path(432,226)(432,226)(432,226)(432,226)(432,226)(433,226)(433,226)(433,226)(433,226)(433,226)(433,721)(434,1022)
\thinlines \path(437,1022)(438,909)(439,828)(441,745)(442,784)(443,724)(445,672)(446,631)(448,592)(449,641)(451,677)(452,703)(454,703)(455,670)(457,743)(458,772)(459,754)(461,745)(462,748)(464,746)(465,749)(467,767)(468,783)(470,762)(471,728)(472,718)(474,703)(475,720)(477,703)(478,694)(480,708)(481,683)(483,676)(484,650)(486,635)(487,620)(488,632)(490,625)(491,611)(493,598)(494,588)(496,582)(497,573)(499,570)(500,565)(501,552)(503,559)(504,568)(506,577)(507,568)(509,566)
\thinlines \path(509,566)(510,572)(512,577)(513,599)(515,590)(516,587)(517,579)(519,576)(520,565)(522,558)(523,554)(525,552)(526,548)(528,545)(529,540)(530,545)(532,552)(533,547)(535,542)(536,550)(538,545)(539,540)(541,538)(542,532)(544,528)(545,535)(546,530)(548,528)(549,525)(551,540)(552,538)(554,531)(555,528)(557,526)(558,524)(559,530)(561,528)(562,543)(564,541)(565,539)(567,530)(568,536)(570,531)(571,527)(573,524)(574,530)(575,528)(577,524)(591,542)(606,530)(620,552)
\thinlines \path(620,552)(635,551)(649,553)(664,553)(678,544)(693,547)(707,553)(722,560)(736,553)(751,568)(765,572)(780,585)(794,579)(809,578)(823,584)(838,583)(852,581)(867,584)(881,576)(896,571)(910,570)(925,564)(939,561)(954,570)(968,566)(983,569)(997,564)(1012,568)(1026,572)(1041,568)(1055,565)(1070,567)(1084,574)(1099,572)(1113,568)(1128,565)(1142,571)(1157,570)(1171,570)(1186,568)(1200,568)(1215,565)(1229,562)(1244,567)(1258,565)(1273,562)(1287,563)(1302,565)(1316,562)(1331,560)(1345,557)
\thinlines \path(1345,557)(1360,556)(1374,554)(1389,556)(1403,558)(1418,554)(1432,553)(1447,551)(1461,549)(1476,546)(1490,544)(1505,542)(1519,542)(1534,541)(1548,539)(1563,539)(1577,538)(1592,542)(1606,544)(1621,541)(1635,541)(1650,541)(1664,540)(1679,538)(1693,538)(1708,538)(1722,537)(1737,538)(1751,534)(1766,533)(1780,532)(1795,532)(1809,531)(1824,532)(1838,533)(1853,534)(1867,534)(1882,531)(2027,535)(2027,535)
\thinlines \path(433,226)(433,226)(435,226)(436,226)(438,413)(439,373)(441,354)(442,335)(444,313)(445,300)(447,292)(448,285)(449,386)(451,379)(452,374)(454,367)(455,401)(457,380)(458,375)(460,395)(461,386)(462,379)(464,371)(465,364)(467,415)(468,459)(470,455)(471,448)(473,437)(474,477)(476,489)(477,505)(478,522)(480,513)(481,505)(483,498)(484,512)(486,524)(487,510)(489,517)(490,505)(491,500)(493,489)(494,482)(496,490)(497,481)(499,478)(500,474)(502,471)(503,466)(505,459)
\thinlines \path(505,459)(506,456)(507,451)(509,446)(510,443)(512,440)(513,437)(515,432)(516,440)(518,438)(519,435)(520,432)(522,437)(523,435)(525,442)(526,440)(528,446)(529,459)(531,454)(532,450)(534,445)(535,444)(536,437)(538,443)(539,448)(541,446)(542,444)(544,442)(545,457)(547,463)(548,460)(549,458)(551,469)(552,467)(554,465)(555,464)(557,461)(558,458)(560,454)(561,460)(563,458)(564,464)(565,462)(567,460)(568,459)(570,456)(571,455)(573,453)(574,458)(576,462)(590,454)
\thinlines \path(590,454)(605,444)(619,435)(634,430)(648,436)(663,441)(677,431)(692,432)(706,435)(721,436)(735,436)(750,441)(764,435)(779,442)(793,451)(808,457)(822,456)(837,450)(851,451)(866,449)(880,447)(895,456)(909,457)(924,464)(938,464)(953,463)(967,466)(982,469)(996,474)(1011,469)(1025,470)(1040,470)(1054,473)(1069,472)(1083,476)(1098,475)(1112,473)(1127,472)(1141,469)(1156,468)(1170,470)(1185,473)(1199,473)(1214,470)(1228,474)(1243,473)(1257,472)(1272,476)(1286,480)(1301,484)(1315,482)
\thinlines \path(1315,482)(1330,482)(1344,481)(1359,483)(1373,483)(1388,485)(1402,487)(1417,487)(1431,485)(1446,484)(1460,483)(1475,483)(1489,482)(1504,481)(1518,482)(1533,488)(1547,487)(1562,488)(1576,488)(1591,489)(1605,490)(1620,489)(1634,490)(1649,489)(1663,490)(1678,489)(1692,488)(1707,487)(1721,487)(1736,488)(1750,487)(1765,484)(1779,484)(1794,485)(1808,484)(1823,483)(1837,482)(1852,484)(1866,483)(1881,485)(2026,489)(2027,489)
\thinlines \path(434,226)(434,226)(435,226)(436,387)(438,329)(439,322)(441,456)(442,434)(444,420)(445,399)(447,383)(448,368)(450,406)(451,394)(452,382)(454,370)(455,361)(457,350)(458,346)(460,342)(461,363)(463,357)(464,350)(465,372)(467,368)(468,385)(470,403)(471,394)(473,407)(474,401)(476,411)(477,407)(479,405)(480,401)(481,416)(483,410)(484,424)(486,419)(487,413)(489,408)(490,407)(492,404)(493,401)(495,398)(496,392)(497,388)(499,386)(500,384)(502,396)(503,392)(505,401)
\thinlines \path(505,401)(506,410)(508,420)(509,427)(510,423)(512,417)(513,410)(515,417)(516,415)(518,412)(519,420)(521,417)(522,414)(524,410)(525,417)(526,415)(528,412)(529,419)(531,418)(532,423)(534,419)(535,416)(537,414)(538,420)(539,417)(541,414)(542,412)(544,410)(545,409)(547,406)(548,405)(550,403)(551,402)(553,400)(554,399)(555,397)(557,403)(558,409)(560,414)(561,413)(563,420)(564,418)(566,423)(567,420)(568,418)(570,424)(571,422)(573,421)(574,420)(576,417)(590,412)
\thinlines \path(590,412)(605,438)(619,440)(634,428)(648,428)(663,434)(677,443)(692,441)(706,444)(721,447)(735,449)(750,450)(764,459)(779,466)(793,473)(808,472)(822,466)(837,465)(851,464)(866,472)(880,471)(895,465)(909,458)(924,452)(938,450)(953,449)(967,446)(982,448)(996,448)(1011,448)(1025,446)(1040,447)(1054,446)(1069,448)(1083,450)(1098,452)(1112,452)(1127,452)(1141,457)(1156,455)(1170,456)(1185,455)(1199,454)(1214,456)(1228,455)(1243,457)(1257,454)(1272,453)(1286,457)(1301,460)(1315,460)
\thinlines \path(1315,460)(1330,459)(1344,460)(1359,459)(1373,457)(1388,457)(1402,459)(1417,461)(1431,460)(1446,462)(1460,460)(1475,461)(1489,462)(1504,461)(1518,462)(1533,460)(1547,461)(1562,461)(1576,459)(1591,458)(1605,458)(1620,458)(1634,458)(1649,457)(1663,460)(1678,463)(1692,462)(1707,464)(1721,462)(1736,465)(1750,465)(1765,464)(1779,462)(1794,464)(1808,465)(1823,468)(1837,467)(1852,466)(1866,466)(1881,465)(2026,464)(2027,464)
\thinlines \path(434,583)(434,583)(435,458)(437,410)(438,376)(440,351)(441,342)(442,320)(444,384)(445,443)(447,417)(448,400)(450,446)(451,537)(453,511)(454,489)(455,460)(457,433)(458,427)(460,480)(461,503)(463,497)(464,489)(466,483)(467,499)(469,485)(470,468)(471,487)(473,504)(474,520)(476,504)(477,521)(479,504)(480,519)(482,505)(483,499)(484,487)(486,481)(487,478)(489,472)(490,465)(492,461)(493,455)(495,447)(496,455)(498,451)(499,459)(500,456)(502,453)(503,464)(505,459)
\thinlines \path(505,459)(506,468)(508,463)(509,458)(511,454)(512,463)(513,460)(515,457)(516,455)(518,448)(519,446)(521,453)(522,450)(524,455)(525,452)(527,458)(528,454)(529,450)(531,447)(532,443)(534,441)(535,436)(537,444)(538,441)(540,438)(541,435)(543,433)(544,432)(545,430)(547,429)(548,427)(550,425)(551,423)(553,422)(554,420)(556,424)(557,430)(558,428)(560,425)(561,422)(563,420)(564,416)(566,415)(567,412)(569,411)(570,407)(572,406)(573,411)(574,416)(576,415)(590,403)
\thinlines \path(590,403)(605,427)(619,447)(634,464)(648,452)(663,446)(677,471)(692,473)(706,477)(721,473)(735,471)(750,466)(764,467)(779,472)(793,465)(808,469)(822,469)(837,464)(851,465)(866,471)(880,467)(895,466)(909,463)(924,458)(938,455)(953,458)(967,456)(982,456)(996,461)(1011,455)(1025,456)(1040,454)(1054,461)(1069,461)(1083,460)(1098,459)(1112,469)(1127,471)(1141,470)(1156,470)(1170,475)(1185,477)(1199,476)(1214,478)(1228,481)(1243,482)(1257,485)(1272,486)(1286,485)(1301,485)(1315,483)
\thinlines \path(1315,483)(1330,481)(1344,482)(1359,481)(1373,481)(1388,482)(1402,478)(1417,480)(1431,482)(1446,480)(1460,479)(1475,475)(1489,473)(1504,473)(1518,471)(1533,473)(1547,474)(1562,473)(1576,473)(1591,473)(1605,473)(1620,473)(1634,476)(1649,478)(1663,477)(1678,477)(1692,477)(1707,477)(1721,478)(1736,477)(1750,479)(1765,481)(1779,482)(1794,483)(1808,483)(1823,483)(1837,483)(1852,484)(1866,486)(1881,485)(2026,487)(2027,487)
\color{black}
\put(1734,844){\makebox(0,0)[r]{MultipleRW ($K=100$)}}
\Thicklines \path(1763,844)(1920,844)
\Thicklines \path(432,226)(432,226)(432,226)(432,226)(432,226)(433,226)(433,226)(433,226)(433,226)(433,226)(433,226)(435,706)(436,687)(438,902)(439,1022)
\Thicklines \path(439,1022)(441,828)(442,745)(443,654)(445,673)(446,617)(448,567)(449,534)(451,496)(452,467)(454,440)(455,462)(457,438)(458,429)(459,446)(461,430)(462,420)(464,436)(465,429)(467,455)(468,446)(470,446)(471,440)(472,430)(474,424)(475,416)(477,408)(478,403)(480,425)(481,418)(483,412)(484,408)(486,413)(487,411)(488,417)(490,411)(491,408)(493,413)(494,419)(496,443)(497,438)(499,432)(500,429)(501,423)(503,420)(504,416)(506,410)(507,417)(509,412)(510,407)(512,412)
\Thicklines \path(512,412)(513,417)(515,414)(516,411)(517,410)(519,407)(520,411)(522,416)(523,420)(525,431)(526,427)(528,422)(529,426)(530,424)(532,421)(533,418)(535,422)(536,426)(538,429)(539,426)(541,424)(542,428)(544,425)(545,422)(546,421)(548,418)(549,416)(551,419)(552,416)(554,419)(555,417)(557,412)(558,410)(559,408)(561,407)(562,404)(564,408)(565,407)(567,404)(568,401)(570,399)(571,403)(573,404)(574,403)(575,402)(577,400)(591,402)(606,391)(620,385)(635,380)(649,381)
\Thicklines \path(649,381)(664,376)(678,380)(693,378)(707,371)(722,369)(736,365)(751,363)(765,365)(780,368)(794,370)(809,368)(823,369)(838,372)(852,372)(867,371)(881,372)(896,371)(910,373)(925,369)(939,369)(954,367)(968,368)(983,367)(997,369)(1012,372)(1026,373)(1041,373)(1055,373)(1070,375)(1084,373)(1099,373)(1113,374)(1128,372)(1142,372)(1157,372)(1171,370)(1186,368)(1200,370)(1215,370)(1229,373)(1244,375)(1258,374)(1273,377)(1287,377)(1302,377)(1316,378)(1331,380)(1345,379)(1360,382)(1374,381)
\Thicklines \path(1374,381)(1389,381)(1403,382)(1418,384)(1432,382)(1447,386)(1461,385)(1476,385)(1490,382)(1505,382)(1519,382)(1534,381)(1548,381)(1563,381)(1577,381)(1592,381)(1606,381)(1621,381)(1635,381)(1650,383)(1664,382)(1679,383)(1693,383)(1708,383)(1722,383)(1737,383)(1751,382)(1766,383)(1780,382)(1795,382)(1809,381)(1824,382)(1838,381)(1853,381)(1867,380)(1882,380)(2027,378)(2027,378)
\Thicklines \path(435,451)(435,451)(436,376)(438,339)(439,320)(441,392)(442,369)(443,402)(445,390)(446,378)(448,360)(449,382)(451,366)(452,356)(454,345)(455,334)(457,323)(458,315)(459,308)(461,323)(462,355)(464,352)(465,363)(467,359)(468,354)(470,365)(471,360)(472,351)(474,349)(475,346)(477,343)(478,339)(480,348)(481,343)(483,338)(484,334)(486,329)(487,325)(488,333)(490,329)(491,346)(493,343)(494,341)(496,339)(497,337)(499,344)(500,341)(501,345)(503,344)(504,349)(506,348)
\Thicklines \path(506,348)(507,346)(509,344)(510,350)(512,354)(513,350)(515,347)(516,350)(517,354)(519,352)(520,349)(522,347)(523,345)(525,343)(526,341)(528,345)(529,344)(530,341)(532,340)(533,338)(535,336)(536,335)(538,333)(539,332)(541,331)(542,329)(544,327)(545,330)(546,333)(548,331)(549,329)(551,332)(552,331)(554,330)(555,329)(557,340)(558,348)(559,345)(561,345)(562,343)(564,342)(565,345)(567,344)(568,343)(570,342)(571,344)(573,342)(574,344)(575,343)(577,343)(591,336)
\Thicklines \path(591,336)(606,333)(620,332)(635,334)(649,333)(664,337)(678,340)(693,334)(707,331)(722,328)(736,329)(751,327)(765,337)(780,340)(794,339)(809,341)(823,341)(838,337)(852,338)(867,338)(881,335)(896,334)(910,335)(925,335)(939,334)(954,333)(968,331)(983,329)(997,330)(1012,327)(1026,330)(1041,330)(1055,328)(1070,328)(1084,328)(1099,326)(1113,326)(1128,324)(1142,323)(1157,322)(1171,320)(1186,319)(1200,318)(1215,319)(1229,318)(1244,318)(1258,318)(1273,320)(1287,318)(1302,318)(1316,318)
\Thicklines \path(1316,318)(1331,317)(1345,318)(1360,321)(1374,321)(1389,322)(1403,321)(1418,322)(1432,322)(1447,321)(1461,322)(1476,321)(1490,320)(1505,319)(1519,319)(1534,318)(1548,319)(1563,317)(1577,318)(1592,318)(1606,318)(1621,317)(1635,319)(1650,319)(1664,318)(1679,320)(1693,320)(1708,320)(1722,320)(1737,320)(1751,320)(1766,319)(1780,320)(1795,319)(1809,319)(1824,318)(1838,318)(1853,319)(1867,319)(1882,318)(2027,316)(2027,316)
\Thicklines \path(435,487)(435,487)(436,654)(438,504)(439,431)(441,383)(442,441)(443,489)(445,508)(446,572)(448,534)(449,516)(451,501)(452,471)(454,448)(455,444)(457,422)(458,410)(459,388)(461,376)(462,385)(464,377)(465,389)(467,382)(468,374)(470,368)(471,379)(472,375)(474,372)(475,379)(477,372)(478,369)(480,367)(481,362)(483,358)(484,353)(486,347)(487,345)(488,339)(490,346)(491,350)(493,354)(494,370)(496,368)(497,366)(499,363)(500,370)(501,368)(503,366)(504,371)(506,368)
\Thicklines \path(506,368)(507,367)(509,366)(510,361)(512,359)(513,356)(515,360)(516,357)(517,354)(519,357)(520,353)(522,351)(523,350)(525,355)(526,353)(528,351)(529,356)(530,360)(532,358)(533,362)(535,360)(536,358)(538,357)(539,354)(541,352)(542,350)(544,347)(545,346)(546,344)(548,347)(549,345)(551,349)(552,348)(554,347)(555,345)(557,344)(558,347)(559,346)(561,344)(562,343)(564,345)(565,343)(567,347)(568,349)(570,347)(571,350)(573,352)(574,352)(575,350)(577,352)(591,345)
\Thicklines \path(591,345)(606,349)(620,343)(635,344)(649,340)(664,337)(678,336)(693,337)(707,335)(722,332)(736,332)(751,332)(765,329)(780,329)(794,329)(809,326)(823,325)(838,323)(852,322)(867,321)(881,320)(896,321)(910,320)(925,321)(939,321)(954,319)(968,318)(983,322)(997,321)(1012,318)(1026,319)(1041,320)(1055,321)(1070,322)(1084,321)(1099,319)(1113,321)(1128,319)(1142,319)(1157,318)(1171,318)(1186,317)(1200,316)(1215,316)(1229,316)(1244,315)(1258,314)(1273,314)(1287,314)(1302,314)(1316,314)
\Thicklines \path(1316,314)(1331,315)(1345,314)(1360,315)(1374,314)(1389,314)(1403,313)(1418,315)(1432,315)(1447,315)(1461,315)(1476,317)(1490,318)(1505,318)(1519,318)(1534,319)(1548,318)(1563,318)(1577,317)(1592,317)(1606,317)(1621,318)(1635,317)(1650,318)(1664,318)(1679,318)(1693,319)(1708,319)(1722,319)(1737,319)(1751,319)(1766,319)(1780,320)(1795,320)(1809,320)(1824,321)(1838,320)(1853,321)(1867,321)(1882,321)(2027,323)(2027,323)
\Thicklines \path(435,226)(435,226)(436,226)(438,226)(439,226)(441,226)(442,226)(443,226)(445,226)(446,281)(448,320)(449,306)(451,332)(452,318)(454,338)(455,361)(457,380)(458,373)(459,385)(461,372)(462,364)(464,359)(465,371)(467,368)(468,361)(470,350)(471,347)(472,345)(474,342)(475,340)(477,336)(478,345)(480,339)(481,334)(483,330)(484,327)(486,348)(487,347)(488,343)(490,337)(491,335)(493,332)(494,338)(496,337)(497,334)(499,330)(500,329)(501,327)(503,325)(504,325)(506,322)
\Thicklines \path(506,322)(507,320)(509,316)(510,314)(512,312)(513,311)(515,317)(516,323)(517,320)(519,318)(520,323)(522,321)(523,320)(525,319)(526,317)(528,315)(529,314)(530,320)(532,318)(533,317)(535,316)(536,315)(538,313)(539,311)(541,310)(542,309)(544,308)(545,313)(546,317)(548,320)(549,329)(551,327)(552,326)(554,325)(555,328)(557,327)(558,331)(559,330)(561,328)(562,328)(564,326)(565,325)(567,323)(568,322)(570,320)(571,319)(573,318)(574,322)(575,325)(577,332)(591,337)
\Thicklines \path(591,337)(606,336)(620,344)(635,342)(649,346)(664,340)(678,338)(693,338)(707,339)(722,335)(736,341)(751,343)(765,346)(780,348)(794,348)(809,343)(823,340)(838,340)(852,340)(867,336)(881,337)(896,334)(910,336)(925,335)(939,335)(954,338)(968,336)(983,336)(997,335)(1012,334)(1026,334)(1041,335)(1055,333)(1070,331)(1084,331)(1099,331)(1113,329)(1128,331)(1142,329)(1157,329)(1171,328)(1186,328)(1200,328)(1215,329)(1229,329)(1244,331)(1258,329)(1273,329)(1287,329)(1302,330)(1316,329)
\Thicklines \path(1316,329)(1331,329)(1345,328)(1360,328)(1374,329)(1389,329)(1403,329)(1418,328)(1432,328)(1447,327)(1461,328)(1476,327)(1490,326)(1505,324)(1519,325)(1534,324)(1548,323)(1563,324)(1577,324)(1592,323)(1606,322)(1621,322)(1635,322)(1650,321)(1664,321)(1679,320)(1693,320)(1708,319)(1722,319)(1737,318)(1751,318)(1766,319)(1780,318)(1795,318)(1809,318)(1824,318)(1838,318)(1853,318)(1867,320)(1882,320)(2027,320)(2027,320)
\thicklines \path(432,1022)(432,226)(2027,226)
\end{picture}

%% file: figs/biasedGroupFreqDist_flickr-links_flickr-groupmemberships_STEPS18612_K100_SCC0_RUNS10000.compute-0-15.local.tex
\setlength{\unitlength}{0.120450pt}
\begin{picture}(1949,1260)(0,0)
\fontsize{7}{8.98332}\selectfont
\color{black}
\thicklines \path(490,226)(449,226)
\put(420,226){\makebox(0,0)[r]{$10^{-2}$}}
\color{black}
\thicklines \path(490,318)(449,318)
\put(420,318){\makebox(0,0)[r]{$0.15$}}
\color{black}
\thicklines \path(490,416)(449,416)
\put(420,416){\makebox(0,0)[r]{$0.3$}}
\color{black}
\thicklines \path(490,547)(449,547)
\put(420,547){\makebox(0,0)[r]{$0.5$}}
\color{black}
\thicklines \path(490,874)(449,874)
\put(420,874){\makebox(0,0)[r]{$1$}}
\color{black}
\thicklines \path(490,1202)(449,1202)
\put(420,1202){\makebox(0,0)[r]{$1.5$}}
\color{black}
\thicklines \path(519,226)(519,185)
\put(519,127){\makebox(0,0){$1$}}
\color{black}
\thicklines \path(877,226)(877,185)
\put(877,127){\makebox(0,0){$50$}}
\color{black}
\thicklines \path(1241,226)(1241,185)
\put(1241,127){\makebox(0,0){$100$}}
\color{black}
\thicklines \path(1606,226)(1606,185)
\put(1606,127){\makebox(0,0){$150$}}
\color{black}
\color{black}
\thicklines \path(490,1202)(490,226)(1876,226)
\color{black}
\put(101,714){\makebox(0,0)[l]{\rotatebox[origin=c]{90}{\NMSE}}}
\color{black}
\color{black}
\put(1183,40){\makebox(0,0){Group index}}
\color{black}
\color{black}
\color{red}
\color{black}
\put(1734,1320){\makebox(0,0)[r]{SingleRW}}
\color{red}
\put(512,400){\makebox(0,0){$\scriptstyle\Diamond$}}
\put(519,445){\makebox(0,0){$\scriptstyle\Diamond$}}
\put(526,395){\makebox(0,0){$\scriptstyle\Diamond$}}
\put(534,396){\makebox(0,0){$\scriptstyle\Diamond$}}
\put(541,395){\makebox(0,0){$\scriptstyle\Diamond$}}
\put(548,481){\makebox(0,0){$\scriptstyle\Diamond$}}
\put(556,425){\makebox(0,0){$\scriptstyle\Diamond$}}
\put(563,437){\makebox(0,0){$\scriptstyle\Diamond$}}
\put(570,409){\makebox(0,0){$\scriptstyle\Diamond$}}
\put(578,418){\makebox(0,0){$\scriptstyle\Diamond$}}
\put(585,406){\makebox(0,0){$\scriptstyle\Diamond$}}
\put(592,422){\makebox(0,0){$\scriptstyle\Diamond$}}
\put(599,401){\makebox(0,0){$\scriptstyle\Diamond$}}
\put(607,412){\makebox(0,0){$\scriptstyle\Diamond$}}
\put(614,1029){\makebox(0,0){$\scriptstyle\Diamond$}}
\put(621,406){\makebox(0,0){$\scriptstyle\Diamond$}}
\put(629,413){\makebox(0,0){$\scriptstyle\Diamond$}}
\put(636,412){\makebox(0,0){$\scriptstyle\Diamond$}}
\put(643,433){\makebox(0,0){$\scriptstyle\Diamond$}}
\put(650,420){\makebox(0,0){$\scriptstyle\Diamond$}}
\put(658,431){\makebox(0,0){$\scriptstyle\Diamond$}}
\put(665,438){\makebox(0,0){$\scriptstyle\Diamond$}}
\put(672,402){\makebox(0,0){$\scriptstyle\Diamond$}}
\put(680,404){\makebox(0,0){$\scriptstyle\Diamond$}}
\put(687,406){\makebox(0,0){$\scriptstyle\Diamond$}}
\put(694,409){\makebox(0,0){$\scriptstyle\Diamond$}}
\put(702,418){\makebox(0,0){$\scriptstyle\Diamond$}}
\put(709,405){\makebox(0,0){$\scriptstyle\Diamond$}}
\put(716,400){\makebox(0,0){$\scriptstyle\Diamond$}}
\put(723,439){\makebox(0,0){$\scriptstyle\Diamond$}}
\put(731,399){\makebox(0,0){$\scriptstyle\Diamond$}}
\put(738,414){\makebox(0,0){$\scriptstyle\Diamond$}}
\put(745,416){\makebox(0,0){$\scriptstyle\Diamond$}}
\put(753,407){\makebox(0,0){$\scriptstyle\Diamond$}}
\put(760,407){\makebox(0,0){$\scriptstyle\Diamond$}}
\put(767,398){\makebox(0,0){$\scriptstyle\Diamond$}}
\put(774,410){\makebox(0,0){$\scriptstyle\Diamond$}}
\put(782,570){\makebox(0,0){$\scriptstyle\Diamond$}}
\put(789,541){\makebox(0,0){$\scriptstyle\Diamond$}}
\put(796,432){\makebox(0,0){$\scriptstyle\Diamond$}}
\put(804,403){\makebox(0,0){$\scriptstyle\Diamond$}}
\put(811,412){\makebox(0,0){$\scriptstyle\Diamond$}}
\put(818,426){\makebox(0,0){$\scriptstyle\Diamond$}}
\put(826,404){\makebox(0,0){$\scriptstyle\Diamond$}}
\put(833,439){\makebox(0,0){$\scriptstyle\Diamond$}}
\put(840,411){\makebox(0,0){$\scriptstyle\Diamond$}}
\put(847,412){\makebox(0,0){$\scriptstyle\Diamond$}}
\put(855,424){\makebox(0,0){$\scriptstyle\Diamond$}}
\put(862,417){\makebox(0,0){$\scriptstyle\Diamond$}}
\put(869,430){\makebox(0,0){$\scriptstyle\Diamond$}}
\put(877,415){\makebox(0,0){$\scriptstyle\Diamond$}}
\put(884,403){\makebox(0,0){$\scriptstyle\Diamond$}}
\put(891,410){\makebox(0,0){$\scriptstyle\Diamond$}}
\put(899,415){\makebox(0,0){$\scriptstyle\Diamond$}}
\put(906,427){\makebox(0,0){$\scriptstyle\Diamond$}}
\put(913,411){\makebox(0,0){$\scriptstyle\Diamond$}}
\put(920,408){\makebox(0,0){$\scriptstyle\Diamond$}}
\put(928,411){\makebox(0,0){$\scriptstyle\Diamond$}}
\put(935,412){\makebox(0,0){$\scriptstyle\Diamond$}}
\put(942,406){\makebox(0,0){$\scriptstyle\Diamond$}}
\put(950,915){\makebox(0,0){$\scriptstyle\Diamond$}}
\put(957,414){\makebox(0,0){$\scriptstyle\Diamond$}}
\put(964,408){\makebox(0,0){$\scriptstyle\Diamond$}}
\put(971,405){\makebox(0,0){$\scriptstyle\Diamond$}}
\put(979,414){\makebox(0,0){$\scriptstyle\Diamond$}}
\put(986,440){\makebox(0,0){$\scriptstyle\Diamond$}}
\put(993,421){\makebox(0,0){$\scriptstyle\Diamond$}}
\put(1001,404){\makebox(0,0){$\scriptstyle\Diamond$}}
\put(1008,418){\makebox(0,0){$\scriptstyle\Diamond$}}
\put(1015,427){\makebox(0,0){$\scriptstyle\Diamond$}}
\put(1023,410){\makebox(0,0){$\scriptstyle\Diamond$}}
\put(1030,417){\makebox(0,0){$\scriptstyle\Diamond$}}
\put(1037,434){\makebox(0,0){$\scriptstyle\Diamond$}}
\put(1044,425){\makebox(0,0){$\scriptstyle\Diamond$}}
\put(1052,427){\makebox(0,0){$\scriptstyle\Diamond$}}
\put(1059,426){\makebox(0,0){$\scriptstyle\Diamond$}}
\put(1066,426){\makebox(0,0){$\scriptstyle\Diamond$}}
\put(1074,419){\makebox(0,0){$\scriptstyle\Diamond$}}
\put(1081,549){\makebox(0,0){$\scriptstyle\Diamond$}}
\put(1088,421){\makebox(0,0){$\scriptstyle\Diamond$}}
\put(1095,432){\makebox(0,0){$\scriptstyle\Diamond$}}
\put(1103,886){\makebox(0,0){$\scriptstyle\Diamond$}}
\put(1110,416){\makebox(0,0){$\scriptstyle\Diamond$}}
\put(1117,405){\makebox(0,0){$\scriptstyle\Diamond$}}
\put(1125,524){\makebox(0,0){$\scriptstyle\Diamond$}}
\put(1132,424){\makebox(0,0){$\scriptstyle\Diamond$}}
\put(1139,401){\makebox(0,0){$\scriptstyle\Diamond$}}
\put(1147,413){\makebox(0,0){$\scriptstyle\Diamond$}}
\put(1154,402){\makebox(0,0){$\scriptstyle\Diamond$}}
\put(1161,419){\makebox(0,0){$\scriptstyle\Diamond$}}
\put(1168,834){\makebox(0,0){$\scriptstyle\Diamond$}}
\put(1176,404){\makebox(0,0){$\scriptstyle\Diamond$}}
\put(1183,821){\makebox(0,0){$\scriptstyle\Diamond$}}
\put(1190,412){\makebox(0,0){$\scriptstyle\Diamond$}}
\put(1198,413){\makebox(0,0){$\scriptstyle\Diamond$}}
\put(1205,417){\makebox(0,0){$\scriptstyle\Diamond$}}
\put(1212,399){\makebox(0,0){$\scriptstyle\Diamond$}}
\put(1219,423){\makebox(0,0){$\scriptstyle\Diamond$}}
\put(1227,403){\makebox(0,0){$\scriptstyle\Diamond$}}
\put(1234,403){\makebox(0,0){$\scriptstyle\Diamond$}}
\put(1241,426){\makebox(0,0){$\scriptstyle\Diamond$}}
\put(1249,800){\makebox(0,0){$\scriptstyle\Diamond$}}
\put(1256,587){\makebox(0,0){$\scriptstyle\Diamond$}}
\put(1263,413){\makebox(0,0){$\scriptstyle\Diamond$}}
\put(1271,790){\makebox(0,0){$\scriptstyle\Diamond$}}
\put(1278,413){\makebox(0,0){$\scriptstyle\Diamond$}}
\put(1285,412){\makebox(0,0){$\scriptstyle\Diamond$}}
\put(1292,408){\makebox(0,0){$\scriptstyle\Diamond$}}
\put(1300,779){\makebox(0,0){$\scriptstyle\Diamond$}}
\put(1307,766){\makebox(0,0){$\scriptstyle\Diamond$}}
\put(1314,399){\makebox(0,0){$\scriptstyle\Diamond$}}
\put(1322,401){\makebox(0,0){$\scriptstyle\Diamond$}}
\put(1329,766){\makebox(0,0){$\scriptstyle\Diamond$}}
\put(1336,400){\makebox(0,0){$\scriptstyle\Diamond$}}
\put(1343,412){\makebox(0,0){$\scriptstyle\Diamond$}}
\put(1351,415){\makebox(0,0){$\scriptstyle\Diamond$}}
\put(1358,404){\makebox(0,0){$\scriptstyle\Diamond$}}
\put(1365,423){\makebox(0,0){$\scriptstyle\Diamond$}}
\put(1373,407){\makebox(0,0){$\scriptstyle\Diamond$}}
\put(1380,423){\makebox(0,0){$\scriptstyle\Diamond$}}
\put(1387,403){\makebox(0,0){$\scriptstyle\Diamond$}}
\put(1395,750){\makebox(0,0){$\scriptstyle\Diamond$}}
\put(1402,409){\makebox(0,0){$\scriptstyle\Diamond$}}
\put(1409,411){\makebox(0,0){$\scriptstyle\Diamond$}}
\put(1416,404){\makebox(0,0){$\scriptstyle\Diamond$}}
\put(1424,414){\makebox(0,0){$\scriptstyle\Diamond$}}
\put(1431,732){\makebox(0,0){$\scriptstyle\Diamond$}}
\put(1438,732){\makebox(0,0){$\scriptstyle\Diamond$}}
\put(1446,856){\makebox(0,0){$\scriptstyle\Diamond$}}
\put(1453,421){\makebox(0,0){$\scriptstyle\Diamond$}}
\put(1460,489){\makebox(0,0){$\scriptstyle\Diamond$}}
\put(1467,480){\makebox(0,0){$\scriptstyle\Diamond$}}
\put(1475,715){\makebox(0,0){$\scriptstyle\Diamond$}}
\put(1482,406){\makebox(0,0){$\scriptstyle\Diamond$}}
\put(1489,711){\makebox(0,0){$\scriptstyle\Diamond$}}
\put(1497,402){\makebox(0,0){$\scriptstyle\Diamond$}}
\put(1504,404){\makebox(0,0){$\scriptstyle\Diamond$}}
\put(1511,401){\makebox(0,0){$\scriptstyle\Diamond$}}
\put(1519,421){\makebox(0,0){$\scriptstyle\Diamond$}}
\put(1526,562){\makebox(0,0){$\scriptstyle\Diamond$}}
\put(1533,418){\makebox(0,0){$\scriptstyle\Diamond$}}
\put(1540,401){\makebox(0,0){$\scriptstyle\Diamond$}}
\put(1548,403){\makebox(0,0){$\scriptstyle\Diamond$}}
\put(1555,661){\makebox(0,0){$\scriptstyle\Diamond$}}
\put(1562,663){\makebox(0,0){$\scriptstyle\Diamond$}}
\put(1570,408){\makebox(0,0){$\scriptstyle\Diamond$}}
\put(1577,414){\makebox(0,0){$\scriptstyle\Diamond$}}
\put(1584,423){\makebox(0,0){$\scriptstyle\Diamond$}}
\put(1592,402){\makebox(0,0){$\scriptstyle\Diamond$}}
\put(1599,406){\makebox(0,0){$\scriptstyle\Diamond$}}
\put(1606,747){\makebox(0,0){$\scriptstyle\Diamond$}}
\put(1613,403){\makebox(0,0){$\scriptstyle\Diamond$}}
\put(1621,438){\makebox(0,0){$\scriptstyle\Diamond$}}
\put(1628,414){\makebox(0,0){$\scriptstyle\Diamond$}}
\put(1635,971){\makebox(0,0){$\scriptstyle\Diamond$}}
\put(1643,409){\makebox(0,0){$\scriptstyle\Diamond$}}
\put(1650,624){\makebox(0,0){$\scriptstyle\Diamond$}}
\put(1657,560){\makebox(0,0){$\scriptstyle\Diamond$}}
\put(1664,454){\makebox(0,0){$\scriptstyle\Diamond$}}
\put(1672,403){\makebox(0,0){$\scriptstyle\Diamond$}}
\put(1679,462){\makebox(0,0){$\scriptstyle\Diamond$}}
\put(1686,443){\makebox(0,0){$\scriptstyle\Diamond$}}
\put(1694,608){\makebox(0,0){$\scriptstyle\Diamond$}}
\put(1701,611){\makebox(0,0){$\scriptstyle\Diamond$}}
\put(1708,405){\makebox(0,0){$\scriptstyle\Diamond$}}
\put(1716,413){\makebox(0,0){$\scriptstyle\Diamond$}}
\put(1723,706){\makebox(0,0){$\scriptstyle\Diamond$}}
\put(1730,575){\makebox(0,0){$\scriptstyle\Diamond$}}
\put(1737,605){\makebox(0,0){$\scriptstyle\Diamond$}}
\put(1745,573){\makebox(0,0){$\scriptstyle\Diamond$}}
\put(1752,1082){\makebox(0,0){$\scriptstyle\Diamond$}}
\put(1759,570){\makebox(0,0){$\scriptstyle\Diamond$}}
\put(1767,405){\makebox(0,0){$\scriptstyle\Diamond$}}
\put(1774,399){\makebox(0,0){$\scriptstyle\Diamond$}}
\put(1781,401){\makebox(0,0){$\scriptstyle\Diamond$}}
\put(1788,539){\makebox(0,0){$\scriptstyle\Diamond$}}
\put(1796,408){\makebox(0,0){$\scriptstyle\Diamond$}}
\put(1803,403){\makebox(0,0){$\scriptstyle\Diamond$}}
\put(1810,404){\makebox(0,0){$\scriptstyle\Diamond$}}
\put(1818,395){\makebox(0,0){$\scriptstyle\Diamond$}}
\put(1825,507){\makebox(0,0){$\scriptstyle\Diamond$}}
\put(1832,408){\makebox(0,0){$\scriptstyle\Diamond$}}
\put(1840,398){\makebox(0,0){$\scriptstyle\Diamond$}}
\put(1847,400){\makebox(0,0){$\scriptstyle\Diamond$}}
\put(1854,398){\makebox(0,0){$\scriptstyle\Diamond$}}
\put(1861,497){\makebox(0,0){$\scriptstyle\Diamond$}}
\put(1869,471){\makebox(0,0){$\scriptstyle\Diamond$}}
\put(1876,398){\makebox(0,0){$\scriptstyle\Diamond$}}
\put(1841,1320){\makebox(0,0){$\scriptstyle\Diamond$}}
\color{blue}
\color{black}
\put(1734,1262){\makebox(0,0)[r]{FS ($m = 100$)}}
\color{blue}
\put(512,299){\makebox(0,0){$\scriptstyle +$}}
\put(519,296){\makebox(0,0){$\scriptstyle +$}}
\put(526,284){\makebox(0,0){$\scriptstyle +$}}
\put(534,285){\makebox(0,0){$\scriptstyle +$}}
\put(541,283){\makebox(0,0){$\scriptstyle +$}}
\put(548,288){\makebox(0,0){$\scriptstyle +$}}
\put(556,344){\makebox(0,0){$\scriptstyle +$}}
\put(563,369){\makebox(0,0){$\scriptstyle +$}}
\put(570,287){\makebox(0,0){$\scriptstyle +$}}
\put(578,333){\makebox(0,0){$\scriptstyle +$}}
\put(585,308){\makebox(0,0){$\scriptstyle +$}}
\put(592,335){\makebox(0,0){$\scriptstyle +$}}
\put(599,299){\makebox(0,0){$\scriptstyle +$}}
\put(607,321){\makebox(0,0){$\scriptstyle +$}}
\put(614,334){\makebox(0,0){$\scriptstyle +$}}
\put(621,311){\makebox(0,0){$\scriptstyle +$}}
\put(629,321){\makebox(0,0){$\scriptstyle +$}}
\put(636,322){\makebox(0,0){$\scriptstyle +$}}
\put(643,348){\makebox(0,0){$\scriptstyle +$}}
\put(650,335){\makebox(0,0){$\scriptstyle +$}}
\put(658,355){\makebox(0,0){$\scriptstyle +$}}
\put(665,359){\makebox(0,0){$\scriptstyle +$}}
\put(672,299){\makebox(0,0){$\scriptstyle +$}}
\put(680,301){\makebox(0,0){$\scriptstyle +$}}
\put(687,309){\makebox(0,0){$\scriptstyle +$}}
\put(694,311){\makebox(0,0){$\scriptstyle +$}}
\put(702,328){\makebox(0,0){$\scriptstyle +$}}
\put(709,307){\makebox(0,0){$\scriptstyle +$}}
\put(716,293){\makebox(0,0){$\scriptstyle +$}}
\put(723,359){\makebox(0,0){$\scriptstyle +$}}
\put(731,294){\makebox(0,0){$\scriptstyle +$}}
\put(738,333){\makebox(0,0){$\scriptstyle +$}}
\put(745,326){\makebox(0,0){$\scriptstyle +$}}
\put(753,311){\makebox(0,0){$\scriptstyle +$}}
\put(760,310){\makebox(0,0){$\scriptstyle +$}}
\put(767,293){\makebox(0,0){$\scriptstyle +$}}
\put(774,318){\makebox(0,0){$\scriptstyle +$}}
\put(782,356){\makebox(0,0){$\scriptstyle +$}}
\put(789,328){\makebox(0,0){$\scriptstyle +$}}
\put(796,353){\makebox(0,0){$\scriptstyle +$}}
\put(804,310){\makebox(0,0){$\scriptstyle +$}}
\put(811,321){\makebox(0,0){$\scriptstyle +$}}
\put(818,333){\makebox(0,0){$\scriptstyle +$}}
\put(826,301){\makebox(0,0){$\scriptstyle +$}}
\put(833,367){\makebox(0,0){$\scriptstyle +$}}
\put(840,316){\makebox(0,0){$\scriptstyle +$}}
\put(847,323){\makebox(0,0){$\scriptstyle +$}}
\put(855,343){\makebox(0,0){$\scriptstyle +$}}
\put(862,326){\makebox(0,0){$\scriptstyle +$}}
\put(869,355){\makebox(0,0){$\scriptstyle +$}}
\put(877,323){\makebox(0,0){$\scriptstyle +$}}
\put(884,300){\makebox(0,0){$\scriptstyle +$}}
\put(891,318){\makebox(0,0){$\scriptstyle +$}}
\put(899,325){\makebox(0,0){$\scriptstyle +$}}
\put(906,351){\makebox(0,0){$\scriptstyle +$}}
\put(913,324){\makebox(0,0){$\scriptstyle +$}}
\put(920,316){\makebox(0,0){$\scriptstyle +$}}
\put(928,313){\makebox(0,0){$\scriptstyle +$}}
\put(935,318){\makebox(0,0){$\scriptstyle +$}}
\put(942,310){\makebox(0,0){$\scriptstyle +$}}
\put(950,371){\makebox(0,0){$\scriptstyle +$}}
\put(957,335){\makebox(0,0){$\scriptstyle +$}}
\put(964,312){\makebox(0,0){$\scriptstyle +$}}
\put(971,308){\makebox(0,0){$\scriptstyle +$}}
\put(979,324){\makebox(0,0){$\scriptstyle +$}}
\put(986,365){\makebox(0,0){$\scriptstyle +$}}
\put(993,342){\makebox(0,0){$\scriptstyle +$}}
\put(1001,305){\makebox(0,0){$\scriptstyle +$}}
\put(1008,346){\makebox(0,0){$\scriptstyle +$}}
\put(1015,328){\makebox(0,0){$\scriptstyle +$}}
\put(1023,313){\makebox(0,0){$\scriptstyle +$}}
\put(1030,335){\makebox(0,0){$\scriptstyle +$}}
\put(1037,357){\makebox(0,0){$\scriptstyle +$}}
\put(1044,343){\makebox(0,0){$\scriptstyle +$}}
\put(1052,337){\makebox(0,0){$\scriptstyle +$}}
\put(1059,350){\makebox(0,0){$\scriptstyle +$}}
\put(1066,348){\makebox(0,0){$\scriptstyle +$}}
\put(1074,331){\makebox(0,0){$\scriptstyle +$}}
\put(1081,370){\makebox(0,0){$\scriptstyle +$}}
\put(1088,344){\makebox(0,0){$\scriptstyle +$}}
\put(1095,307){\makebox(0,0){$\scriptstyle +$}}
\put(1103,351){\makebox(0,0){$\scriptstyle +$}}
\put(1110,329){\makebox(0,0){$\scriptstyle +$}}
\put(1117,305){\makebox(0,0){$\scriptstyle +$}}
\put(1125,333){\makebox(0,0){$\scriptstyle +$}}
\put(1132,344){\makebox(0,0){$\scriptstyle +$}}
\put(1139,295){\makebox(0,0){$\scriptstyle +$}}
\put(1147,324){\makebox(0,0){$\scriptstyle +$}}
\put(1154,297){\makebox(0,0){$\scriptstyle +$}}
\put(1161,336){\makebox(0,0){$\scriptstyle +$}}
\put(1168,354){\makebox(0,0){$\scriptstyle +$}}
\put(1176,306){\makebox(0,0){$\scriptstyle +$}}
\put(1183,305){\makebox(0,0){$\scriptstyle +$}}
\put(1190,315){\makebox(0,0){$\scriptstyle +$}}
\put(1198,321){\makebox(0,0){$\scriptstyle +$}}
\put(1205,329){\makebox(0,0){$\scriptstyle +$}}
\put(1212,297){\makebox(0,0){$\scriptstyle +$}}
\put(1219,337){\makebox(0,0){$\scriptstyle +$}}
\put(1227,298){\makebox(0,0){$\scriptstyle +$}}
\put(1234,301){\makebox(0,0){$\scriptstyle +$}}
\put(1241,344){\makebox(0,0){$\scriptstyle +$}}
\put(1249,331){\makebox(0,0){$\scriptstyle +$}}
\put(1256,338){\makebox(0,0){$\scriptstyle +$}}
\put(1263,320){\makebox(0,0){$\scriptstyle +$}}
\put(1271,325){\makebox(0,0){$\scriptstyle +$}}
\put(1278,325){\makebox(0,0){$\scriptstyle +$}}
\put(1285,318){\makebox(0,0){$\scriptstyle +$}}
\put(1292,312){\makebox(0,0){$\scriptstyle +$}}
\put(1300,345){\makebox(0,0){$\scriptstyle +$}}
\put(1307,305){\makebox(0,0){$\scriptstyle +$}}
\put(1314,292){\makebox(0,0){$\scriptstyle +$}}
\put(1322,298){\makebox(0,0){$\scriptstyle +$}}
\put(1329,325){\makebox(0,0){$\scriptstyle +$}}
\put(1336,302){\makebox(0,0){$\scriptstyle +$}}
\put(1343,319){\makebox(0,0){$\scriptstyle +$}}
\put(1351,323){\makebox(0,0){$\scriptstyle +$}}
\put(1358,304){\makebox(0,0){$\scriptstyle +$}}
\put(1365,341){\makebox(0,0){$\scriptstyle +$}}
\put(1373,313){\makebox(0,0){$\scriptstyle +$}}
\put(1380,328){\makebox(0,0){$\scriptstyle +$}}
\put(1387,301){\makebox(0,0){$\scriptstyle +$}}
\put(1395,328){\makebox(0,0){$\scriptstyle +$}}
\put(1402,320){\makebox(0,0){$\scriptstyle +$}}
\put(1409,318){\makebox(0,0){$\scriptstyle +$}}
\put(1416,311){\makebox(0,0){$\scriptstyle +$}}
\put(1424,321){\makebox(0,0){$\scriptstyle +$}}
\put(1431,300){\makebox(0,0){$\scriptstyle +$}}
\put(1438,306){\makebox(0,0){$\scriptstyle +$}}
\put(1446,319){\makebox(0,0){$\scriptstyle +$}}
\put(1453,330){\makebox(0,0){$\scriptstyle +$}}
\put(1460,328){\makebox(0,0){$\scriptstyle +$}}
\put(1467,322){\makebox(0,0){$\scriptstyle +$}}
\put(1475,316){\makebox(0,0){$\scriptstyle +$}}
\put(1482,310){\makebox(0,0){$\scriptstyle +$}}
\put(1489,317){\makebox(0,0){$\scriptstyle +$}}
\put(1497,298){\makebox(0,0){$\scriptstyle +$}}
\put(1504,302){\makebox(0,0){$\scriptstyle +$}}
\put(1511,298){\makebox(0,0){$\scriptstyle +$}}
\put(1519,341){\makebox(0,0){$\scriptstyle +$}}
\put(1526,316){\makebox(0,0){$\scriptstyle +$}}
\put(1533,333){\makebox(0,0){$\scriptstyle +$}}
\put(1540,298){\makebox(0,0){$\scriptstyle +$}}
\put(1548,302){\makebox(0,0){$\scriptstyle +$}}
\put(1555,300){\makebox(0,0){$\scriptstyle +$}}
\put(1562,319){\makebox(0,0){$\scriptstyle +$}}
\put(1570,313){\makebox(0,0){$\scriptstyle +$}}
\put(1577,324){\makebox(0,0){$\scriptstyle +$}}
\put(1584,320){\makebox(0,0){$\scriptstyle +$}}
\put(1592,301){\makebox(0,0){$\scriptstyle +$}}
\put(1599,309){\makebox(0,0){$\scriptstyle +$}}
\put(1606,307){\makebox(0,0){$\scriptstyle +$}}
\put(1613,299){\makebox(0,0){$\scriptstyle +$}}
\put(1621,356){\makebox(0,0){$\scriptstyle +$}}
\put(1628,328){\makebox(0,0){$\scriptstyle +$}}
\put(1635,306){\makebox(0,0){$\scriptstyle +$}}
\put(1643,319){\makebox(0,0){$\scriptstyle +$}}
\put(1650,311){\makebox(0,0){$\scriptstyle +$}}
\put(1657,310){\makebox(0,0){$\scriptstyle +$}}
\put(1664,295){\makebox(0,0){$\scriptstyle +$}}
\put(1672,300){\makebox(0,0){$\scriptstyle +$}}
\put(1679,310){\makebox(0,0){$\scriptstyle +$}}
\put(1686,322){\makebox(0,0){$\scriptstyle +$}}
\put(1694,290){\makebox(0,0){$\scriptstyle +$}}
\put(1701,315){\makebox(0,0){$\scriptstyle +$}}
\put(1708,308){\makebox(0,0){$\scriptstyle +$}}
\put(1716,328){\makebox(0,0){$\scriptstyle +$}}
\put(1723,308){\makebox(0,0){$\scriptstyle +$}}
\put(1730,289){\makebox(0,0){$\scriptstyle +$}}
\put(1737,320){\makebox(0,0){$\scriptstyle +$}}
\put(1745,306){\makebox(0,0){$\scriptstyle +$}}
\put(1752,346){\makebox(0,0){$\scriptstyle +$}}
\put(1759,296){\makebox(0,0){$\scriptstyle +$}}
\put(1767,310){\makebox(0,0){$\scriptstyle +$}}
\put(1774,293){\makebox(0,0){$\scriptstyle +$}}
\put(1781,301){\makebox(0,0){$\scriptstyle +$}}
\put(1788,304){\makebox(0,0){$\scriptstyle +$}}
\put(1796,314){\makebox(0,0){$\scriptstyle +$}}
\put(1803,306){\makebox(0,0){$\scriptstyle +$}}
\put(1810,305){\makebox(0,0){$\scriptstyle +$}}
\put(1818,284){\makebox(0,0){$\scriptstyle +$}}
\put(1825,303){\makebox(0,0){$\scriptstyle +$}}
\put(1832,305){\makebox(0,0){$\scriptstyle +$}}
\put(1840,290){\makebox(0,0){$\scriptstyle +$}}
\put(1847,294){\makebox(0,0){$\scriptstyle +$}}
\put(1854,293){\makebox(0,0){$\scriptstyle +$}}
\put(1861,299){\makebox(0,0){$\scriptstyle +$}}
\put(1869,297){\makebox(0,0){$\scriptstyle +$}}
\put(1876,292){\makebox(0,0){$\scriptstyle +$}}
\put(1841,1262){\makebox(0,0){$\scriptstyle +$}}
\color{green}
\color{black}
\put(1734,1204){\makebox(0,0)[r]{MultipleRW ($K=100$)}}
\color{green}
\put(512,567){\makebox(0,0){$\scriptstyle\Box$}}
\put(519,573){\makebox(0,0){$\scriptstyle\Box$}}
\put(526,547){\makebox(0,0){$\scriptstyle\Box$}}
\put(534,549){\makebox(0,0){$\scriptstyle\Box$}}
\put(541,551){\makebox(0,0){$\scriptstyle\Box$}}
\put(548,549){\makebox(0,0){$\scriptstyle\Box$}}
\put(556,665){\makebox(0,0){$\scriptstyle\Box$}}
\put(563,752){\makebox(0,0){$\scriptstyle\Box$}}
\put(570,551){\makebox(0,0){$\scriptstyle\Box$}}
\put(578,685){\makebox(0,0){$\scriptstyle\Box$}}
\put(585,542){\makebox(0,0){$\scriptstyle\Box$}}
\put(592,733){\makebox(0,0){$\scriptstyle\Box$}}
\put(599,627){\makebox(0,0){$\scriptstyle\Box$}}
\put(607,572){\makebox(0,0){$\scriptstyle\Box$}}
\put(614,560){\makebox(0,0){$\scriptstyle\Box$}}
\put(621,579){\makebox(0,0){$\scriptstyle\Box$}}
\put(629,642){\makebox(0,0){$\scriptstyle\Box$}}
\put(636,555){\makebox(0,0){$\scriptstyle\Box$}}
\put(643,590){\makebox(0,0){$\scriptstyle\Box$}}
\put(650,554){\makebox(0,0){$\scriptstyle\Box$}}
\put(658,688){\makebox(0,0){$\scriptstyle\Box$}}
\put(665,835){\makebox(0,0){$\scriptstyle\Box$}}
\put(672,671){\makebox(0,0){$\scriptstyle\Box$}}
\put(680,569){\makebox(0,0){$\scriptstyle\Box$}}
\put(687,547){\makebox(0,0){$\scriptstyle\Box$}}
\put(694,542){\makebox(0,0){$\scriptstyle\Box$}}
\put(702,581){\makebox(0,0){$\scriptstyle\Box$}}
\put(709,584){\makebox(0,0){$\scriptstyle\Box$}}
\put(716,542){\makebox(0,0){$\scriptstyle\Box$}}
\put(723,552){\makebox(0,0){$\scriptstyle\Box$}}
\put(731,541){\makebox(0,0){$\scriptstyle\Box$}}
\put(738,731){\makebox(0,0){$\scriptstyle\Box$}}
\put(745,604){\makebox(0,0){$\scriptstyle\Box$}}
\put(753,622){\makebox(0,0){$\scriptstyle\Box$}}
\put(760,543){\makebox(0,0){$\scriptstyle\Box$}}
\put(767,544){\makebox(0,0){$\scriptstyle\Box$}}
\put(774,581){\makebox(0,0){$\scriptstyle\Box$}}
\put(782,676){\makebox(0,0){$\scriptstyle\Box$}}
\put(789,645){\makebox(0,0){$\scriptstyle\Box$}}
\put(796,681){\makebox(0,0){$\scriptstyle\Box$}}
\put(804,616){\makebox(0,0){$\scriptstyle\Box$}}
\put(811,600){\makebox(0,0){$\scriptstyle\Box$}}
\put(818,551){\makebox(0,0){$\scriptstyle\Box$}}
\put(826,591){\makebox(0,0){$\scriptstyle\Box$}}
\put(833,668){\makebox(0,0){$\scriptstyle\Box$}}
\put(840,601){\makebox(0,0){$\scriptstyle\Box$}}
\put(847,648){\makebox(0,0){$\scriptstyle\Box$}}
\put(855,680){\makebox(0,0){$\scriptstyle\Box$}}
\put(862,599){\makebox(0,0){$\scriptstyle\Box$}}
\put(869,922){\makebox(0,0){$\scriptstyle\Box$}}
\put(877,544){\makebox(0,0){$\scriptstyle\Box$}}
\put(884,588){\makebox(0,0){$\scriptstyle\Box$}}
\put(891,599){\makebox(0,0){$\scriptstyle\Box$}}
\put(899,579){\makebox(0,0){$\scriptstyle\Box$}}
\put(906,711){\makebox(0,0){$\scriptstyle\Box$}}
\put(913,619){\makebox(0,0){$\scriptstyle\Box$}}
\put(920,622){\makebox(0,0){$\scriptstyle\Box$}}
\put(928,588){\makebox(0,0){$\scriptstyle\Box$}}
\put(935,564){\makebox(0,0){$\scriptstyle\Box$}}
\put(942,552){\makebox(0,0){$\scriptstyle\Box$}}
\put(950,713){\makebox(0,0){$\scriptstyle\Box$}}
\put(957,640){\makebox(0,0){$\scriptstyle\Box$}}
\put(964,616){\makebox(0,0){$\scriptstyle\Box$}}
\put(971,588){\makebox(0,0){$\scriptstyle\Box$}}
\put(979,592){\makebox(0,0){$\scriptstyle\Box$}}
\put(986,750){\makebox(0,0){$\scriptstyle\Box$}}
\put(993,613){\makebox(0,0){$\scriptstyle\Box$}}
\put(1001,603){\makebox(0,0){$\scriptstyle\Box$}}
\put(1008,589){\makebox(0,0){$\scriptstyle\Box$}}
\put(1015,580){\makebox(0,0){$\scriptstyle\Box$}}
\put(1023,557){\makebox(0,0){$\scriptstyle\Box$}}
\put(1030,697){\makebox(0,0){$\scriptstyle\Box$}}
\put(1037,584){\makebox(0,0){$\scriptstyle\Box$}}
\put(1044,665){\makebox(0,0){$\scriptstyle\Box$}}
\put(1052,584){\makebox(0,0){$\scriptstyle\Box$}}
\put(1059,787){\makebox(0,0){$\scriptstyle\Box$}}
\put(1066,643){\makebox(0,0){$\scriptstyle\Box$}}
\put(1074,598){\makebox(0,0){$\scriptstyle\Box$}}
\put(1081,702){\makebox(0,0){$\scriptstyle\Box$}}
\put(1088,659){\makebox(0,0){$\scriptstyle\Box$}}
\put(1095,577){\makebox(0,0){$\scriptstyle\Box$}}
\put(1103,625){\makebox(0,0){$\scriptstyle\Box$}}
\put(1110,652){\makebox(0,0){$\scriptstyle\Box$}}
\put(1117,595){\makebox(0,0){$\scriptstyle\Box$}}
\put(1125,614){\makebox(0,0){$\scriptstyle\Box$}}
\put(1132,615){\makebox(0,0){$\scriptstyle\Box$}}
\put(1139,541){\makebox(0,0){$\scriptstyle\Box$}}
\put(1147,623){\makebox(0,0){$\scriptstyle\Box$}}
\put(1154,540){\makebox(0,0){$\scriptstyle\Box$}}
\put(1161,633){\makebox(0,0){$\scriptstyle\Box$}}
\put(1168,776){\makebox(0,0){$\scriptstyle\Box$}}
\put(1176,591){\makebox(0,0){$\scriptstyle\Box$}}
\put(1183,556){\makebox(0,0){$\scriptstyle\Box$}}
\put(1190,614){\makebox(0,0){$\scriptstyle\Box$}}
\put(1198,581){\makebox(0,0){$\scriptstyle\Box$}}
\put(1205,684){\makebox(0,0){$\scriptstyle\Box$}}
\put(1212,602){\makebox(0,0){$\scriptstyle\Box$}}
\put(1219,557){\makebox(0,0){$\scriptstyle\Box$}}
\put(1227,543){\makebox(0,0){$\scriptstyle\Box$}}
\put(1234,618){\makebox(0,0){$\scriptstyle\Box$}}
\put(1241,679){\makebox(0,0){$\scriptstyle\Box$}}
\put(1249,572){\makebox(0,0){$\scriptstyle\Box$}}
\put(1256,625){\makebox(0,0){$\scriptstyle\Box$}}
\put(1263,630){\makebox(0,0){$\scriptstyle\Box$}}
\put(1271,607){\makebox(0,0){$\scriptstyle\Box$}}
\put(1278,645){\makebox(0,0){$\scriptstyle\Box$}}
\put(1285,616){\makebox(0,0){$\scriptstyle\Box$}}
\put(1292,612){\makebox(0,0){$\scriptstyle\Box$}}
\put(1300,688){\makebox(0,0){$\scriptstyle\Box$}}
\put(1307,587){\makebox(0,0){$\scriptstyle\Box$}}
\put(1314,544){\makebox(0,0){$\scriptstyle\Box$}}
\put(1322,557){\makebox(0,0){$\scriptstyle\Box$}}
\put(1329,622){\makebox(0,0){$\scriptstyle\Box$}}
\put(1336,559){\makebox(0,0){$\scriptstyle\Box$}}
\put(1343,581){\makebox(0,0){$\scriptstyle\Box$}}
\put(1351,604){\makebox(0,0){$\scriptstyle\Box$}}
\put(1358,594){\makebox(0,0){$\scriptstyle\Box$}}
\put(1365,617){\makebox(0,0){$\scriptstyle\Box$}}
\put(1373,631){\makebox(0,0){$\scriptstyle\Box$}}
\put(1380,586){\makebox(0,0){$\scriptstyle\Box$}}
\put(1387,590){\makebox(0,0){$\scriptstyle\Box$}}
\put(1395,629){\makebox(0,0){$\scriptstyle\Box$}}
\put(1402,670){\makebox(0,0){$\scriptstyle\Box$}}
\put(1409,669){\makebox(0,0){$\scriptstyle\Box$}}
\put(1416,633){\makebox(0,0){$\scriptstyle\Box$}}
\put(1424,657){\makebox(0,0){$\scriptstyle\Box$}}
\put(1431,581){\makebox(0,0){$\scriptstyle\Box$}}
\put(1438,607){\makebox(0,0){$\scriptstyle\Box$}}
\put(1446,598){\makebox(0,0){$\scriptstyle\Box$}}
\put(1453,637){\makebox(0,0){$\scriptstyle\Box$}}
\put(1460,625){\makebox(0,0){$\scriptstyle\Box$}}
\put(1467,619){\makebox(0,0){$\scriptstyle\Box$}}
\put(1475,650){\makebox(0,0){$\scriptstyle\Box$}}
\put(1482,591){\makebox(0,0){$\scriptstyle\Box$}}
\put(1489,618){\makebox(0,0){$\scriptstyle\Box$}}
\put(1497,557){\makebox(0,0){$\scriptstyle\Box$}}
\put(1504,629){\makebox(0,0){$\scriptstyle\Box$}}
\put(1511,586){\makebox(0,0){$\scriptstyle\Box$}}
\put(1519,753){\makebox(0,0){$\scriptstyle\Box$}}
\put(1526,599){\makebox(0,0){$\scriptstyle\Box$}}
\put(1533,617){\makebox(0,0){$\scriptstyle\Box$}}
\put(1540,543){\makebox(0,0){$\scriptstyle\Box$}}
\put(1548,615){\makebox(0,0){$\scriptstyle\Box$}}
\put(1555,553){\makebox(0,0){$\scriptstyle\Box$}}
\put(1562,636){\makebox(0,0){$\scriptstyle\Box$}}
\put(1570,611){\makebox(0,0){$\scriptstyle\Box$}}
\put(1577,603){\makebox(0,0){$\scriptstyle\Box$}}
\put(1584,652){\makebox(0,0){$\scriptstyle\Box$}}
\put(1592,582){\makebox(0,0){$\scriptstyle\Box$}}
\put(1599,604){\makebox(0,0){$\scriptstyle\Box$}}
\put(1606,622){\makebox(0,0){$\scriptstyle\Box$}}
\put(1613,552){\makebox(0,0){$\scriptstyle\Box$}}
\put(1621,563){\makebox(0,0){$\scriptstyle\Box$}}
\put(1628,646){\makebox(0,0){$\scriptstyle\Box$}}
\put(1635,593){\makebox(0,0){$\scriptstyle\Box$}}
\put(1643,597){\makebox(0,0){$\scriptstyle\Box$}}
\put(1650,648){\makebox(0,0){$\scriptstyle\Box$}}
\put(1657,556){\makebox(0,0){$\scriptstyle\Box$}}
\put(1664,583){\makebox(0,0){$\scriptstyle\Box$}}
\put(1672,559){\makebox(0,0){$\scriptstyle\Box$}}
\put(1679,576){\makebox(0,0){$\scriptstyle\Box$}}
\put(1686,661){\makebox(0,0){$\scriptstyle\Box$}}
\put(1694,572){\makebox(0,0){$\scriptstyle\Box$}}
\put(1701,600){\makebox(0,0){$\scriptstyle\Box$}}
\put(1708,604){\makebox(0,0){$\scriptstyle\Box$}}
\put(1716,637){\makebox(0,0){$\scriptstyle\Box$}}
\put(1723,577){\makebox(0,0){$\scriptstyle\Box$}}
\put(1730,583){\makebox(0,0){$\scriptstyle\Box$}}
\put(1737,607){\makebox(0,0){$\scriptstyle\Box$}}
\put(1745,604){\makebox(0,0){$\scriptstyle\Box$}}
\put(1752,680){\makebox(0,0){$\scriptstyle\Box$}}
\put(1759,561){\makebox(0,0){$\scriptstyle\Box$}}
\put(1767,609){\makebox(0,0){$\scriptstyle\Box$}}
\put(1774,545){\makebox(0,0){$\scriptstyle\Box$}}
\put(1781,564){\makebox(0,0){$\scriptstyle\Box$}}
\put(1788,596){\makebox(0,0){$\scriptstyle\Box$}}
\put(1796,576){\makebox(0,0){$\scriptstyle\Box$}}
\put(1803,584){\makebox(0,0){$\scriptstyle\Box$}}
\put(1810,569){\makebox(0,0){$\scriptstyle\Box$}}
\put(1818,549){\makebox(0,0){$\scriptstyle\Box$}}
\put(1825,580){\makebox(0,0){$\scriptstyle\Box$}}
\put(1832,567){\makebox(0,0){$\scriptstyle\Box$}}
\put(1840,549){\makebox(0,0){$\scriptstyle\Box$}}
\put(1847,554){\makebox(0,0){$\scriptstyle\Box$}}
\put(1854,570){\makebox(0,0){$\scriptstyle\Box$}}
\put(1861,587){\makebox(0,0){$\scriptstyle\Box$}}
\put(1869,570){\makebox(0,0){$\scriptstyle\Box$}}
\put(1876,563){\makebox(0,0){$\scriptstyle\Box$}}
\put(1841,1204){\makebox(0,0){$\scriptstyle\Box$}}
\color{black}
\thicklines \path(490,1202)(490,226)(1876,226)
\color{black}
\end{picture}

%% file: appendix.tex
\appendix

\balance
\section{Proof of the Frontier Sampling Theorem}\label{appx:proofFST}
 In what follows we restate Theorem~\ref{thm:FST} and present a proof.
 \begin{thm*}
 Recall that $G$ is a directed symmetric graph. 
 If $G$ is connected and non-bipartite, then in steady state FS has the following properties: 
 \renewcommand{\labelenumi}{(\Roman{enumi})}
 \begin{enumerate}
 \item edges are sampled uniformly at random and form a stationary sequence, 
 \item has distribution, $L_\infty=(v_1,\dots,v_m)$, equal to $$\frac{\sum_{i = 1}^m \deg(v_i)}{ m \vert V \vert^{m-1} \vol(V) } \, ,\text{ ,}$$ which is unique, and
 \item the sequence of sampled edges satisfies the Strong Law of Large Numbers (Theorem~\ref{thm:SLLN}).
 \end{enumerate}
 \end{thm*}
 
\begin{proof}
Consider the $(n-1)$-st step of Frontier sampling. 
The reader may find Figure~\ref{fig:exFS} helpful in following the proof.
Let $L_n=(v_1,\dots,v_m)$ be the state of Frontier sampling before the $n$-th step.
Clearly $L_n \in V^m$.
In what follows let $e(L_n)$ denote the collection of all edges associated to the vertices in $L_n$. 
We refer to $e(L_n)$ as the edge frontier at the $n$-th step.
We describe the transition from state $L_n$ to state $L_{n+1}$ as follows
(lines~\ref{FSloop} and~\ref{FSselect} of the frontier sampling algorithm):
Select a vertex $v \in L_n$ with probability proportional to $\deg(v)$ and 
then replace element $v$ in $L_n$ with one of its neighbors (selected uniformly at random).
This is equivalent to randomly sampling an edge from $e(L_n)$ with probability 
\[
 p = \frac{1}{\vert e(L_n) \vert} = \frac{1}{\sum_{\forall v \in L_n} \deg(v)}.
\]
Thus, $L_n$ is able to transition to state $L_{n+1}$ iff $(L_n,L_{n+1}) \in E_m$ and the transition probability from $L_n$ to $L_{n+1}$ is $1/\vert e(L_n) \vert$.
Thus, we conclude that Frontier sampling is a single random walker over the $m$-th Cartesian power of $G$, $G^m = (V^m,E_m)$, where
\[
V^m= \{(v_1,\dots,v_m) \given v_1 \in V \wedge \dots \wedge v_m \in V\}
\]
is the $m$-ary Cartesian product of $V$ and $\forall {\bf v},{\bf u}\in V^m$, $({\bf v},{\bf u}) \in E_m\,$ if exists an index $i$ such that $(v_i,u_i) \in E$ and $u_j = v_j$ for $j \neq i$.
 Note that $\vert E_m \vert = m \vert V \vert^{m-1} \vert E \vert$.

 Now we need to prove that the distribution of $L_\infty$ is stable and unique.
 For this we only need to show that the random walk over $G^m$ is ergodic. A random walk (Markov chain) is ergodic when it is aperiodic and recurrent non-null.
 Recall that the random walk over $G$ is ergodic.
 The probability that Frontier sampling transitions from $L_n \in V^m$ to $L_{n+1} \in V^m$ such that $L_n$ and $L_{n+1}$ only differ in their $i$-th element is always greater than zero, otherwise there is an infinite increasing degree sequence in the vertices of $G$.
 But this is not possible as the random walk over $G$ is recurrent non-null (an infinite increasing degree sequence would be a sink in the random walk over $G$).
 Thus, any finite sequence of transitions $\{L_{n+w}\}_{w=1}^{\Delta}$, $\Delta > 1$, that only updates its $i$-th element has probability greater than zero.
 Thus, as the sequence $\{L_{n+w}\}_{w=1}^{\Delta}$ is also a single random walk over $G$, it is aperiodic for any chosen $i=1,\dots,m$ and then a random walker over $G^m$ must also be aperiodic.
 We can use the same argument to show that the random walk over $G^m$ is recurrent non-null.
 As random walk over $G^m$ is ergodic, we have that $L^\star$ is distributed according to the steady state distribution of a random walk over $G^m$.
 Let $L_\infty \equiv \lim_{n \to \infty} L_n $ denote the state (vertex in $G^m$) when FS is in steady state.
 The number of edges out of vertex $L_\infty \in V^m$ is $\sum_{i=1}^m \deg(v_i)$.
 The number of edges in $G^m$ is $m \vert V \vert^{m-1} \vol(V)$ as each vertex $v \in V$ appears $m \vert V \vert^{m-1}$ times in $G^m$. Thus,
 \[
 P[ L_\infty = (v_1,\dots,v_m) ] = \frac{\sum_{i=1}^m \deg(v_i)}{ m \vert V \vert^{m-1} \vol(V)} \, ,
 \]
 which is unique and stable (similar to a single random walker as seen in Section~\ref{sec:RW}).
 
 The rest of the proof is straightforward.
 Each edge in $G^m$ is actually an edge in $G$.
 As each edge in $G$ is copied $m \vert V \vert^{m-1}$ times into $G^m$, we have that edges in $G$ are also sampled uniformly at random in a random walk over $G^m$.
 As Frontier sampling is a random walk over $G^m$, its samples form a stationary sequence and follow the Strong Law of Large Numbers seen in Theorem~\ref{thm:SLLN}.
 The same is true for the sequence of sampled vertices.
\end{proof}

\begin{table}
\begin{center}
\begin{tabular}{@{}llrrr@{}} \toprule
& & \multicolumn{3}{c}{Sampling prob. error}  \\ \cmidrule(r){3-5}
Graph & $B$ ({\small sampling budget})  & FS  & MRW & SRW ~ \\ \midrule
Internet RLT  & 100 & $17\%$ & $257\%$ & $156\%$ \\  
YouTube & 20 & $43\%$ & $236\%$ & $216\%$  \\
Hep-Th & 20 & $36\%$ & $1510\%$ & $781\%$  \\ 
\bottomrule
\end{tabular}
\caption{Relative worst-case difference between the steady state and the transient edge sampling probabilities after $B-K$ steps. Frontier edge sampling probabilities are closer to steady state in all graphs.  Legend: (FS) = Frontier sampling ($K=10$), (SRW) = Single ($K=1$) Random Walker, and (MRW) = Multiple ($K=10$) Random Walkers.\label{tab:mixvals}}
 \vspace{-19pt}
\end{center}
\end{table}

\section{Convergence to Uniform Edge Sampling}\label{appx:convergence}
 A question we seek to answer with our simulations is how fast these random walk methods converge to their stationary edge sampling probabilities.
 In this simulation we set $K \in \{1, 10\}$ (number of independent random walkers), $m=10$ (Frontier sampling dimension) and restrict our analysis to the largest connected component of the three graphs in our datasets with the smallest number of vertices (in order to speed the computation): ``Internet RLT'', ``YouTube'', and ``Hep-th''. 
 Let $p_{u,v}^{(B)}$ denote the probability that a random walker, whose initial vertex is chosen uniformly at random, samples edge $(u,v)$ at its the end of its sampling budget $B$. 
 To measure the convergence to the stationary edge sampling probability, we use the largest relative difference between the stationary sampling probability $1/\vert E \vert$ and $p_{u,v}^{(B)}$:
\[
  \max_{(u,v) \in E} 1- \frac{p_{u,v}^{(B)}}{1/\vert E \vert} \, .
\]
 Table~\ref{tab:mixvals} presents a Monte Carlo estimate of this relative difference.
 The $95\%$ confidence interval of the Monte Carlo simulation is $\pm 1\%$. 
 Our estimates show that the difference between the transient and the stationary edge sampling probabilities of independent random walkers are between 5 and 42 times larger than the difference of Frontier sampling.
 This means that Frontier sampling converges faster to stationarity edge sampling probability.